\documentclass[times,sort&compress,3p]{elsarticle}
\journal{Journal of Multivariate Analysis}
\usepackage[labelfont=bf]{caption}

\usepackage{amsmath,amsfonts,amssymb,amsthm,booktabs,color,epsfig,graphicx,hyperref,url,enumerate}

\newtheorem{thm}{Theorem}
\newtheorem{lemma}{Lemma}

\begin{document}

\begin{frontmatter}

\title{Family of mean-mixtures of multivariate normal distributions: properties, inference and assessment of  multivariate skewness}

\author[a]{Me$'$raj Abdi}
\author[a]{Mohsen Madadi\corref{mycorrespondingauthor}}
\author[b]{Narayanaswamy Balakrishnan}
\author[a]{Ahad Jamalizadeh}

\address[a]{Department of Statistics, Faculty of Mathematics and Computer, Shahid Bahonar University of Kerman, Kerman, Iran}
\address[b]{Department of Mathematics and Statistics, McMaster University, Hamilton, Ontario, Canada}
\cortext[mycorrespondingauthor]{Corresponding author. Email address: \url{madadi@uk.ac.ir}\\
Email addresses: \url{me.abdi.z@gmail.com} (M. Abdi), \url{bala@mcmaster.ca} (N. Balakrishnan) and \url{a.jamalizadeh@uk.ac.ir} (A. Jamalizadeh).
}

\begin{abstract}
In this paper, a new mixture family of  multivariate normal distributions, formed by mixing multivariate normal distribution and a skewed distribution, is constructed. Some properties of this family, such as characteristic function, moment generating function, and the first four moments are derived. The distributions of affine transformations and canonical forms of the model are also derived.
An EM-type algorithm is developed for the maximum likelihood estimation of model parameters.
 Some special cases of the family, using standard gamma and standard exponential mixture distributions, denoted by $\mathcal{MMNG}$ and $\mathcal{MMNE}$, respectively, are considered.
For the proposed family of distributions, different multivariate measures of skewness ​ are computed.
In order to examine the performance of the developed estimation method, some simulation studies are carried out to show that the maximum likelihood estimates do provide  a good performance.
For different choices of parameters of $\mathcal{MMNE}$ distribution, several multivariate measures of skewness ​ are computed and compared. Because some measures of skewness are scalar and some are vectors, in order to evaluate them properly, a simulation study  is carried out  to determine the power of tests, based on sample versions of skewness measures as test statistics for testing  the fit of the $\mathcal{MMNE}$ distribution. Finally, two real data sets are used to illustrate the usefulness of the proposed model and the associated inferential  methods.
\end{abstract}

\begin{keyword} 
Canonical form \sep
 EM algorithm \sep
 Mean mixtures of normal distribution \sep
 Moments \sep
 Multivariate measures of skewness.

\MSC[2010] 60E07, 60E10, 62H05, 62H10 and 62H12.
\end{keyword}
\end{frontmatter}

\section{Introduction}\label{sec1}
The multivariate normal distribution plays a fundamental role in  statistical analyses and applications.
One of the most basic properties of the normal distribution is the symmetry of its density function.
However, in practice, data sets do not follow the normal distribution or even possess symmetry, and for this reason,
researchers search for new distributions to fit data with different features allowing flexibility in skewness, kurtosis, tails and multimodality; see for example, \cite{Eling (2008),Fung and Hsieh (2000)}.
Several new  families of distributions have been introduced for modeling skewed data,  including normal distribution as a special case.
 One such prominent distribution in the univariate case is the skew normal($\mathcal{SN}$) distribution due to \cite{Azzalini (1985),Azzalini (1986)}.
The multivariate version of the $\mathcal{SN}$ distribution has been introduced in \cite{Azzalini and Dalla Valle (1996)}.
 This distribution has found diverse applications such as  portfolio optimization concepts and risk measurement indices in financial markets; see \cite{Bernardi et al. (2020)}.
A complete set of extensions of multivariate $\mathcal{SN}$ distributions can be found in \cite{Azzalini (2005), Azzalini and Capitanio (2014)}. \cite{Balakrishnan and Scarpa (2012)} calculated and compared several different measures of  skewness for the multivariate $\mathcal{SN}$ distribution. \cite{Balakrishnan et al. (2014)}  proposed a test to assess if a sample comes from a multivariate $\mathcal{SN}$ distribution.
Here, we use $\phi_p(.;{\boldsymbol{\mu}}, {\mathbf{\Sigma}})$ and $\Phi_p(.;{\boldsymbol{\mu}}, {\mathbf{\Sigma}})$ to
denote the probability density function (PDF) and the cumulative distribution function (CDF) of  the $p$-variate normal distribution, with mean ${\boldsymbol{\mu}}$ and covariance matrix ${\mathbf{\Sigma}}$,   respectively,  and also $\phi(.)$ and  $\Phi(.)$ to denote the   PDF and   CDF of  the univariate standard normal distribution, respectively.

 From \cite{Azzalini and Capitanio (2014)} and \cite{Azzalini and Dalla Valle (1996)}, a $p$-dimensional random vector ${\mathbf{Y}}$ follows a multivariate
$\mathcal{SN}$ distribution if it has the PDF
\begin{eqnarray*}
f({\boldsymbol{y}})=2\phi_p({\boldsymbol{y}};{\boldsymbol{\xi}}, {\mathbf{\Omega}})
\Phi\left(\frac{ {\boldsymbol{\delta}}^\top  {\overline{\mathbf{\Omega}}}^{-1}  {\boldsymbol{\omega}}^{-1} ({\boldsymbol{y}}-{\boldsymbol{\xi}}) }{\sqrt{1-{\boldsymbol{\delta}}^\top {\overline{\mathbf{\Omega}}}^{-1} {\boldsymbol{\delta}} }} \right),
 \end{eqnarray*}
with  stochastic representation
\begin{eqnarray}\label{Repre}
{\mathbf{Y}} \stackrel{d}{=}  {\boldsymbol{\xi}}+{\boldsymbol{\omega}}\left( {\boldsymbol{\delta}} U + {\mathbf{Z}}\right),
\end{eqnarray}
where $\stackrel{d}{=}$ stands for equality in distribution, ${\boldsymbol{\xi}}\in \mathbb{R}^p$,  ${\mathbf{Z}} \sim \mathcal{N}_p\left(\boldsymbol{0},{\overline{\mathbf{\Omega}}}-{\mathbf{{\boldsymbol{\delta}}{\boldsymbol{\delta}}}^\top}\right)$ and
univariate random variable $ U$  has a  standard normal distribution within the  truncated interval $\left(0, \infty\right)$,  independently of  ${\mathbf{Z}}$. Truncated normal distribution in the interval $(0, \infty)$ with parameters $(a, b)$ is denoted by $\mathcal{TN}\left(a,b,(0,\infty)\right)$.
The vector  ${\boldsymbol{\delta}}=\left(\delta_1,\ldots,\delta_p\right)^\top$ is the skewness parameter vector, such that $-1<\delta_i<1$, for $i\in\{1,\ldots, p\}$. The matrix
${\boldsymbol{\omega}}=\textrm{diag}\left(\omega_1,\ldots, \omega_p\right)=\left( {\mathbf{\Omega}}\odot \mathbf{I}_p\right)^{1/2}>0$ is a diagonal matrix formed by the standard
deviations of  ${\mathbf{\Omega}}$ and ${\mathbf{\Omega}}={\boldsymbol{\omega}} {\overline{\mathbf{\Omega}}}{\boldsymbol{\omega}}$.  Here,  $\mathbf{I}_p$ is identity matrix of size $p$. The Hadamard product of matrices ${\mathbf{A}}=\left(a_{ij}\right): m\times n$ and ${\mathbf{B}}=\left(b_{ij}\right):  m\times n$ is given by   the
$ m\times n$ matrix  ${\mathbf{A}} \odot{\mathbf{B}}=\left(a_{ij}b_{ij}\right)$. In the  stochastic representation in (\ref{Repre}),  positive definite  matrices ${\mathbf{\Omega}}$ and ${\overline{\mathbf{\Omega}}}$ are   covariance   and correlation matrices, respectively.
The parameters  ${\boldsymbol{\xi}}$,   ${\boldsymbol{\omega}}$ and  ${\boldsymbol{\delta}}$ are the location, scale  and skewness parameters,
respectively.

Upon using the stochastic representation in  (\ref{Repre}), a general new family of mixture distributions of multivariate normal distribution can be introduced based on  arbitrary random variable $U$. A $p$-dimensional random vector ${\mathbf{Y}}$ follows a multivariate
mean mixture of normal ($\mathcal{MMN}$) distribution if, in (\ref{Repre}), $U$ is an arbitrary random variable with CDF $H (. ; {\boldsymbol{\nu}})$, independently of ${\mathbf{Z}}$, indexed by the parameter ${\boldsymbol{\nu}}=\left(\nu_1,\ldots,\nu_p\right)^\top$.
Then, we say that $\textbf{Y}$ has a mean mixture of multivariate normal ($\mathcal{MMN}$)
distribution, and denote it by ${\mathbf{Y}}\sim \mathcal{MMN}_p({\boldsymbol{\xi}},{\mathbf{\Omega}}, {\boldsymbol{\delta}}; H)$.

\cite{Negarestani et al. (2019)} presented a new family of distributions as a mixture of normal distribution and studied its properties in the univariate and multivariate cases. These authors defined a $p$-dimensional random vector ${\mathbf{Y}}$ to have a
multivariate mean mixture  of normal distribution if it has the stochastic representation
$
{\mathbf{Y}} \stackrel{d}{=} {\boldsymbol{\xi}}+{\boldsymbol{\delta}} U + {\mathbf{Z}}
$,
where ${\mathbf{Z}} \sim N_p(\boldsymbol{0},{\mathbf{\Omega}})$ and $ U$ is an arbitrary positive random variable with CDF $H (. ; {\boldsymbol{\nu}})$ independently of ${\mathbf{Z}}$, indexed by the parameter vectors ${\boldsymbol{\nu}}=\left(\nu_1,\ldots,\nu_p\right)^\top$ and ${\boldsymbol{\delta}}=\left(\delta_1,\ldots,\delta_p\right)^\top\in \mathbb{R}^p$.
The stochastic representation    used by \cite{Negarestani et al. (2019)}  is along the lines of the stochastic representation of the restricted multivariate $\mathcal{SN}$ distribution (see \cite{Azzalini (2005)}), but in this work, we use a different stochastic representation in (\ref{Repre}). \cite{Negarestani et al. (2019)} examined some properties  of this family in the univariate case for general $U$, and also two specific cases of the family.
In the present work, we consider the multivariate form of this family  and study its properties.

In (\ref{Repre}), if the random variable $U$ is a skewed  random variable, then the $p$-dimensional vector ${\mathbf{Y}} $ will also be skewed.
In the $\mathcal{MMN}$ family,  skewness can be regulated through the parameter $ {\boldsymbol{\delta}}$.
If in   (\ref{Repre}) ${\boldsymbol{\delta}}={\boldsymbol{{0}}}$, the $\mathcal{MMN}$ family is reduced to the multivariate normal distribution.
The extended form of the $\mathcal{SN}$ distribution is obtained   from (\ref{Repre})  when $U$ is distributed as $\mathcal{N}(0,1)$ variable truncated below $-\tau$ instead of $0$, for some constant $\tau$.
The  representation in (\ref{Repre}) means that the $\mathcal{MMN}$ distribution is a ``mean mixture'' of the multivariate normal distribution when the mixing random variable is $U$. Specifically, we have the following hierarchical representation for the $\mathcal{MMN}$ distribution:
\begin{eqnarray}\label{Hierarchical.U}
{\mathbf{Y}}|{(U=u)}  \sim  N_p \left({\boldsymbol{\xi}}+{\boldsymbol{\omega}}{\boldsymbol{\delta}} u, {\mathbf{\Omega}}-{\boldsymbol{\omega}}{\boldsymbol{\delta}}{\boldsymbol{\delta}}^\top {\boldsymbol{\omega}}\right),
~~~~~ U  \sim   H\left(. ; {\boldsymbol{\nu}}\right).
\end{eqnarray}
According to (\ref {Hierarchical.U}),  in the $\mathcal{MMN}$ model, just the mean parameter is mixed with arbitrary random variable $U$, and so  this class can not be obtained from the Normal Mean-Variance Mixture ($\mathcal{NMVM}$) family. The family of multivariate $\mathcal{NMVM}$ distributions, originated by  \cite{Barndorff et al. (1982)}, is another extension of multivariate normal distribution, with a skewness parameter ${\boldsymbol{\delta}}\in \mathbb{R}^p$.
A $p$-dimensional random vector $\mathbf{Y}$ is said to have a multivariate $\mathcal{NMVM}$ distribution if it has the representation
\begin{eqnarray}\label{Rep-NMVM}
\mathbf{Y}={\boldsymbol{\xi}}+{\boldsymbol{\delta}} U+\sqrt{U}\mathbf{Z},
\end{eqnarray}
where $\mathbf{Z}\sim \mathcal{N}_p(\mathbf{0},\mathbf{\Omega})$ and $U$ is a positive random variable and the CDF of $U$, $H(.;{\boldsymbol{\nu}})$, is the mean-variance mixing distribution.
Both families of distributions  in (\ref{Repre}) and (\ref{Rep-NMVM}) include the multivariate normal distribution as a special case and  can be used for modeling data possessing skewness. In (\ref{Rep-NMVM}), both mean and variance are mixed with the
same positive random variable $U$,  while in (\ref{Repre}) just the mean parameter is mixed with
 $U$, however, the class in (\ref{Repre}) cannot be obtained from the class in (\ref{Rep-NMVM}).

Besides, skewness is a feature commonly found in the returns of some financial assets. For more information on applications of skewed distributions in finance theory, one may refer to \cite{Adcock et al. (2015)}.
In the presence of skewness in asset returns, the
$\mathcal{SN}$ and skew-t ($\mathcal{ST}$) distributions have been found to be useful models in both theoretical and empirical work.
Their parametrization is parsimonious, and they are mathematically tractable, and in financial applications, the distributions are interpretable in terms of the efficient market hypothesis. Furthermore, they lead to theoretical results that are useful for portfolio selection and asset pricing. In actuarial science, the presence of skewness
 in insurance claims data is the primary motivation for using $\mathcal{SN}$ distribution and its
extensions. In this regard, the $\mathcal{MMN}$ family that is developed here will also prove useful in finance, insurance science, and other applied fields.

\cite{Simaan (1993)}  proposed that the
$n$-dimensional vector of returns on financial assets should be represented as
${\boldsymbol{X}}={\boldsymbol{U}}+{\boldsymbol{\lambda}} V$. The $n$-dimensional vector ${\boldsymbol{U}}$ is assumed to have  a multivariate elliptically symmetric
distribution, independently of the non-negative univariate
random variable $V$, having an unspecified skewed distribution.
The vector ${\boldsymbol{\lambda}}$, whose elements may take any real values, induces
skewness in the return of individual assets. \cite{Adcock and Shutes
(2012)} have described multivariate versions of the normal-exponential
and normal-gamma distributions. Both of them are specific cases of the
model  proposed in \cite{Simaan (1993)}.
 \cite{Adcock (2014)} and \cite{Adcock and Shutes (2012)} used the representation in \cite{Simaan (1993)}, with specific choices of ${\boldsymbol{U}} $ and $ V$, and introduced a number of  distributions such as  $\mathcal{SN}$, extended $\mathcal{SN}$, $\mathcal{ST}$, normal-exponential, and normal-gamma,  and investigated the corresponding distributions and their applications in capital pricing, return on financial assets and portfolio selection.

In this paper, with  an  arbitrary random variable $U$ for the $\mathcal{MMN}$ family with stochastic representation in (\ref{Repre}),
basic distributional properties of the class such as the characteristic function (CF), the moment generating function (MGF), the first four moments of the model,  distributions of linear and affine transformations,  the canonical form of the family and the mode of the model are derived in general. Also, the  maximum likelihood estimation of the parameters by using an EM-type algorithm is discussed,  and then different measures of multivariate skewness are obtained.
The special cases when $U$ has standard gamma and standard exponential distributions, with the corresponding distributions denoted by $\mathcal{MMNG}$ and $\mathcal{MMNE}$ distributions, respectively, are studied in detail.
For the $\mathcal{MMNG}$ distribution, in addition to all the above basic properties of the distribution the  infinitely divisibility of the model is also discussed.  For the  $\mathcal{MMNE}$ distribution, the basic properties of the distribution as well as log-concavity of the model are discussed.
The  maximum likelihood estimates of the parameters of the $\mathcal{MMNE}$ distribution
are   evaluated using the bias and the mean square error by means of a simulation study.
Moreover,  various multivariate measures of skewness ​  are computed and compared.
Finally, for two real data sets,   the $\mathcal{MMNE}$ distribution is fitted  and compared with the $\mathcal{SN}$ and $\mathcal{ST}$ distributions in terms of log-likelihood value  as well as  AIC  and BIC criteria.

\section{ Model and Properties}\label{sec2}
In this section, some basic properties of the model  are studied.
From  (\ref{Repre}), if $U$ has  a PDF $h (. ; {\boldsymbol{\nu}})$, an integral form of
the PDF of ${\mathbf{Y}}\sim \mathcal{MMN}_p({\boldsymbol{\xi}},{\mathbf{\Omega}}, {\boldsymbol{\delta}}; H)$ can be obtained as
\begin{eqnarray}\label{density-y}
f_{MMN_p}({\boldsymbol{y}};{\boldsymbol{\xi}},{\mathbf{\Omega}}, {\boldsymbol{\delta}},{\boldsymbol{\nu}}) =\int_{-\infty}^{+\infty}\phi_p\left({\boldsymbol{y}}; {\boldsymbol{\xi}}+{\boldsymbol{\omega}}{\boldsymbol{\delta}} u, {\mathbf{\Omega}}-{\boldsymbol{\omega}}{\mathbf{{\boldsymbol{\delta}}{\boldsymbol{\delta}}}^\top} {\boldsymbol{\omega}}\right)d H(u;{\boldsymbol{\nu}})=\int_{-\infty}^{+\infty}\phi_p\left({\boldsymbol{y}}; {\boldsymbol{\xi}}+{\boldsymbol{\omega}}{\boldsymbol{\delta}} u, {\mathbf{\Omega}}-{\boldsymbol{\omega}}{\mathbf{{\boldsymbol{\delta}}{\boldsymbol{\delta}}}^\top} {\boldsymbol{\omega}}\right) h(u;{\boldsymbol{\nu}}) du.
 \end{eqnarray}
 We now present  some theorems and lemmas with regard to different properties of  these distributions, proofs of which are presented in Appendix A.\\

\noindent{\textbf{Remark 1.}}
 We can introduce the normalized $\mathcal{MMN}$ distribution through the transformation
$\textbf{X}={\boldsymbol{\omega}}^{-1}\left({\mathbf{Y}}- {\boldsymbol{\xi}}\right)$. It is immediate that the stochastic representation of
$ \textbf{X}= {\boldsymbol{\delta}} U + {\mathbf{Z}}$ has the hierarchial representation
$
{\mathbf{X}}|(U=u) \sim \mathcal{N}_p \left({\boldsymbol{\delta}} u, {\overline{\mathbf{\Omega}}}-{\mathbf{{\boldsymbol{\delta}}{\boldsymbol{\delta}}}^\top}\right)
$
and
$
U\sim H(.; {\boldsymbol{\nu}})
$.
Then, we say that ${\mathbf{X}}$ has a normalized mean mixture of multivariate normal
distributions, and denote it by ${\mathbf{X}}\sim \mathcal{MMN}_p\left({\boldsymbol{0}}, {\overline{\mathbf{\Omega}}},{\boldsymbol{\delta}}; H\right)$.

\begin{lemma}\label{Lem-CF-MG}
If ${\mathbf{Y}}\sim \mathcal{MMN}_p({\boldsymbol{\xi}},{\mathbf{\Omega}}, {\boldsymbol{\delta}}; H)$,  the CF and MGF of ${\mathbf{Y}}$ are  as follows:
\begin{eqnarray}
C_{\mathbf{Y}}({\mathbf{t}})=e^{i{\mathbf{t}}^\top {\boldsymbol{\xi}}+  \frac{1}{2}{
{\mathbf{t}}^\top {\mathbf{ \Sigma}}_{\mathbf{Y}}  {\mathbf{t}}}} C_U\left(i{\mathbf{t}}^\top{\boldsymbol{\omega}}{\boldsymbol{\delta}};{\boldsymbol{\nu}}\right),~~~~~
M_{\mathbf{Y}}({\mathbf{t}})=e^{{\mathbf{t}}^\top {\boldsymbol{\xi}}+  \frac{1}{2}{
{\mathbf{t}}^\top {\mathbf{ \Sigma}}_{\mathbf{Y}}  {\mathbf{t}}}} M_U\left({\mathbf{t}}^\top{\boldsymbol{\omega}}{\boldsymbol{\delta}};{\boldsymbol{\nu}}\right),\label{MtY}
\end{eqnarray}
respectively, where $i=\sqrt{-1}$,  ${\mathbf{ \Sigma}}_{\mathbf{Y}}={\mathbf{\Omega}}-{\boldsymbol{\omega}}{\mathbf{{\boldsymbol{\delta}}{\boldsymbol{\delta}}}^\top} {\boldsymbol{\omega}}$, and $C_U(.;{\boldsymbol{\nu}})=C_U(.)$ and   $M_U(.;{\boldsymbol{\nu}})=M_U(.)$ are the CF and MGF of $U$, respectively.
\end{lemma}

Moreover, if ${\mathbf{X}}\sim \mathcal{MMN}_p\left({\boldsymbol{0}}, {\overline{\mathbf{\Omega}}},{\boldsymbol{\delta}}; H\right)$,  the CF and MGF of ${\mathbf{X}}$ are
\begin{eqnarray}
C_{\mathbf{X}}({\mathbf{t}})=e^{  \frac{1}{2}{
{\mathbf{t}}^\top {\mathbf{ \Sigma}}_{\mathbf{X}}  {\mathbf{t}}}} C_U\left(i{\mathbf{t}}^\top{\boldsymbol{\delta}};{\boldsymbol{\nu}}\right),~~~~~
M_{\mathbf{X}}({\mathbf{t}})=e^{ \frac{1}{2}{
{\mathbf{t}}^\top {\mathbf{ \Sigma}}_{\mathbf{X}}  {\mathbf{t}}}} M_U\left({\mathbf{t}}^\top{\boldsymbol{\delta}};{\boldsymbol{\nu}}\right)\label{MtX},
\end{eqnarray}
respectively, where ${\mathbf{ \Sigma}}_{\mathbf{X}}={\overline{\mathbf{\Omega}}}-{\mathbf{{\boldsymbol{\delta}}{\boldsymbol{\delta}}}^\top}$.  The first four moments of  ${\mathbf{X}}$, presented  in the following lemma, are derived  by using the partial derivatives of MGF of normalized MMN distribution, and these, in turn, can be used to obtain  the first four moments of ${\mathbf{Y}}$.

\begin{lemma}\label{Lem1}
Suppose   ${\mathbf{X}}\sim \mathcal{MMN}_p\left({\boldsymbol{0}}, {\overline{\mathbf{\Omega}}},{\boldsymbol{\delta}}; H\right)$. Then, the first four moments of  ${\mathbf{X}}$ are as follows:
\begin{eqnarray}
M_1({\mathbf{X}})&=&M_1^{\mathbf{X}}={\rm E}[U]{\boldsymbol{\delta}},\label{M1X}\\
M_2({\mathbf{X}})&=&M_2^{\mathbf{X}}={\mathbf{ \Sigma}}_{\mathbf{X}}+{\rm E}\left[U^2\right]\left({\boldsymbol{\delta}}\otimes{\boldsymbol{\delta}}^\top\right),\label{M2X}\\
M_3({\mathbf{X}})&=&M_3^{\mathbf{X}}= {\rm E}[U]\left\{{\boldsymbol{\delta}} \otimes {\mathbf{ \Sigma}}_{\mathbf{X}}+\mathrm{vec}\left({\mathbf{ \Sigma}}_{\mathbf{X}}\right){\boldsymbol{\delta}}^\top+ \left(\mathbf{I}_p\otimes {\boldsymbol{\delta}} \right){\mathbf{ \Sigma}}_{\mathbf{X}} \right\}+  {\rm E}\left[U^3\right] \left(\mathbf{I}_p\otimes {\boldsymbol{\delta}} \right)\left({\boldsymbol{\delta}}\otimes{\boldsymbol{\delta}}^\top\right), \label{M3X}\\
M_4({\mathbf{X}})&=&M_4^{\mathbf{X}}=\left(\mathbf{I}_{p^2} +\mathbf{U}_{p,p}\right)\left({\mathbf{ \Sigma}}_{\mathbf{X}}\otimes {\mathbf{ \Sigma}}_{\mathbf{X}}\right)+
\mathrm{vec}({\mathbf{ \Sigma}}_{\mathbf{X}})\left(\mathrm{vec}({\mathbf{ \Sigma}}_{\mathbf{X}})\right)^\top+{\rm E}\left[U^2\right]\left[
{\boldsymbol{\delta}}  \otimes {\boldsymbol{\delta}}^\top \otimes {\mathbf{ \Sigma}}_{\mathbf{X}}+ {\boldsymbol{\delta}}\otimes {\mathbf{ \Sigma}}_{\mathbf{X}}   \otimes  {\boldsymbol{\delta}}^\top\right.\nonumber\\
&&+ \left. {\mathbf{ \Sigma}}_{\mathbf{X}} \otimes {\boldsymbol{\delta}}\otimes{\boldsymbol{\delta}}^\top+  {\boldsymbol{\delta}}^\top\otimes {\mathbf{ \Sigma}}_{\mathbf{X}}   \otimes  {\boldsymbol{\delta}}+{\boldsymbol{\delta}}^\top\otimes \mathrm{vec}({\mathbf{ \Sigma}}_{\mathbf{X}})   \otimes  {\boldsymbol{\delta}}^\top  +({\boldsymbol{\delta}}  \otimes {\boldsymbol{\delta}}) (\mathrm{vec}({\mathbf{ \Sigma}}_{\mathbf{X}}) )^\top \right]+{\rm E}\left[U^4\right]{\boldsymbol{\delta}}{\boldsymbol{\delta}}^\top\otimes {\boldsymbol{\delta}}{\boldsymbol{\delta}}^\top ,
\end{eqnarray}
where   ${\rm E}(U^k)=M_U^{(k)}(0)$, with  $M_U^{(k)}(.)$  being the k-th derivative of $M_U(t)$ with respect to $t$.
\end{lemma}
The Kronecker product of matrices ${\mathbf{A}}=\left(a_{ij}\right): m\times n$ and ${\mathbf{B}}=\left(b_{ij}\right): p\times q$ is a
$mp \times nq$ matrix ${\mathbf{A}} \otimes{\mathbf{B}}=\left(a_{ij}{\mathbf{B}}\right)$. A matrix ${\mathbf{A}}  = \left({\mathbf{a}}_1,\ldots , {\mathbf{ a}}_n \right): m \times n$ with columns ${\mathbf{a}}_1, \ldots,  {\mathbf{a}}_n$ is sometimes written as a vector and called $\mathrm{vec}({\mathbf{A}})$, defined by $\mathrm{vec}({\mathbf{A}})=\left({\mathbf{a}}_1^\top, \ldots,  {\mathbf{a}}_n^\top \right)^\top$. The matrix $\mathbf{U}_{p,p}$ is the permutation matrix (commutation matrix)  associated with a $p \times p$ matrix (its size is $p^2 \times p^2 $). For details about  Kronecker product, permutation matrix and its properties, see \cite{Graham (1981)} and  \cite{Schott (2016)}.
We extend the results of Lemma \ref{Lem1}, using  the   stochastic representation in (\ref{Repre}), to incorporate  location and scale parameters, ${\boldsymbol{\xi}}$ and ${\boldsymbol{\omega}}$,  through the transformation ${\mathbf{Y}}={\boldsymbol{\xi}}+{\boldsymbol{\omega}}{\mathbf{X}}$.

\begin{thm}\label{ThmLem2}
If ${\mathbf{Y}}\sim \mathcal{MMN}_p({\boldsymbol{\xi}},{\mathbf{\Omega}}, {\boldsymbol{\delta}}; H)$,
 then its first four moments are as follows:
\begin{eqnarray}
M_1({\mathbf{Y}})&=& {\boldsymbol{\xi}}+{\boldsymbol{\omega}}M_1^{\mathbf{X}},\label{M1Y}\\
M_2({\mathbf{Y}})&=&{\boldsymbol{\xi}}\otimes{\boldsymbol{\xi}}^\top+{\boldsymbol{\xi}}\otimes \left({\boldsymbol{\omega}}M_1^{\mathbf{X}}\right)^\top+{\boldsymbol{\omega}}M_1^{\mathbf{X}}\otimes {\boldsymbol{\xi}}^\top+  {\boldsymbol{\omega}}M_2^{\mathbf{X}} {\boldsymbol{\omega}}, \label{M2Y}\\
M_3({\mathbf{Y}})&=&{\boldsymbol{\xi}}{\boldsymbol{\xi}}^\top \otimes{\boldsymbol{\xi}}+{\boldsymbol{\xi}}{\boldsymbol{\xi}}^\top\otimes\left({\boldsymbol{\omega}}M_1^{\mathbf{X}}\right )
+{\boldsymbol{\xi}}\left({\boldsymbol{\omega}}M_1^{\mathbf{X}} \right)^\top \otimes{\boldsymbol{\xi}}
+\left({\boldsymbol{\omega}}M_1^{\mathbf{X}} \right)\otimes{\boldsymbol{\xi}}{\boldsymbol{\xi}}^\top+ \left({\boldsymbol{\omega}}M_2^{\mathbf{X}}{\boldsymbol{\omega}} \right)\otimes{\boldsymbol{\xi}}+{\boldsymbol{\xi}}\otimes \left({\boldsymbol{\omega}}M_2^{\mathbf{X}}{\boldsymbol{\omega}}\right)\nonumber\\
&&+  ({\boldsymbol{\omega}}\otimes{\boldsymbol{\omega}})\mathrm{vec}\left(M_2^{\mathbf{X}}\right)\otimes{\boldsymbol{\xi}}^\top
+({\boldsymbol{\omega}}\otimes{\boldsymbol{\omega}})M_3^{\mathbf{X}}{\boldsymbol{\omega}},  \label{M3Y}\\
M_4({\mathbf{Y}})&=& {\boldsymbol{\xi}}{\boldsymbol{\xi}}^\top \otimes{\boldsymbol{\xi}}{\boldsymbol{\xi}}^\top
+  {\boldsymbol{\xi}}{\boldsymbol{\xi}}^\top \otimes{\boldsymbol{\xi}}\left({\boldsymbol{\omega}}M_1^{\mathbf{X}} \right)^\top
+{\boldsymbol{\xi}}{\boldsymbol{\xi}}^\top \otimes \left({\boldsymbol{\omega}}M_1^{\mathbf{X}} \right){\boldsymbol{\xi}}^\top
+ {\boldsymbol{\xi}}{\boldsymbol{\xi}}^\top \otimes \left({\boldsymbol{\omega}}M_2^{\mathbf{X}}{\boldsymbol{\omega}} \right)+{\boldsymbol{\xi}}\left({\boldsymbol{\omega}}M_1^{\mathbf{X}}\right)^\top \otimes{\boldsymbol{\xi}}{\boldsymbol{\xi}}^\top\nonumber\\
&&+  \left({\boldsymbol{\xi}} \otimes{\boldsymbol{\xi}}\right) \left(\mathrm{vec}\left(M_2^{\mathbf{X}}\right) \right)^\top({\boldsymbol{\omega}}\otimes{\boldsymbol{\omega}})
+{\boldsymbol{\xi}}  \otimes \left({\boldsymbol{\omega}}M_2^{\mathbf{X}}{\boldsymbol{\omega}} \right) \otimes{\boldsymbol{\xi}}^\top +{\boldsymbol{\xi}}  \otimes {\boldsymbol{\omega}}\left(M_3^{\mathbf{X}} \right)^\top({\boldsymbol{\omega}}\otimes{\boldsymbol{\omega}})
+\left({\boldsymbol{\omega}}M_1^{\mathbf{X}} \right){\boldsymbol{\xi}}^\top \otimes{\boldsymbol{\xi}}{\boldsymbol{\xi}}^\top\nonumber\\
&&+{\boldsymbol{\xi}}^\top  \otimes    \left({\boldsymbol{\omega}}M_2^{\mathbf{X}}{\boldsymbol{\omega}} \right)   \otimes {\boldsymbol{\xi}}
+{\boldsymbol{\xi}}^\top  \otimes ({\boldsymbol{\omega}}\otimes{\boldsymbol{\omega}})\mathrm{vec}\left(M_2^{\mathbf{X}}\right)   \otimes{\boldsymbol{\xi}}^\top
+{\boldsymbol{\xi}}^\top  \otimes ({\boldsymbol{\omega}}\otimes{\boldsymbol{\omega}}) M_3^{\mathbf{X}}{\boldsymbol{\omega}}
+\left({\boldsymbol{\omega}}M_2^{\mathbf{X}}{\boldsymbol{\omega}} \right)   \otimes {\boldsymbol{\xi}}{\boldsymbol{\xi}}^\top\nonumber\\
&&+ {\boldsymbol{\omega}} \left(M_3^{\mathbf{X}}\right)^\top ({\boldsymbol{\omega}}\otimes{\boldsymbol{\omega}})\otimes{\boldsymbol{\xi}}
+ ({\boldsymbol{\omega}}\otimes{\boldsymbol{\omega}}) M_3^{\mathbf{X}}{\boldsymbol{\omega}}   \otimes{\boldsymbol{\xi}}^\top
+({\boldsymbol{\omega}}\otimes{\boldsymbol{\omega}}) M_4^{\mathbf{X}}({\boldsymbol{\omega}}\otimes{\boldsymbol{\omega}}). \label{M4Y}
\end{eqnarray}
\end{thm}
From the above expressions, we can obtain  the  mean vector and covariance matrix of $\mathcal{MMN}_p({\boldsymbol{\xi}},{\mathbf{\Omega}}, {\boldsymbol{\delta}}; H)$  family   as
${\rm E}({\mathbf{Y}})={\boldsymbol{\xi}}+{\rm E}(U){\boldsymbol{\omega}}{\boldsymbol{\delta}}$ and
${\rm var}\left({\mathbf{Y}}\right)=
{\mathbf{\Omega}}+ ({\rm var}(U)-1){\boldsymbol{\omega}}{\boldsymbol{\delta}}{\boldsymbol{\delta}}^\top{\boldsymbol{\omega}}$.
 Multiplication of $M_{\mathbf{X}}({\mathbf{t}})$ by the MGF of the $\mathcal{N}_p({\boldsymbol{\mu}},{\mathbf{\Sigma}})$
distribution, $\exp\left({{\mathbf{t}}^\top {\boldsymbol{\mu}}+  {
{\mathbf{t}}^\top {\mathbf{ \Sigma}}{\mathbf{t}}}}/2\right)$, is still a function of type $M_{\mathbf{X}}({\mathbf{t}})$, and we thus obtain the following result.

\begin{thm}
If ${\mathbf{Y}}_1\sim \mathcal{MMN}_p({\boldsymbol{\xi}},{\mathbf{\Omega}}, {\boldsymbol{\delta}}; H)$ and
${\mathbf{Y}}_2 \sim \mathcal{N}_p({\boldsymbol{\mu}},{\mathbf{\Sigma}})$ are independent variables, then
$
{\mathbf{Y}}={\mathbf{Y}}_1+{\mathbf{Y}}_2 \sim \mathcal{MMN}_p\left({\boldsymbol{\xi}}_{\mathbf{Y}},{\mathbf{\Omega}}_{\mathbf{Y}}, {\boldsymbol{\delta}}_{\mathbf{Y}}; H\right),
$
where
$
{\boldsymbol{\xi}}_{\mathbf{Y}}={\boldsymbol{\xi}}+{\boldsymbol{\mu}} $,
${\mathbf{\Omega}}_{\mathbf{Y}}={\mathbf{\Omega}}+{\mathbf{\Sigma}}$, and
$
{\boldsymbol{\delta}}_{\mathbf{Y}}={\boldsymbol{\omega}}_{\mathbf{Y}}^{-1}{\boldsymbol{\omega}}{\boldsymbol{\delta}},
$
with  ${\boldsymbol{\omega}}_{\mathbf{Y}}=( {\mathbf{\Omega}}_{\mathbf{Y}}\odot \mathbf{I}_p)^{1/2}$.
\end{thm}

From the MGF's $M_{\mathbf{X}}({\mathbf{t}})$ and $M_{\mathbf{Y}}({\mathbf{t}})$, it is clear that the family
of MMN distributions is closed under affine transformations, as given in the following results.

\begin{thm}\label{Thm1}
If ${\mathbf{X}}\sim \mathcal{MMN}_p\left({\boldsymbol{0}},{\overline{\mathbf{\Omega}}}, {\boldsymbol{\delta}}; H\right)$ and  ${\mathbf{A}}$ is a non-singular $p\times p$ matrix such that   $\mathrm{diag}\left({\mathbf{A}}^\top {\overline{\mathbf{\Omega}}} {\mathbf{A}} \right)={\mathbf{I}}_p$, that is,  ${\mathbf{A}}^\top {\overline{\mathbf{\Omega}}} {\mathbf{A}}$ is a correlation matrix, then
$
{\mathbf{A}}^\top{\mathbf{X}} \sim  \mathcal{MMN}_p\left({\boldsymbol{0}},{\mathbf{A}}^\top{\overline{\mathbf{\Omega}}}{\mathbf{A}},{\mathbf{A}}^\top{\mathbf{\delta}}; H\right).
$
\end{thm}

\begin{thm}\label{Thm2}
If ${\mathbf{Y}}\sim \mathcal{MMN}_p({\boldsymbol{\xi}},{\mathbf{\Omega}}, {\boldsymbol{\delta}}; H)$,   A is a full-rank $p\times h$ matrix,
with $h \leq p$, and $\textbf{c}\in \mathbb{R}^h$, then
$
{\mathbf{T}}=\textbf{c}+{\mathbf{A}}^\top {\mathbf{Y}} \sim \mathcal{MMN}_h\left({\boldsymbol{\xi}}_{\mathbf{T}},{\mathbf{\Omega}}_{\mathbf{T}}, {\boldsymbol{\delta}_{\mathbf{T}}}; H\right),
$
where
 $
{\boldsymbol{\xi}}_{\mathbf{T}}=\textbf{c}+{\mathbf{A}}^\top {\boldsymbol{\xi}}
$,
$
{\mathbf{\Omega}}_{\mathbf{T}}={\mathbf{A}}^\top {\mathbf{\Omega}}{\mathbf{A}}
$, and
$
{\boldsymbol{\delta}}_{\mathbf{T}}={\boldsymbol{\omega}}_{\mathbf{T}}^{-1}{\mathbf{A}}^\top{\boldsymbol{\omega}}{\boldsymbol{\delta}},
$
with ${\boldsymbol{\omega}}_{\mathbf{T}}=\left( {\mathbf{\Omega}}_{\mathbf{T}}\odot \mathbf{I}_h\right)^{1/2}$.
\end{thm}

As in the case of multivariate $\mathcal{SN}$ distribution, (see \cite{Azzalini and Capitanio (2014)}), it can be shown  that,  if the random vector ${\mathbf{Y}}$ is partitioned into a number of random vectors, the independence occurs between its components when at least one component follows the $\mathcal{MMN}$ distribution and the others have normal distribution, that is, the independence between components occurs when only one component of the skewness parameter ${\boldsymbol{\delta}}$ is non-zero and all others are zero.
Without loss of generality, from here on, it is assumed that the first element of ${\boldsymbol{\delta}}$ is non-zero.
We now focus  on a specific type of linear transformation of the $\mathcal{MMN}$ variable, having special relevance for theoretical developments
but also to some extent  for practical reasons.
\begin{thm}\label{Thm3}
For a given variable  ${\mathbf{Y}}\sim \mathcal{MMN}_p({\boldsymbol{\xi}},{\mathbf{\Omega}}, {\boldsymbol{\delta}}; H)$, there exists a linear transformation
$\textbf{Z}^*={\mathbf{A}}_*({\mathbf{Y}}-{\boldsymbol{\xi}})$ such that $\mathbf{Z}^*\sim \mathcal{MMN}_p\left({\boldsymbol{0}},{\mathbf{I}}_p, {\boldsymbol{\delta}}_{\textbf{Z}^*}; H\right)$, where at most one component of ${\boldsymbol{\delta}}_{\textbf{Z}^*}$ is not zero, and ${\boldsymbol{\delta}}_{\textbf{Z}^*}=({\mathbf{\delta}}_*, 0, \ldots, 0)^\top$ with $\delta_*=\left({\boldsymbol{\delta}}^\top \overline{{\mathbf{\Omega}}}^{-1} {\boldsymbol{\delta}}\right)^{1/2}$.
\end{thm}
The variable $\textbf{Z}^*$, which we shall sometimes refer to as a canonical variate,
 consists of  $p$ independent components. The joint density is given by
the product of $p-1$ standard normal densities and at most one non-Gaussian
component $\mathcal{MMN}_1\left(0, 1, {\mathbf{\delta}}_*; H\right)$; that is, the density of $\textbf{Z}^*$ is
\begin{eqnarray}\label{density}
f_{\textbf{Z}^*}(\textbf{z})= f_{Z_1^*}(z_1) \prod_{i=2}^{p}\phi( z_i),
\end{eqnarray}
where $Z_1^*\sim \mathcal{MMN}_1(0,1,{\mathbf{\delta}}_*; H)$  (for univariate $\mathcal{MMN}$ distribution, see \cite{Negarestani et al. (2019)}).
 Although  Theorem \ref{Thm3} ensures that it is possible to obtain a canonical
form, note that in general there are many possible ways
to do so, but it is not obvious how to achieve the canonical form in practice.
To find the appropriate ${\mathbf{A}}_*$  in the linear transformation   $\textbf{Z}^*={\mathbf{A}}_*({\mathbf{Y}}-{\boldsymbol{\xi}})$, it is sufficient to find  a ${\mathbf{A}}_*$ satisfying  the following two conditions:
$
{\mathbf{A}}_*^\top {\mathbf{\Omega}} {\mathbf{A}}_*={\mathbf{I}}_p$
and
$
{\mathbf{A}}_*^\top{\boldsymbol{\omega}}{\boldsymbol{\delta}}={\boldsymbol{\delta}}_{\textbf{Z}^*}=({\mathbf{\delta}}_*, 0,\ldots, 0)^\top
$.
The canonical form facilitates the computation of the  mode of  the distribution and  the multivariate coefficients of skewness.
\begin{thm}\label{Thm5-mode}
If ${\mathbf{Y}}\sim \mathcal{MMN}_p({\boldsymbol{\xi}},{\mathbf{\Omega}},{\boldsymbol{\delta}}; H)$, the mode of ${\mathbf{Y}}$ is
$
\textbf{M}_0={\boldsymbol{\xi}}+\frac{m_0^*}{{\mathbf{\delta}}_*}{\boldsymbol{\omega}}{\boldsymbol{\delta}},
$
where $\delta_*=\left({\boldsymbol{\delta}}^\top \overline{{\mathbf{\Omega}}}^{-1} {\boldsymbol{\delta}}\right)^{1/2}$ and $m_0^*$ is
the mode of the univariate $\mathcal{MMN}_1\left(0,1,\delta_*; H\right)$ distribution.
\end{thm}

\section{Likelihood Estimation through EM Algorithm}\label{Est-sec}
For obtaining the maximum likelihood
estimates of all the parameters of $\mathcal{MMN}_p({\boldsymbol{\xi}},{\mathbf{\Omega}},{\boldsymbol{\delta}}; H)$, we propose an EM-type algorithm as in  \cite{Meng and Rubin (1993)}. Let $\mathbf{Y} = (\mathbf{Y}_1, \ldots, \mathbf{Y}_n)^\top$ be a random sample of
size $n$ from a $\mathcal{MMN}_p({\boldsymbol{\xi}},{\mathbf{\Omega}},{\boldsymbol{\delta}}; H)$  distribution.
Consider the stochastic representation in (\ref{Repre})
for $\mathbf{Y}_i, i\in\{1,\ldots, n\}$. Following the EM algorithm, let $(\mathbf{Y}_i,U_i),  i\in\{1,\ldots, n\}$, be the complete data, where $\mathbf{Y}_i$ is the observed data and $U_i$ is considered as missing data. Let ${\boldsymbol{\theta}}=({\boldsymbol{\xi}},{\mathbf{\Omega}},{\boldsymbol{\delta}},{\boldsymbol{\nu}})$.
Using (\ref{Hierarchical.U}), the distribution of ${\mathbf{Y}}_i$, for
$i\in\{1,\ldots, n\}$, can be written hierarchically as
\begin{eqnarray*}\label{hierarchicall}
{\mathbf{Y}}_i|(U_i=u_i) \sim \mathcal{N}_p ({\boldsymbol{\xi}}+{\boldsymbol{\omega}}{\boldsymbol{\delta}} u_i, {\mathbf{\Sigma}}_{\mathbf{Y}}),~~~~~~
U_i \stackrel{iid}{\sim} H(.;{\boldsymbol{\nu}}),
\end{eqnarray*}
where $\stackrel{iid}{\sim}$ denotes independence of random variables and  ${\mathbf{\Sigma}}_{\mathbf{Y}}={\mathbf{\Omega}}-{\boldsymbol{\omega}}{\mathbf{{\boldsymbol{\delta}}{\boldsymbol{\delta}}}^\top} {\boldsymbol{\omega}}$.
 Let ${\boldsymbol{y}}=\left({\boldsymbol{y}}^\top_1, \ldots, {\boldsymbol{y}}^\top_n\right)$ , where ${\mathbf{y}}_i$ is a realization of $\mathcal{MMN}_p({\boldsymbol{\xi}},{\mathbf{\Omega}},{\boldsymbol{\delta}}; H)$.
Because
\begin{eqnarray}\label{complet-i}
f({\boldsymbol{y}}_i,u_i)=f({\boldsymbol{y}}_i|u_i)h(u_i;{\boldsymbol{\nu}}),
\end{eqnarray}
the complete data log-likelihood function, ignoring additive constants, is obtained from (\ref{complet-i}) as
\begin{eqnarray*}\label{complet-c}
\ell_c({\boldsymbol{\theta}})&=&-\frac{n}{2}\ln |{\mathbf{\Sigma}}_{\mathbf{Y}}|
-\frac{1}{2}\sum_{i=1}^{n}({\boldsymbol{y}}_i- {\boldsymbol{\xi}})^\top{\mathbf{\Sigma}}_{\mathbf{Y}}^{-1}({\boldsymbol{y}}_i- {\boldsymbol{\xi}})+{\boldsymbol{\alpha}}^\top {\mathbf{\Sigma}}_{\mathbf{Y}}^{-1} \sum_{i=1}^{n} u_i ({\boldsymbol{y}}_i- {\boldsymbol{\xi}})-\frac{1}{2} {\boldsymbol{\alpha}}^\top  {\mathbf{\Sigma}}_{\mathbf{Y}}^{-1}{\boldsymbol{\alpha}}\sum_{i=1}^{n}u_i^2+\sum_{i=1}^{n} \ln h(u_i;{\boldsymbol{\nu}}),
\end{eqnarray*}
where ${\boldsymbol{\alpha}}={\boldsymbol{\omega}}{\boldsymbol{\delta}}$.  Let us set
\begin{eqnarray}
\widehat{E_{i1}}^{(k)}={\rm E}\left[U_i|{\mathbf{Y}}_i={\boldsymbol{y}}_i, \widehat{{\boldsymbol{\theta}}}^{(k)}\right],~~~~~
\widehat{E_{i2}}^{(k)}={\rm E}\left[U_i^2|{\mathbf{Y}}_i={\boldsymbol{y}}_i,\widehat{{\boldsymbol{\theta}}}^{(k)}\right],\label{Eihat-general}
\end{eqnarray}
where   $\widehat{{\boldsymbol{\theta}}}^{(k)}=\left(\widehat{{\boldsymbol{\xi}}}^{(k)}, \widehat{{\boldsymbol{\Omega}}}^{(k)}, \widehat{{\boldsymbol{\delta}}}^{(k)},\widehat{{\boldsymbol{\nu}}}^{(k)}\right)$.
After some simple algebra and using (\ref{Eihat-general}),  the expectation with
respect to $U$ conditional on ${\boldsymbol{Y}}$, of the complete log-likelihood function, has the form
\begin{eqnarray}\label{Com-Like}
Q\left({\boldsymbol{\theta}}|\widehat{{\boldsymbol{\theta}}}^{(k)}\right)&=&\frac{n}{2}\ln|{\mathbf{\Sigma}}_{\mathbf{Y}}^{-1}|
-\frac{1}{2}\sum_{i=1}^{n}({\boldsymbol{y}}_i- {\boldsymbol{\xi}})^\top{\mathbf{\Sigma}}_{\mathbf{Y}}^{-1}({\boldsymbol{y}}_i- {\boldsymbol{\xi}})+ \sum_{i=1}^n \textrm{tr}\left[ {\mathbf{\Sigma}}_{\mathbf{Y}}^{-1} (\boldsymbol{y}_i-{\boldsymbol{\xi}}) {\boldsymbol{\alpha}}^\top\right]\widehat{E_{i1}}^{(k)}\nonumber\\
&&-\frac{1}{2} \textrm{tr}\left[ {\mathbf{\Sigma}}_{\mathbf{Y}}^{-1}{\boldsymbol{\alpha}}{\boldsymbol{\alpha}}^\top\right]\sum_{i=1}^n \widehat{E_{i2}}^{(k)}+\sum_{i=1}^n {\rm E}\left[\ln h(u_i;{\boldsymbol{\nu}})|{\mathbf{Y}}_i={\boldsymbol{y}}_i,\widehat{{\boldsymbol{\theta}}}^{(k)}\right],
\end{eqnarray}
where $\overline{{\boldsymbol{y}}}=\frac{1}{n} \sum_{i=1}^n {\boldsymbol{y}}_i$ is the sample mean vector.
The EM-type  algorithm for the ML estimation of ${\boldsymbol{\theta}} =({\boldsymbol{\xi}},{\mathbf{\Omega}},{\boldsymbol{\delta}},{\boldsymbol{\nu}})$ then proceeds as
follows:\\

\noindent \textbf{Algorithm 1.}  Based on the initial value of   $\boldsymbol{\theta}^{(0)}=\left(\boldsymbol{\xi}^{(0)}, \mathbf{\Omega}^{(0)}, \boldsymbol{\delta}^{(0)}, \boldsymbol{\nu}^{(0)}\right)$, the EM-type algorithm
iterates between the following E-step and M-step:

\noindent \textbf{E-step}: Given the estimates of model parameters at the $k$-th iteration, say  ${\boldsymbol{\theta}}=\widehat{{\boldsymbol{\theta}}}^{(k)}$, compute $\widehat{E_{i1}}^{(k)}$ and $\widehat{E_{i2}}^{(k)}$, for $i \in\{1, 2, \ldots, n\}$;

\noindent \textbf{M-step 1}:  Maximization of (\ref{Com-Like}) over parameters ${\boldsymbol{\xi}}$, ${\boldsymbol{\alpha}}$ and ${\mathbf{\Sigma}}_{\mathbf{Y}}$ leads to the following  closed-form expressions:
\begin{eqnarray*}
\widehat{{\boldsymbol{\alpha}}}^{(k+1)}&=&\frac{\sum_{i=1}^n {\boldsymbol{y}}_i\widehat{E_{i1}}^{(k)}-\overline{{\boldsymbol{y}}}\sum_{i=1}^n \widehat{E_{i1}}^{(k)}}{\sum_{i=1}^n \widehat{E_{i2}}^{(k)}-\frac{1}{n}\left(\sum_{i=1}^n\widehat{E_{i1}}^{(k)}\right)^2},~~~~
\widehat{{\boldsymbol{\xi}}}^{(k+1)}=\overline{{\boldsymbol{y}}}-\frac{{~\widehat{{\boldsymbol{\alpha}}}^{(k+1)}}}{n} \sum_{i=1}^n \widehat{E_{i1}}^{(k)},\\
\widehat{{\boldsymbol{\Sigma}}}^{(k+1)}_{\mathbf{Y}}&=&
\frac{1}{n}\sum_{i=1}^{n}\left({\boldsymbol{y}}_i- \widehat{{\boldsymbol{\xi}}}^{(k+1)}\right)\left({\boldsymbol{y}}_i-
\widehat{{\boldsymbol{\xi}}}^{(k+1)}\right)^\top-\frac{2}{n}\sum_{i=1}^n  \widehat{E_{i1}}^{(k)} \left(\boldsymbol{y}_i-\widehat{{\boldsymbol{\xi}}}^{(k+1)}\right){~\widehat{{\boldsymbol{\alpha}}}^{(k+1)}}^\top
+\frac{1}{n}{~\widehat{{\boldsymbol{\alpha}}}^{(k+1)}} {~\widehat{{\boldsymbol{\alpha}}}^{(k+1)}}^\top \sum_{i=1}^n \widehat{E_{i2}}^{(k)}.
\end{eqnarray*}
Therefore,  we can compute $\widehat{{\boldsymbol{\Omega}}}^{(k+1)}=\widehat{{\boldsymbol{\Sigma}}}^{(k+1)}_{\mathbf{Y}}+\widehat{{\boldsymbol{\alpha}}}^{(k+1)}{~\widehat{{\boldsymbol{\alpha}}}^{(k+1)}}^\top$
and ${~\widehat{{\boldsymbol{\delta}}}^{(k+1)}}={~\widehat{{\boldsymbol{\omega}}}^{(k+1)}}^{-1}{\widehat{{\boldsymbol{\alpha}}}^{(k+1)}}$, where
$\widehat{{\boldsymbol{\omega}}}=\left(\widehat{{\boldsymbol{\Omega}}}\odot \mathbf{I}_p\right)^{1/2}$.

\noindent \textbf{M-step 2}:  The update of $\widehat{{\boldsymbol{\nu}}}^{(k)}$ depends on the chosen distribution for $U$, and is obtained as
\begin{eqnarray*}
\widehat{{\boldsymbol{\nu}}}^{(k+1)}=\arg \max_{\boldsymbol{\nu}} \sum_{i=1}^n {\rm E}\left[\ln h(u_i;{\boldsymbol{\nu}})|{\mathbf{Y}}_i={\boldsymbol{y}}_i,\widehat{{\boldsymbol{\theta}}}^{(k)}\right].
\end{eqnarray*}
Updating   of   $\widehat{{\boldsymbol{\nu}}}^{(k)}$ is strongly related to the form of
$h(u_i;{\boldsymbol{\nu}})$.   If the conditional expectation
${\rm E} \left[\ln h(u_i; {\boldsymbol{\nu}}) | {\mathbf{Y}}_i={\boldsymbol{y}}_i, \widehat{{\boldsymbol{\theta}}}^{(k)}\right]$
 is difficult to evaluate, one may resort to
maximizing the restricted actual log-likelihood function, as follows:\\
\noindent \textbf{Modified M-step 2}: (Liu and Rubin \cite{Liu and  Rubin (1994)})
Update  $\widehat{{\boldsymbol{\nu}}}^{(k)}$ by
$
\widehat{{\boldsymbol{\nu}}}^{(k+1)}=\arg \max_{\boldsymbol{\nu}} \sum_{i=1}^n \ln f_{\mathcal{MMN}_p}\left({\boldsymbol{y}}_i;\widehat{{\boldsymbol{\xi}}}^{(k+1)}, \widehat{{\boldsymbol{\Omega}}}^{(k+1)}, \widehat{{\boldsymbol{\delta}}}^{(k+1)},{\boldsymbol{\nu}}\right)
$.\\
The above  algorithm iterates between
the E-step and M-step until  a suitable convergence criterion is satisfied.
 We adopt the distance involving two successive evaluations of the log-likelihood  function,
i.e.,   $\left|{\ell\left({\boldsymbol{\theta}}^{(k+1)}|{\boldsymbol{y}}\right)}/{\ell\left({\widehat{\boldsymbol{\theta}}}^{(k)}|{\boldsymbol{y}}\right)}-1\right|$,
as a convergence criterion, where
$
\ell({\boldsymbol{\theta}}|{\boldsymbol{y}})=\sum_{i=1}^n \ln f_{\mathcal{MMN}_p}
\left({\boldsymbol{y}}_i;{\boldsymbol{\xi}}, {\mathbf{\Omega}}, {\boldsymbol{\delta}}, {\boldsymbol{\nu}}\right)
$.

\section{Special Case of  $\mathcal{MMN}$ Distribution}\label{MMNG-sec}
In this section, we study in detail a special case of the $\mathcal{MMN}$ family.
In the stochastic representation in (\ref{Repre}), if the random variable $U$ follows the standard gamma
distribution with corresponding PDF $h(u; \nu)=u^{\nu-1} e^{-u}/\Gamma(\nu),~u>0$, we denote it by ${\mathbf{Y}}\sim \mathcal{MMNG}_p({\boldsymbol{\xi}}, {\mathbf{\Omega}}, {\boldsymbol{\delta}}, \nu)$. Then, the PDF of ${\mathbf{Y}}$ can be obtained from (\ref{density-y}) as follows:
\begin{eqnarray*}
f_{\mathcal{MMNG}_p}({\boldsymbol{y}})=\frac{\sqrt{2\pi}}{\eta^\nu\Gamma(\nu)}  \exp\left({\frac{A^2}{2}}\right) \phi_p({\boldsymbol{y}}; {\boldsymbol{\xi}}, {\mathbf{ \Sigma}}_{\mathbf{Y}}) \int_{-A}^{+\infty}(z+A)^{\nu-1}\phi(z) dz, ~~~~~{\mathbf{y}}\in {\mathbb{R}}^p,
\end{eqnarray*}where $\eta=\sqrt{{\boldsymbol{\delta}}^\top {\boldsymbol{\omega}} {\mathbf{ \Sigma}}_{\mathbf{Y}}^{-1}{\boldsymbol{\omega}}{\boldsymbol{\delta}}}$,  $A=\eta^{-1}\left[{\boldsymbol{\delta}}^\top {\boldsymbol{\omega}} {\mathbf{ \Sigma}}_{\mathbf{Y}}^{-1}({\boldsymbol{y}}-{\boldsymbol{\xi}}) -1 \right]$ and  ${\mathbf{ \Sigma}}_{\mathbf{Y}}={\mathbf{\Omega}}-{\boldsymbol{\omega}}{\mathbf{{\boldsymbol{\delta}}{\boldsymbol{\delta}}}^\top} {\boldsymbol{\omega}}$.
By using  the MGF  in (\ref{MtY}),  for ${\mathbf{Y}}\sim \mathcal{MMNG}_p({\boldsymbol{\xi}},{\mathbf{\Omega}},{\boldsymbol{\delta}},\nu)$, we  obtain
\begin{eqnarray*}
M_{\mathbf{Y}}({\mathbf{t}})=e^{{\mathbf{t}}^\top {\boldsymbol{\xi}}+  \frac{1}{2}{
{\mathbf{t}}^\top {\mathbf{ \Sigma}}_{\mathbf{Y}}  {\mathbf{t}}}}
 \left(1-{\mathbf{t}}^\top{\boldsymbol{\omega}}{\boldsymbol{\delta}}\right)^{-\nu},~~~~{\mathbf{t}}^\top{\boldsymbol{\omega}}{\boldsymbol{\delta}}\neq1,~~~\forall {\mathbf{t}}.
\end{eqnarray*}
From   the expressions in (\ref{M1X})-(\ref{M4Y}), and the fact that ${\rm E}(U^r)=\Gamma(\nu+r)/\Gamma(\nu)$, for positive constant $r$, we can compute the first four moments of ${\mathbf{Y}}$ by substituting
${\rm E}(U)=\nu$, ${\rm E}\left(U^2\right)=\nu(\nu+1)$, ${\rm E}\left(U^3\right)=\nu(\nu+1)(\nu+2)$, and ${\rm E}\left(U^4\right)=\nu(\nu+1)(\nu+2)(\nu+3)$. Specifically, we find   that
${\rm E}({\mathbf{Y}})={\boldsymbol{\xi}}+\nu{\boldsymbol{\omega}}{\boldsymbol{\delta}}$ and
$\textrm{var}\left({\mathbf{Y}}\right)={\mathbf{\Omega}}+(\nu-1){\boldsymbol{\omega}}{\boldsymbol{\delta}}{\boldsymbol{\delta}}^\top{\boldsymbol{\omega}}$.

\begin{figure}[htp]
 \centerline{ \includegraphics[scale=.7]{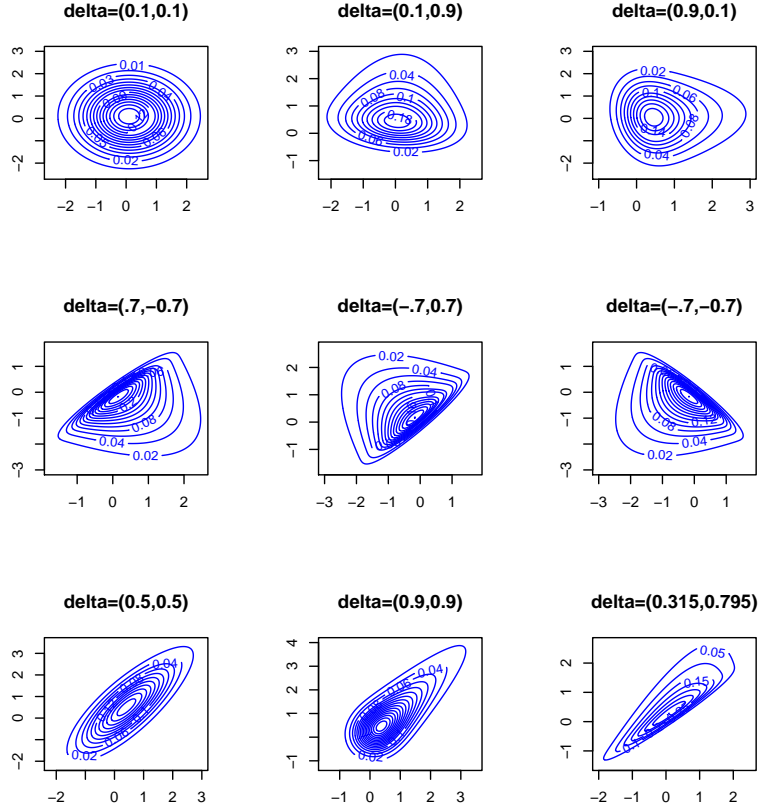}   }
\caption{Contour plots of $\mathcal{MMNE}_2$ distribution for different choices of ${\boldsymbol{\delta}}$.  For  the  first two rows, the scale matrix is ${\mathbf{\Omega}}=(1, 0; 0 ,1)$, while for the third row, it is  ${\mathbf{\Omega}}=(1, 1; 1 ,1.5)$.
 \label{PDFContour1}}
\end{figure}

\noindent\textbf{Definition 1.}  (Bose {\rm et al.} \cite{Bose et al. (2002)} ;  Steutel and  Van Harn \cite{Steutel and Van Harn (2004)}  )
A random vector ${\mathbf{Y}}$ (or its distribution) is said to be infinitely divisible if, for
each $n\geq1$, there exist independent and identically distributed (iid) random vectors
${\mathbf{Y}}_1,  \ldots, {\mathbf{Y}}_n$ such that ${\mathbf{Y}}\stackrel{d}{=}{\mathbf{Y}}_1+  \cdots + {\mathbf{Y}}_n$.

\begin{thm}\label{infinitely-divisible}
The $\mathcal{MMNG}$ distribution, in the multivariate case, is infinitely divisible.
\end{thm}
\begin{proof}[\bf Proof.]  Without loss of generality, let ${\mathbf{X}}\sim \mathcal{MMNG}_p\left({\boldsymbol{0}}, {\overline{\mathbf{\Omega}}}, {\boldsymbol{\delta}}, \nu\right)$ and  $ \textbf{X}_i\stackrel{d}{=} {\boldsymbol{\delta}} U_i + {\mathbf{Z}}_i$, where $U_i\sim  Gamma\left(\alpha=\frac{\nu}{n},\beta=1\right)$ and ${\mathbf{Z}}_i \sim \mathcal{N}_p\left(0,\frac{1}{n}\left({\overline{\mathbf{\Omega}}}-{\boldsymbol{\delta}}{\boldsymbol{\delta}}^\top\right)\right)$ be independent random variables. It is easy to show that $\sum_{i=1}^{n} U_i\sim  Gamma(\nu,1)$ and $\sum_{i=1}^{n} {\mathbf{Z}}_i\sim \mathcal{N}_p\left(0,{\overline{\mathbf{\Omega}}}-{\boldsymbol{\delta}}{\boldsymbol{\delta}}^\top\right)$, and so we can write
${\mathbf{X}}\stackrel{d}{=}{\mathbf{X}}_1+ \cdots+ {\mathbf{X}}_n$. Hence, the required result.
\end{proof}

In the following, a particular case of the $\mathcal{MMNG}$ distribution with $\nu=1$ is considered.
Upon  substituting  $\nu=1$, the mixing distribution of $U$ follows the standard exponential distribution and the distribution of  ${\mathbf{Y}}$ in this case is denoted by $\mathcal{MMNE}_p({\boldsymbol{\xi}},{\mathbf{\Omega}},{\boldsymbol{\delta}})$. Then, the PDF of ${\mathbf{Y}}$ can be obtained as
\begin{eqnarray*}
f_{\mathcal{MMNE}_p}({\boldsymbol{y}})=\frac{\sqrt{2\pi}}{\eta} \exp\left({\frac{A^2}{2}}\right) \phi_p({\boldsymbol{y}}; {\boldsymbol{\xi}}, {\mathbf{ \Sigma}}_{\mathbf{Y}})\Phi(A), ~~~~~{\mathbf{y}}\in {\mathbb{R}}^p.
\end{eqnarray*}
 Fig. \ref{PDFContour1}  presents the PDFs of the bivariate  $\mathcal{MMNE}$ distribution for
 ${\mathbf{\Omega}}=(1, 0; 0 ,1)$ and
 ${\mathbf{\Omega}}=(1, 1; 1 ,1.5)$,
and different choices of ${\boldsymbol{\delta}}$ for  ${\boldsymbol{\xi}}=(0,0)^\top$.
 Fig. \ref{PDFContour1}   shows that the $\mathcal{MMNE}$ distribution exhibits a wide variety of density shapes,  in terms of skewness. The  PDF of  the $\mathcal{MMNE}$ distribution clearly  depends on ${\mathbf{\Omega}}$ and  ${\boldsymbol{\delta}}$.
The following theorem is useful in the implementation of the EM algorithm for the
ML estimation of the parameters of the $\mathcal{MMNE}$ distribution.

\begin{thm}\label{Thm-conditional}
If ${\mathbf{Y}}\sim \mathcal{MMNE}_p({\boldsymbol{\xi}},{\mathbf{\Omega}},{\boldsymbol{\delta}})$ and the random variable $U$ follows the standard exponential distribution, then
$
U|\left({\mathbf{Y}}={\mathbf{y}}\right) \sim \mathcal{TN}\left(\eta^{-1}A,\eta^{-2},(0,\infty)\right).
$
Furthermore, for $k\in\{2, 3,  \ldots \}$,
\begin{eqnarray*}\label{EU|Y}
{\rm E}[U|{\mathbf{Y}}={\boldsymbol{y}}]=\eta^{-1}\left(A +\frac{\phi(A)}{\Phi(A)}\right),~~~~
{\rm E}[U^k|{\mathbf{Y}}={\boldsymbol{y}}]=A\eta^{-1} {\rm E}\left[U^{k-1}|{\mathbf{Y}}={\boldsymbol{y}}\right]+(k-1)\eta^{-2}
{\rm E}\left[U^{k-2}|{\mathbf{Y}}={\boldsymbol{y}}\right].
\end{eqnarray*}
\end{thm}
\begin{proof}[\bf Proof.]
The proof of the conditional distribution is  completed easily by the use of Bayes rule.
\end{proof}

Now, we can obtain the ML estimates of the parameters of $\mathcal{MMNE}$ distribution.  By using Theorem \ref{Thm-conditional} and  letting
\begin{eqnarray*}
\widehat{E_{i1}}^{(k)}={\rm E}\left[U_i|{\mathbf{Y}}_i={\boldsymbol{y}}_i,\widehat{{\boldsymbol{\theta}}}^{(k)}\right]=\frac{1}{\widehat{\eta}^{(k)}}\left(\widehat{A}_i^{(k)} +\frac{\phi\left(\widehat{A}_i^{(k)}\right)}{\Phi\left(\widehat{A}_i^{(k)}\right)}\right)\mbox{ and }
\widehat{E_{i2}}^{(k)}=
{\rm E}\left[U_i^2|{\mathbf{Y}}_i={\boldsymbol{y}}_i,\widehat{{\boldsymbol{\theta}}}^{(k)}\right]=\frac{1}{{~\widehat{\eta}^{(k)}}^2}
\left[{~\widehat{A}_i^{(k)}}^2 +\widehat{A}_i^{(k)}\frac{\phi\left(\widehat{A}_i^{(k)}\right)}{\Phi\left(\widehat{A}_i^{(k)}\right)}+1\right],
\end{eqnarray*}
in  expression  (\ref{Eihat-general}),  the EM algorithm for the $\mathcal{MMNE}$ distribution  can be performed,
where  $\widehat{\eta}^{(k)}= \sqrt{{~\widehat{{\boldsymbol{\alpha}}}^{(k)}}^\top  {~\widehat{\mathbf{\Sigma}}_{\mathbf{Y}}^{(k)}}^{-1} \widehat{{\boldsymbol{\alpha}}}^{(k)}}$    and     $\widehat{A}_i^{(k)}= {~\widehat{\eta}^{(k)}}^{-1} \left[{~\widehat{{\boldsymbol{\alpha}}}^{(k)}}^\top  {~\widehat{\mathbf{\Sigma}}_{\mathbf{Y}}^{(k)}}^{-1}\left({\boldsymbol{y}}_i-\widehat{{\boldsymbol{\xi}}}^{(k)}\right)   -1 \right]$.
Note that, in the case of $\mathcal{MMNE}$ distribution, the distribution of $U$ does not have any parameter,  and  so there is no need to estimate ${\boldsymbol{\nu}}$ in the EM algorithm and so M-step 2 must be skipped.

By using the fact that $E(U^m)=m!$, for $m\in\{1, 2, \ldots $\}, and by using expressions  in (\ref{M1Y})-(\ref{M4Y}),  for random vector ${\mathbf{Y}}\sim \mathcal{MMNE}_p({\boldsymbol{\xi}},{\mathbf{\Omega}},{\boldsymbol{\delta}})$, we have
\begin{eqnarray}
M_1({\mathbf{Y}})&=& {\boldsymbol{\xi}}+{\boldsymbol{\omega}}{\boldsymbol{\delta}},\label{M1Yexp}\\
M_2({\mathbf{Y}})&=& {\boldsymbol{\xi}}\otimes{\boldsymbol{\xi}}^\top+{\boldsymbol{\xi}}\otimes {\boldsymbol{\delta}}^\top{\boldsymbol{\omega}}+{\boldsymbol{\omega}}{\boldsymbol{\delta}}\otimes {\boldsymbol{\xi}}^\top+\left({\mathbf{ \Sigma}}_{\mathbf{Y}}+2{{\boldsymbol{\omega}}\boldsymbol{\delta}}{\boldsymbol{\delta}}^\top{\boldsymbol{\omega}}\right), \label{M2Yexp}\\
M_3({\mathbf{Y}})&=&{\boldsymbol{\xi}}{\boldsymbol{\xi}}^\top \otimes{\boldsymbol{\xi}}+{\boldsymbol{\xi}}{\boldsymbol{\xi}}^\top\otimes{\boldsymbol{\omega}}{\boldsymbol{\delta}}
+{\boldsymbol{\xi}}{\boldsymbol{\delta}}^\top{\boldsymbol{\omega}} \otimes{\boldsymbol{\xi}}
+{\boldsymbol{\omega}}{\boldsymbol{\delta}}\otimes{\boldsymbol{\xi}}{\boldsymbol{\xi}}^\top+ \left({\mathbf{ \Sigma}}_{\mathbf{Y}}+2{{\boldsymbol{\omega}}\boldsymbol{\delta}}{\boldsymbol{\delta}}^\top{\boldsymbol{\omega}} \right)\otimes{\boldsymbol{\xi}} + {\boldsymbol{\xi}}\otimes \left({\mathbf{ \Sigma}}_{\mathbf{Y}}+2{{\boldsymbol{\omega}}\boldsymbol{\delta}}{\boldsymbol{\delta}}^\top{\boldsymbol{\omega}} \right)\nonumber\\
&&+
\mathrm{vec}\left({\mathbf{ \Sigma}}_{\mathbf{Y}}+2{{\boldsymbol{\omega}}\boldsymbol{\delta}}{\boldsymbol{\delta}}^\top{\boldsymbol{\omega}} \right)\otimes{\boldsymbol{\xi}}^\top
+ {\boldsymbol{\omega}}{\mathbf{{\boldsymbol{\delta}}}}  \otimes{\mathbf{ \Sigma}}_{\mathbf{Y}}
+ \mathrm{vec}\left({\mathbf{ \Sigma}}_{\mathbf{Y}}\right){\boldsymbol{\delta}}^\top{\boldsymbol{\omega}}+\left({\mathbf{I}}_p\otimes {\boldsymbol{\omega}}{\boldsymbol{\delta}}\right)\left[ {\mathbf{ \Sigma}}_{\mathbf{Y}}+ 6{{\boldsymbol{\omega}}\boldsymbol{\delta}}{\boldsymbol{\delta}}^\top{\boldsymbol{\omega}}\right],  \label{M3Yexp}\\ 
M_4({\mathbf{Y}})&=& {\boldsymbol{\xi}}{\boldsymbol{\xi}}^\top \otimes{\boldsymbol{\xi}}{\boldsymbol{\xi}}^\top
+  {\boldsymbol{\xi}}{\boldsymbol{\xi}}^\top \otimes{\boldsymbol{\xi}}({\boldsymbol{\omega}}  {\boldsymbol{\delta}})^\top
+{\boldsymbol{\xi}}{\boldsymbol{\xi}}^\top \otimes ({\boldsymbol{\omega}}{\boldsymbol{\delta}} ){\boldsymbol{\xi}}^\top+ {\boldsymbol{\xi}}({\boldsymbol{\omega}}{\boldsymbol{\delta}} )^\top \otimes{\boldsymbol{\xi}}{\boldsymbol{\xi}}^\top
+{\boldsymbol{\xi}}{\boldsymbol{\xi}}^\top \otimes \left({\mathbf{ \Sigma}}_{\mathbf{Y}}+2{{\boldsymbol{\omega}}\boldsymbol{\delta}}{\boldsymbol{\delta}}^\top{\boldsymbol{\omega}}\right)\nonumber\\
&&+
({\boldsymbol{\xi}} \otimes{\boldsymbol{\xi}}) \left(\mathrm{vec}\left({\mathbf{ \Sigma}}_{\mathbf{Y}}+2{{\boldsymbol{\omega}}\boldsymbol{\delta}}{\boldsymbol{\delta}}^\top{\boldsymbol{\omega}} \right) \right)^\top
+{\boldsymbol{\xi}}  \otimes \left({\mathbf{ \Sigma}}_{\mathbf{Y}}+2{{\boldsymbol{\omega}}\boldsymbol{\delta}}{\boldsymbol{\delta}}^\top{\boldsymbol{\omega}} \right) \otimes{\boldsymbol{\xi}}^\top + {\boldsymbol{\xi}}  \otimes {\boldsymbol{\omega}}\left(M_3^{\mathbf{X}} \right)^\top\left({\boldsymbol{\omega}}\otimes{\boldsymbol{\omega}}\right)
+({\boldsymbol{\omega}}{\boldsymbol{\delta}} ){\boldsymbol{\xi}}^\top \otimes{\boldsymbol{\xi}}{\boldsymbol{\xi}}^\top
 \nonumber\\
&&+{\boldsymbol{\xi}}^\top  \otimes    \left({\mathbf{ \Sigma}}_{\mathbf{Y}}+2{{\boldsymbol{\omega}}\boldsymbol{\delta}}{\boldsymbol{\delta}}^\top{\boldsymbol{\omega}}\right)   \otimes {\boldsymbol{\xi}}
+{\boldsymbol{\xi}}^\top  \otimes \mathrm{vec}\left({\mathbf{ \Sigma}}_{\mathbf{Y}}+2{{\boldsymbol{\omega}}\boldsymbol{\delta}}{\boldsymbol{\delta}}^\top{\boldsymbol{\omega}}\right)   \otimes{\boldsymbol{\xi}}^\top
+{\boldsymbol{\xi}}^\top  \otimes ({\boldsymbol{\omega}}\otimes{\boldsymbol{\omega}}) M_3^{\mathbf{X}} {\boldsymbol{\omega}}
 \nonumber\\
&&+\left({\mathbf{ \Sigma}}_{\mathbf{Y}}+2{{\boldsymbol{\omega}}\boldsymbol{\delta}}{\boldsymbol{\delta}}^\top{\boldsymbol{\omega}} \right)   \otimes {\boldsymbol{\xi}}{\boldsymbol{\xi}}^\top+ {\boldsymbol{\omega}}  (M_3^{\mathbf{X}})^\top \left({\boldsymbol{\omega}}\otimes{\boldsymbol{\omega}}\right)\otimes{\boldsymbol{\xi}}
+  ({\boldsymbol{\omega}}\otimes{\boldsymbol{\omega}}) M_3^{\mathbf{X}} {\boldsymbol{\omega}}   \otimes{\boldsymbol{\xi}}^\top
+({\boldsymbol{\omega}}\otimes{\boldsymbol{\omega}}) M_4^{\mathbf{X}} ({\boldsymbol{\omega}}\otimes{\boldsymbol{\omega}}),\label{M4Yexp}
\end{eqnarray}
where
\begin{eqnarray*}
M_3^{\mathbf{X}}&=& {\boldsymbol{\delta}} \otimes {\mathbf{ \Sigma}}_{\mathbf{X}}+\mathrm{vec}({\mathbf{ \Sigma}}_{\mathbf{X}}){\boldsymbol{\delta}}^\top+ (\mathbf{I}_p\otimes {\boldsymbol{\delta}} ){\mathbf{ \Sigma}}_{\mathbf{X}} +  6 (\mathbf{I}_p\otimes {\boldsymbol{\delta}} )({\boldsymbol{\delta}}\otimes{\boldsymbol{\delta}}^\top),\\
M_4^{\mathbf{X}}&=&(\mathbf{I}_{p^2} +{\bf U}_{p,p})({\mathbf{ \Sigma}}_{\mathbf{X}}\otimes {\mathbf{ \Sigma}}_{\mathbf{X}})+
vec({\mathbf{ \Sigma}}_{\mathbf{X}})(\mathrm{vec}({\mathbf{ \Sigma}}_{\mathbf{X}}))^\top+ 2[
{\boldsymbol{\delta}}  \otimes {\boldsymbol{\delta}}^\top \otimes {\mathbf{ \Sigma}}_{\mathbf{X}}+ {\boldsymbol{\delta}}\otimes {\mathbf{ \Sigma}}_{\mathbf{X}}   \otimes  {\boldsymbol{\delta}}^\top+ {\mathbf{ \Sigma}}_{\mathbf{X}} \otimes {\boldsymbol{\delta}}\otimes{\boldsymbol{\delta}}^\top\nonumber \\
&&+  {\boldsymbol{\delta}}^\top\otimes {\mathbf{ \Sigma}}_{\mathbf{X}}   \otimes  {\boldsymbol{\delta}}+ {\boldsymbol{\delta}}^\top\otimes \mathrm{vec}({\mathbf{ \Sigma}}_{\mathbf{X}})   \otimes  {\boldsymbol{\delta}}^\top  +({\boldsymbol{\delta}}  \otimes {\boldsymbol{\delta}}) (\mathrm{vec}({\mathbf{ \Sigma}}_{\mathbf{X}}) )^\top ]+24 {\boldsymbol{\delta}}{\boldsymbol{\delta}}^\top\otimes {\boldsymbol{\delta}}{\boldsymbol{\delta}}^\top,
\end{eqnarray*}
and ${\mathbf{ \Sigma}}_{\mathbf{X}}={\overline{\mathbf{\Omega}}}-{\mathbf{{\boldsymbol{\delta}}{\boldsymbol{\delta}}}^\top}$.
In particular, the mean vector and covariance matrix are
$ {\rm E}({\mathbf{Y}})={\boldsymbol{\xi}}+{\boldsymbol{\omega}}{\boldsymbol{\delta}}$  and $\textrm{var}\left({\mathbf{Y}}\right)={\mathbf{\Omega}}$.

\begin{thm}
The $\mathcal{MMNE}$ distribution,  in the multivariate case, is log-concave.
\end{thm}

\begin{proof}[\bf Proof.] Because log-concavity is preserved by affine transformations, it is sufficient to prove this property for the canonical form $\mathbf{Z}^*\sim \mathcal{MMNE}_p({\boldsymbol{0}},{\mathbf{I}}_p, {\boldsymbol{\delta}}_{\textbf{Z}^*})$.
From \cite{An (1996)} and \cite{Preekopa (1973)}, if the elements of a random vector are independent, and each has a log-concave density function, then their joint density is log-concave. We know that in the canonical form with  PDF in (\ref{density}), the random variables $Z_1, \ldots, Z_p$ are independent of each other. Log-concavity of $\mathcal{MMNE}$ distribution in the univariate case has been established in Proposition 3.1 of \cite{Negarestani et al. (2019)}, and the PDF of the univariate normal distribution is also known to be log-concave. Hence, the result.
\end{proof}

As   shown in Section \ref{sec2}, to compute the mode of the $\mathcal{MMNE}$ distribution, it is sufficient to obtain the mode of the distribution in its canonical form, and then compute the mode of  the distribution using Theorem \ref{Thm5-mode}.
To compute the mode of the distribution in its canonical form, we must calculate the value of the mode in the univariate case. Existence and uniqueness of the mode (log-concavity) of the $\mathcal{MMNE}$ distribution in the univariate case has been discussed  in Proposition 3.1 of  \cite{Negarestani et al. (2019)}.
For this purpose, we recall  the density function of  the univariate $\mathcal{MMNE}$ distribution (given
in \cite{Negarestani et al. (2019)}) as
\begin{eqnarray*}
f_{Z_1}(z;\xi,\omega^2,\lambda)=\frac{\sqrt{1+\lambda^2}}{\omega|\lambda|}e^{-\frac{\sqrt{1+\lambda^2}}{\lambda}z+\frac{1}{2\lambda^2}} \Phi\left(\frac{\lambda\sqrt{1+\lambda^2}z-1}{|\lambda|}\right),
\end{eqnarray*}
where $z={(y-\xi)}/{\omega}$, $\lambda={\delta}/{\sqrt{1-\delta^2}}\neq0$, $y\in \mathbb{R}$,  $\xi\in \mathbb{R}$ is a location parameter and $\omega>0$ is a scale parameter. It is  denoted by $\mathcal{MMNE}_1(\xi,\omega^2,\lambda)$.  For obtaining the mode of $\mathcal{MMNE}_1$,  based on Theorem \ref{Thm5-mode}, we need to solve the equation
${\partial f_{Z_1^*}\left(z;0,1,\lambda_*\right)}/{\partial z}=0$, where $\lambda_*={\delta_*}/{\sqrt{1-{\delta_*}^2}}$. The solution  need to be obtained by  using numerical methods.

\section{Multivariate Measures of Skewness}\label{measures-skewness}
The skewed shape of   the distribution is usually captured  by multivariate skewness measures.
The skewness is a measure of the asymmetry of
a distribution about its mean and its value far from zero indicates stronger
asymmetry of the underlying distribution than that with close to zero skewness
value.
\begin{table}[htb!]
\scriptsize
\center
\caption{Multivariate measures of skewness for the $\mathcal{MMN}$ family.}
\label{skewformula}
\begin{tabular}{ll}
    \hline\\[-0.2cm]
Mardia \cite{Mardia (1970)}  \&   Malkovich and Afifi \cite{Malkovich and Afifi (1973)}&                         $\beta_{1,p}=(\gamma_1^*)^2$\\
Srivastava  \cite{Srivastava (1984)}&                     $\beta_{1p}^2=\frac{1}{p}\sum_{i=1}^{p} \left\{\frac{{\rm E}[{\boldsymbol{\gamma}}_i^\top({\mathbf{Y}}-{\boldsymbol{\mu}})]^3}{\lambda_i^{3/2}}\right\}^2$\\
M\'{o}ri-Rohatgi-Sz\'{e}kely  \cite{Mori et al. (1993)}&   ${\mathbf{s}}=\sum_{i=1}^{p}{\rm E}\left(Z_i^2{\mathbf{Z}}\right) =\left(\sum_{i=1}^{p}{\rm E}\left(Z_i^2 Z_1\right), \ldots,\sum_{i=1}^{p} {\rm E}\left(Z_i^2 Z_p\right) \right)^\top$\\\\
Kollo \cite{Kollo (2008)}&                          ${\mathbf{b}}={\rm E}\left(\sum_{i,j}^{p}(Z_i Z_j){\mathbf{Z}}\right)=\left(\sum_{i,j}^{p}E\left[(Z_i Z_j) Z_1\right], \ldots,\sum_{i,j}^{p}{\rm E}\left[(Z_i Z_j) Z_p\right]\right)^\top$\\\\
Balakrishnan-Brito-Quiroz  \cite{Balakrishnan et al. (2007)}&      $\mathbf{T}=\int_{{\phi}_p} {\mathbf{u}}c_1({\mathbf{u}})d \lambda({\mathbf{u}})$, $Q^*={\mathbf{T}}^\top {\mathbf{\Sigma}}_{\mathbf{Z}}^{-1} {\mathbf{T}}$\\
& The elements of ${\mathbf{T}}$ are $T_r=\frac{3}{p(p+2)} {\rm E}\left(Z_r^3\right)+ 3 \sum_{i\neq r} \frac{1}{p(p+2)} {\rm E}\left(Z_i^2 Z_r\right)$ \\[0.1cm]
Isogai \cite{Isogai (1982)} &                         $s_I=\frac{\left[{\mathbf{\delta}}_* {\rm E}(U)-m_0^* \right]^2}{1+{\mathbf{\delta}}_*^2[\textrm{var}(U)-1])}
$, $s_C=\left({\rm E}(U)-\frac{m_0^*}{{\mathbf{\delta}}_*}\right){\boldsymbol{\delta}}
$\\[0.1cm]
\hline
\end{tabular}
\end{table}
 The best-known scalar function of the vectorial measure of skewness proposed by \cite{Mori et al. (1993)} is its squared norm. Its sampling properties has been thoroughly discussed by \cite{Henze (1997)} and have been implemented in the \textsf{R} package \texttt{MultiSkew}. \cite{Franceschini and Loperfido (2019)} discusses  the usage of Multiskew and briefly review the literature related to the same squared norm. Also, the skewness measure in \cite{Malkovich and Afifi (1973)}  has become an useful tool in projection pursuit (\cite{Loperfido (2018)}).

In this work, multivariate measures of skewness by Mardia \cite{Mardia (1970)},  Malkovich and Afifi \cite{Malkovich and Afifi (1973)},
 Srivastava \cite{Srivastava (1984)}, M\'{o}ri {\rm et al.} \cite{Mori et al. (1993)},  Kollo \cite{Kollo (2008)},  Balakrishnan {\rm et al.} \cite{Balakrishnan et al. (2007)}   and Isogai \cite{Isogai (1982)}  are studied for the $\mathcal{MMN}$ family. Table \ref{skewformula} presents these measures for the $\mathcal{MMN}$ family of distributions. The relevant derivations are given in Appendix B.
In Table \ref{skewformula},  $\gamma_1^*$ is  the  skewness of  $Z_1^*\sim  \mathcal{MMN}_1(0,1,{\mathbf{\delta}}_*; H)$ of the canonical form, respectively.
Srivastava measure uses principal components  ${\mathbf{F}}={\mathbf{\Gamma}} {\mathbf{Y}}$,   where ${\mathbf{\Gamma}}=({\boldsymbol{\gamma}}_1, \ldots, {\boldsymbol{\gamma}}_p)$ is the matrix of eigenvectors of the covariance matrix ${\mathbf{\Delta}}$, that is, an orthogonal matrix such that ${\mathbf{\Gamma}}^\top{\mathbf{\Delta}}{\mathbf{\Gamma}}=\mathbf{\Lambda}$, and  $\mathbf{\Lambda}=\textrm{diag}(\lambda_1, \ldots, \lambda_p)$ is diagonal matrix of corresponding eigenvalues. Here, ${\mathbf{Z}}={\mathbf{\Delta}}^{-1/2}({\mathbf{Y}}-{\boldsymbol{\mu}})=(Z_1, \ldots, Z_p)^\top$ has the distribution $\mathcal{MMN}_p({\boldsymbol{\xi}}_{\mathbf{Z}},{\mathbf{\Omega}}_{\mathbf{Z}},{\boldsymbol{\delta}}_{\mathbf{Z}}; H)$, with its parameters  as
${\boldsymbol{\xi}}_{\mathbf{Z}}={\mathbf{\Delta}}^{-1/2}({\boldsymbol{\xi}}-{\boldsymbol{\mu}})$,
${\mathbf{\Omega}}_{\mathbf{Z}}={\mathbf{\Delta}}^{-1/2}{\mathbf{\Omega}} {\mathbf{\Delta}}^{-1/2}$,
${\boldsymbol{\delta}}_{\mathbf{Z}}={\boldsymbol{\omega}}_{\mathbf{Z}}^{-1} {\mathbf{\Delta}}^{-1/2}  {\boldsymbol{\omega}}{\boldsymbol{\delta}}$, and ${\boldsymbol{\omega}}_{\mathbf{Z}}=( {\mathbf{\Omega}}_{\mathbf{Z}}\odot \mathbf{I}_p)^{1/2}$.
Also, $m_0^*$ is the mode of the  scalar $\mathcal{MMN}$ distribution in the canonical form.
From Table \ref{skewformula}, and using the moments in  (\ref{M1Yexp})-(\ref{M4Yexp}), we can obtain different measures of skewness for the  $ \mathcal{MMNE}_p ({\boldsymbol{\xi}}, {\mathbf{\Omega}}, {\boldsymbol{\delta}})$  distribution as follows:
\begin{itemize}
\item \textit{\textbf{Mardia and Malkovich-Afifi indices}}:
 $\beta_{1,p}=\beta_{1}^*=4\delta_*^6$;

\item\textbf{\textit{ Srivastava index}}:
$
\beta_{1p}^2=\frac{1}{p}\sum_{i=1}^{p} \left\{\frac{{\rm E}[{\boldsymbol{\gamma}}_i^\top({\mathbf{Y}}-{\boldsymbol{\mu}})]^3}{\lambda_i^{3/2}}\right\}^2,
$
where ${\boldsymbol{\gamma}}_i$ and $\lambda_i$ are eigenvectors and corresponding eigenvalues for covariance matrix $\textrm{var}\left({\mathbf{Y}}\right)={\mathbf{\Omega}}$,  when ${\mathbf{Y}}\sim \mathcal{MMNE}_p({\boldsymbol{\xi}},{\mathbf{\Omega}},{\boldsymbol{\delta}})$;

\item \textit{\textbf{M\'{o}ri-Rohatgi-Sz\'{e}kely index}}:
If ${\mathbf{Y}}\sim \mathcal{MMNE}_p({\boldsymbol{\xi}},{\mathbf{\Omega}},{\boldsymbol{\delta}})$,  then for the standardized variable ${\mathbf{Z}}={\mathbf{\Omega}}^{-1/2}({\mathbf{Y}}-{\boldsymbol{\mu}})$, and with ${\mathbf{A}}={\mathbf{\Omega}}^{-1/2}$, we have (see  Appendix A)
\begin{eqnarray}\label{noncentralzzEE}
M_3({\mathbf{Z}})&=&{\rm E}\left[{\mathbf{A}}^\top({\mathbf{Y}}-{\boldsymbol{\mu}})\right]^3= \left({\mathbf{A}}^\top \otimes {\mathbf{A}}^\top\right) M_3({\mathbf{Y}}){\mathbf{A}}-\left[{\mathbf{A}}^\top M_2({\mathbf{Y}}) {\mathbf{A}} \right]\otimes \left[{\mathbf{A}}^\top {\rm E}({\mathbf{Y}})\right]
-{\mathbf{A}}^\top {\rm E}({\mathbf{Y}})\otimes \left[{\mathbf{A}}^\top M_2({\mathbf{Y}}) {\mathbf{A}} \right]\nonumber\\
&&-\mathrm{vec}\left({\mathbf{A}}^\top M_2({\mathbf{Y}}) {\mathbf{A}}\right){\rm E}({\mathbf{Y}})^\top {\mathbf{A}}
+2\left[{\mathbf{A}}^\top {\rm E}({\mathbf{Y}}){\rm E}({\mathbf{Y}})^\top {\mathbf{A}}\right]\otimes\left[{\mathbf{A}}^\top {\rm E}({\mathbf{Y}})\right].
\end{eqnarray}

All the quantities in the M\'{o}ri-Rohatgi-Sz\'{e}kely measure of skewness are specific non-central moments of third order of ${\mathbf{Z}}$, where ${\mathbf{Z}}={\mathbf{\Omega}}^{-1/2}({\mathbf{Y}}-{\boldsymbol{\mu}})\sim \mathcal{MMNE}_p({\boldsymbol{\xi}}_{\mathbf{Z}},{\mathbf{\Omega}}_{\mathbf{Z}},{\boldsymbol{\delta}}_{\mathbf{Z}})$, such that
${\boldsymbol{\xi}}_{\mathbf{Z}}=-{\mathbf{\Omega}}^{-1/2}{\boldsymbol{\omega}}{\boldsymbol{\delta}}$,
${\mathbf{\Omega}}_{\mathbf{Z}}=\textbf{I}_p$  and
${\boldsymbol{\delta}}_{\mathbf{Z}}= {\mathbf{\Omega}}^{-1/2}  {\boldsymbol{\omega}}{\boldsymbol{\delta}}$;

\item \textbf{\textit{Kollo index}}:
To  obtain  Kollo's  measure, we  use the elements of
 non-central moments of third order of ${\mathbf{Z}}$;

\item \textit{\textbf{Balakrishnan-Brito-Quiroz index}}:
Upon substituting $E(U^m)=m!$ for $m=1, 2, \ldots$, the elements of ${\mathbf{T}}$ in Table  \ref{skewformula}, for $r=1, 2, \ldots, p$,  are
$
{\mathbf{T}}_r=\frac{3}{p(p+2)} \left({\mathbf{M}}_3^{\mathbf{Z}}[(r - 1)p + r, r] +\sum_{i\neq r}  {\mathbf{M}}_3^{\mathbf{Z}}[(i - 1)p + i, r]\right),
$
where  ${\mathbf{M}}_3^{\mathbf{Z}}[., .]$ denotes the elements of matrix ${\mathbf{M}}_3^{\mathbf{Z}}$, third moments of  $\mathcal{MMNE}_p({\boldsymbol{\xi}}_{\mathbf{Z}},{\mathbf{\Omega}}_{\mathbf{Z}},{\boldsymbol{\delta}}_{\mathbf{Z}})$ distribution,
and  we can then compute ${\mathbf{Q}}={\mathbf{T}}^\top {\mathbf{\Sigma}}_{\mathbf{T}}^{-1}{\mathbf{T}}$ and  ${\mathbf{Q}}^*={\mathbf{T}}^\top {\mathbf{T}}$;

\item\textbf{\textit{ Isogai index}}:
 By substituting ${\rm E}(U)=\textrm{var}(U)=1$ in Isogai measure of skewness, we have
$S_I=\left({\mathbf{\delta}}_*-m_0^* \right)^2$,
where  $m_0^*$
is the mode of the  $\mathcal{MMNE}_1$ distribution in the canonical form. This index is location and scale invariant.
 The vectorial  measure, given by \cite{Balakrishnan and Scarpa (2012)}, is $S_C=\left(1-\frac{m_0^*}{{\mathbf{\delta}}_*}\right){\boldsymbol{\delta}}$.
Therefore, the direction of ${\boldsymbol{\delta}}$ can be regarded as a measure of vectorial skewness for the $\mathcal{MMNE}$ distribution.
\end{itemize}
\begin{table}[htb!]
\scriptsize
\center
\caption{The Average computational time spent (Atime), the average values (Mean), thecorresponding standard
deviations (Std.), Bias and MSE of the  EM estimates over 1000 samples from the $\mathcal{MMNE}$ model in Subsection \ref{Simul-EST}.   }
\label{tab-sim}
\begin{tabular}{cclccccccccc}
    \hline
n& Atime&& $\xi_1$ & $\xi_2$& $\xi_3$&$\delta_1$&$\delta_2$&$\delta_3$&$\sigma_{11}$&$\sigma_{22}$&$\sigma_{33}$\\
    \hline
$50$&0.3265&Mean&5.0804 &10.1518 &15.1735 &0.1663 &0.4940 &0.2259& 0.3947 &0.5931& 0.9800 \\
&&Std.&  0.2625 &0.3163 &0.4202 &0.3870 &0.3912 &0.4041 &0.0846 &0.1512 &0.2003 \\
&&Bias&0.0804& 0.1518& 0.1735& -0.1337 &-0.2060 &-0.1741& -0.0053 &-0.0069 &-0.0200 \\
&&MSE&0.0753 &0.1230 &0.2065 &0.1675 &0.1953 &0.1934 &0.0072 &0.0229 &0.0405 \\
$100$&0.4824&Mean&5.0409 &10.0554 &15.0637 &0.2360 &0.6221 &0.3346 &0.3979 &0.5973 &0.9871 \\
&&Std.&  0.1755 &0.2028 &0.2814 &0.2581 &0.2453 &0.2651 &0.0584 &0.1110 &0.1492 \\
&&Bias&0.0409 &0.0554 &0.0637 &-0.0640 &-0.0779 &-0.0654 &-0.0021& -0.0027& -0.0129 \\
&&MSE&0.0324 &0.0441 &0.0832 &0.0706 &0.0662 &0.0745 &0.0034 &0.0123 &0.0224 \\
$500$&1.8063&Mean&5.0006& 10.0036& 15.0024& 0.3001 &0.6973 &0.3964 &0.3993 &0.5995 &1.0004 \\
&&Std.&  0.0508 &0.0534 &0.0733 &0.0686 &0.0492 &0.0587 &0.0268 &0.0483 &0.0644 \\
&&Bias&0.0006 &0.0036 &0.0025 &0.0006 & -0.0027 &-0.0036 &0.0007 & 0.0005 & 0.0004\\
&&MSE&0.0026 &0.0029 &0.0054 &0.0047 &0.0024 &0.0035 &0.0007 &0.0023 &0.0041\\
$1000$&4.0388&Mean&5.0004 &9.9996 &15.0025 &0.3006 &0.6995 &0.3971 &0.3999 &0.5997 &0.9981 \\
&&Std.& 0.0345 &0.0355 &0.0529 &0.0435 &0.0328 &0.0424 &0.0180 &0.0336 &0.0447  \\
&&Bias&0.0004 &-0.0004 &0.0024 &0.0001 &-0.0005 &-0.0029 &-0.0001 &-0.0003 &-0.0002  \\
&&MSE&0.0012 &0.0013 &0.0028 &0.0019 &0.0011 &0.0018 &0.0003 &0.0011 &0.0020\\
\hline
\end{tabular}
\end{table}
\section{Simulation Study}
\subsection{Model Fitting}\label{Simul-EST}
This subsection presents the results of a Monte Carlo simulation
study carried out to examine the performance of the proposed estimation method for the $\mathcal{MMNE}$ distribution in the trivariate case.
We evaluate the estimates in terms of Bias and MSE (mean squared error). The results are  based on $1000$ simulated samples from the $\mathcal{MMNE}$ distribution with  parameters
$ {\boldsymbol{\xi}}=(5, 10, 15)^\top$, ${\mathbf{\Omega}}=\textrm{diag}(0.4, 0.6, 1.0)$,  ${\boldsymbol{\delta}}=(0.3, 0.7, 0.4)^\top$ for different sample sizes $n\in\{50, 100, 500, 1000\}$.  We computed the  Bias and the MSE as
$
\textrm{Bias}=\frac{1}{1000}\sum_{j=1}^{1000} (\widehat{\theta}_j-\theta)$ and
$\textrm{MSE}=\frac{1}{1000}\sum_{j=1}^{1000} (\widehat{\theta}_j-\theta)^2$,
where  $\theta$ is the true parameter (each of
${\boldsymbol{\xi}}=(\xi_1,\xi_2,\xi_3)^\top$,
${\boldsymbol{\delta}}=(\delta_1,\delta_2,\delta_3)^\top$  and
${\mathbf{\Omega}}=\textrm{diag}(\sigma_{11},\sigma_{22},\sigma_{33})^\top$)
and  $\widehat{\theta}_j$ is the estimate from the $j$-th simulated sample.
Table \ref{tab-sim} presents the  average computational time spent (Atime) (in seconds) (computational time for the convergence of the EM algorithm), average values (Mean), the corresponding
standard deviations (Std.),  Bias and  MSE  of the EM estimates of all the parameters of the $\mathcal{MMNE}$ model in 1000 simulated samples for each sample size.
It can be observed form Table \ref{tab-sim}
that the Bias and MSE decrease as $n$ increases, revealing  the asymptotic unbiasedness and
consistency of the ML estimates obtained through the EM algorithm.
Note that the EM algorithm presented in this work, for the $\mathcal{MMNE}$ distribution, leads to closed-form expressions, and so  the computational time  required for the convergence of the EM estimates of the parameters is quite short.
\begin{table}[htb!]
\scriptsize
\center
\caption{ Skewness measures in Section \ref{measures-skewness}, for some bivariate $\mathcal{MMNE}$ distributions with presumed parameter ${\boldsymbol{\xi}}=\boldsymbol{0}$ and different choices of parameters ${\boldsymbol{\Omega}}$ and  ${\boldsymbol{\delta}}$.}
\label{tab.skew2var}
\begin{tabular}{cclcccccccc}
\hline
$\#$&\multicolumn{2}{c}{Parameters}&$\beta_{1,p}$&$\beta_{1p}^2$& $s$  &  $b$& $Q^*$& $T$  &  $s_I$& $s_C$\\
\hline\\[-0.15cm]
 $1$&
 ${\boldsymbol{\Omega}}=\begin{bmatrix}   1& 1 \\  1&2.5 \\ \end{bmatrix}$&
 ${\boldsymbol{\delta}}=\begin{bmatrix} 0.750\\ 0.985 \\ \end{bmatrix}$&
 3.966& 1.975 &
 $\begin{bmatrix} 0.825 \\1.812 \\ \end{bmatrix}$
&
$\begin{bmatrix} 1.448 \\3.179\\ \end{bmatrix}$
&
0.558
&
$\begin{bmatrix} 0.310\\ 0.680\\ \end{bmatrix}$
&
 0.788
&
$\begin{bmatrix} 0.667\\ 0.876  \\ \end{bmatrix}$
 \\[0.2cm]
 \\[-0.15cm]
   $2$&
 ${\boldsymbol{\Omega}}=\begin{bmatrix}   1& 0 \\  0&2.5 \\ \end{bmatrix}$&
 ${\boldsymbol{\delta}}=\begin{bmatrix} 0.200\\ 0.975\\ \end{bmatrix}$&
 3.889& 1.718 &
 $\begin{bmatrix} 0.396\\ 1.932 \\ \end{bmatrix}$
&
$\begin{bmatrix} 0.552\\ 2.692\\ \end{bmatrix}$
&
0.547
&
$\begin{bmatrix} 0.149\\ 0.724\\ \end{bmatrix}$
&
0.673
&
$\begin{bmatrix}  0.165\\ 0.804  \\ \end{bmatrix}$
 \\[0.2cm]
 \\[-0.15cm]
     $3$&
 ${\boldsymbol{\Omega}}=\begin{bmatrix}   1& 0 \\  0&2.5 \\ \end{bmatrix}$&
 ${\boldsymbol{\delta}}=\begin{bmatrix} 0.000\\ 0.995\\ \end{bmatrix}$&
 3.881& 1.941 &
 $\begin{bmatrix} 0.000 \\1.970 \\ \end{bmatrix}$
&
$\begin{bmatrix} 0.000\\ 1.970\\ \end{bmatrix}$
&
 0.546
&
$\begin{bmatrix} 0.000\\ 0.739\\ \end{bmatrix}$
&
 0.666
&
$\begin{bmatrix}   0.000\\ 0.816  \\ \end{bmatrix}$
 \\[0.2cm]
 \\[-0.15cm]
     $4$&
 ${\boldsymbol{\Omega}}=\begin{bmatrix}   1& 1 \\ 1&2.5 \\ \end{bmatrix}$&
 ${\boldsymbol{\delta}}=\begin{bmatrix} 0.650\\ 0.995\\ \end{bmatrix}$&
3.890& 1.774 &
 $\begin{bmatrix} 0.562\\ 1.890 \\ \end{bmatrix}$
&
$\begin{bmatrix}  0.870\\ 2.924\\ \end{bmatrix}$
&
0.547
&
$\begin{bmatrix} 0.211 \\0.709\\ \end{bmatrix}$
&
 0.675
&
$\begin{bmatrix}   0.536\\ 0.821  \\ \end{bmatrix}$
 \\[0.2cm]
 \\[-0.15cm]
     $5$&
 ${\boldsymbol{\Omega}}=\begin{bmatrix}   1& 1 \\ 1&2.5 \\ \end{bmatrix}$&
  ${\boldsymbol{\delta}}=\begin{bmatrix} 0.850\\ 0.900\\ \end{bmatrix}$&
3.337 &1.511 &
 $\begin{bmatrix} 1.099\\ 1.460 \\ \end{bmatrix}$
&
$\begin{bmatrix}  2.154\\ 2.862 \\ \end{bmatrix}$
&
 0.469
&
$\begin{bmatrix}  0.412\\ 0.547\\ \end{bmatrix}$
&
 0.412
&
$\begin{bmatrix}   0.562\\ 0.595  \\ \end{bmatrix}$
 \\[0.2cm]
 \\[-0.15cm]
    $6$&
 ${\boldsymbol{\Omega}}=\begin{bmatrix}   1& 0 \\ 0&2.5 \\ \end{bmatrix}$&
  ${\boldsymbol{\delta}}=\begin{bmatrix} 0.550\\-0.800 \\ \end{bmatrix}$&
 3.349&  0.580 &
 $\begin{bmatrix} 1.037\\ -1.508 \\ \end{bmatrix}$
&
$\begin{bmatrix}  0.069\\ -0.100 \\ \end{bmatrix}$
&
0.471
&
$\begin{bmatrix}   0.389 \\-0.566 \\ \end{bmatrix}$
&
0.415
&
$\begin{bmatrix}   0.365\\ -0.531  \\ \end{bmatrix}$
 \\[0.2cm]
 \\[-0.15cm]
   $7$&
 ${\boldsymbol{\Omega}}=\begin{bmatrix}   1& 1 \\1&2.5 \\ \end{bmatrix}$&
 ${\boldsymbol{\delta}}=\begin{bmatrix} 0.900\\0.775 \\ \end{bmatrix}$&
 2.731& 0.843 &
 $\begin{bmatrix} 1.254\\ 1.077 \\ \end{bmatrix}$
&
$\begin{bmatrix}   2.493\\ 2.141\\ \end{bmatrix}$
&
0.384
&
$\begin{bmatrix}   0.470\\ 0.404\\ \end{bmatrix}$
&
 0.286
&
$\begin{bmatrix}  0.513\\ 0.442  \\ \end{bmatrix}$
 \\[0.2cm]
 \\[-0.15cm]
$8$&
 ${\boldsymbol{\Omega}}=\begin{bmatrix}   1& 0 \\0&2.5 \\ \end{bmatrix}$&
  ${\boldsymbol{\delta}}=\begin{bmatrix} 0.800\\0.400 \\ \end{bmatrix}$&
 2.048& 0.532 &
 $\begin{bmatrix} 1.280\\ 0.640 \\ \end{bmatrix}$
&
$\begin{bmatrix}    2.304\\ 1.152\\ \end{bmatrix}$
&
0.288
&
$\begin{bmatrix}   0.480\\ 0.240\\ \end{bmatrix}$
&
0.193
&
$\begin{bmatrix}   0.393\\ 0.196   \\ \end{bmatrix}$
  \\[0.2cm]
 \\[-0.15cm]
$9$&
 ${\boldsymbol{\Omega}}=\begin{bmatrix}   1& 1 \\1&2.5 \\ \end{bmatrix}$&
 ${\boldsymbol{\delta}}=\begin{bmatrix}  0.750\\ 0.150 \\ \end{bmatrix}$&
 1.607&  0.521 &
 $\begin{bmatrix} 1.263\\ -0.110 \\ \end{bmatrix}$
&
$\begin{bmatrix}    1.045\\ -0.091\\ \end{bmatrix}$
&
 0.226
&
$\begin{bmatrix}  0.473\\ -0.041\\ \end{bmatrix}$
&
0.1451
&
$\begin{bmatrix}    0.333 \\ 0.066   \\ \end{bmatrix}$
 \\[0.2cm]
 \\[-0.15cm]
$10$&
 ${\boldsymbol{\Omega}}=\begin{bmatrix}   1& 1 \\1&2.5 \\ \end{bmatrix}$&
 ${\boldsymbol{\delta}}=\begin{bmatrix}  -0.750\\ -0.150 \\ \end{bmatrix}$&
1.607 & 0.521 &
 $\begin{bmatrix} -1.263\\  0.110  \\ \end{bmatrix}$
&
$\begin{bmatrix}    -1.045\\  0.091\\ \end{bmatrix}$
&
0.226
&
$\begin{bmatrix}   -0.474 \\ 0.041\\ \end{bmatrix}$
&
0.145
&
$\begin{bmatrix}    -0.333\\ -0.067   \\ \end{bmatrix}$
 \\[0.2cm]
 \\[-0.15cm]
$11$&
 ${\boldsymbol{\Omega}}=\begin{bmatrix}   1& 0 \\0&2.5 \\ \end{bmatrix}$&
 ${\boldsymbol{\delta}}=\begin{bmatrix}  0.700\\ 0.000 \\ \end{bmatrix}$&
 0.471& 0.235  &
 $\begin{bmatrix} 0.686\\ 0.000  \\ \end{bmatrix}$
&
$\begin{bmatrix}    0.686\\ 0.000\\ \end{bmatrix}$
&
0.066
&
$\begin{bmatrix}   0.257\\ 0.000\\ \end{bmatrix}$
&
 0.043
&
$\begin{bmatrix}   0.208\\ 0.000   \\ \end{bmatrix}$
 \\[0.2cm]
 \\[-0.15cm]
$12$&
 ${\boldsymbol{\Omega}}=\begin{bmatrix}   1& -1 \\-1&2.5 \\ \end{bmatrix}$&
  ${\boldsymbol{\delta}}=\begin{bmatrix}  0.000\\ 0.000 \\ \end{bmatrix}$&
 0.000& 0.000  &
 $\begin{bmatrix} 0.000\\ 0.000  \\ \end{bmatrix}$
&
$\begin{bmatrix}    0.000\\ 0.000\\ \end{bmatrix}$
&
0.000
&
$\begin{bmatrix}   0.000\\ 0.000\\ \end{bmatrix}$
&
0.000
&
$\begin{bmatrix}   0.000\\ 0.000   \\ \end{bmatrix}$
 \\[0.2cm]
  \hline
\end{tabular}
\end{table}
\subsection{Assessment of Skewness}
To study and compare  different multivariate measures of skewness for the $\mathcal{MMN}$ distributions,
we consider the  $\mathcal{MMNE}$ distribution. We compute the values of all the skewness
measures for different choices of the parameters of the bivariate and trivariate $\mathcal{MMNE}$  distributions. Tables \ref{tab.skew2var} and \ref{tab.skew3var}  present the values of all the skewness measures. It should be noted that all the measures are location and scale invariant, a desirable property indeed for any measure of skewness.  For similar work on skewness comparisons for $\mathcal{SN}$ distribution, one may refer to \cite{Balakrishnan and Scarpa (2012)}, and also to \cite{Kim and Zhao (2018)} for a similar work on scale mixtures of $\mathcal{SN}$ distributions.
From Table  \ref{tab.skew2var}, we find that in all  cases with scalar measures of skewness,   Mardia's  measure  have the highest value and  Srivastava's measure is  the next largest.
Just as in the case of  $\mathcal{SN}$ distribution, for the bivariate  $\mathcal{MMNE}$ distribution, the vectorial measures yield very similar results in terms of skewness directions, especially when the distribution is highly asymmetric \cite{Balakrishnan and Scarpa (2012)}.
It is important to note that Cases 9 and 10 deal with reflected distributions; and in these cases, all the measures
are the same and the vectorial ones are reflected as well.
Table \ref{tab.skew3var}  presents the values of all the measures for the   trivariate $\mathcal{MMNE}$ distribution.  In this case, differences among the
measures become much more pronounced.  From Table  \ref{tab.skew3var}, we find that in all  cases, among the vectorial measures of skewness,  Mardia's  measure has the highest value.
Of course, the magnitude of the measures alone does not say much; one has to know how significant the values are!
\begin{table}[htb!]
\scriptsize
\center
\caption{
 Skewness measures in Section \ref{measures-skewness}, for some trivariate $\mathcal{MMNE}$ distributions with presumed parameter ${\boldsymbol{\xi}}=\boldsymbol{0}$ and different choices of parameters ${\boldsymbol{\Omega}}$ and  ${\boldsymbol{\delta}}$.}
\label{tab.skew3var}
\begin{tabular}{cclcccccccc}
\hline
$\#$&\multicolumn{2}{c}{Parameters}&$\beta_{1,p}$&$\beta_{1p}^2$& $s$  &  $b$& $Q^*$& $T$  &  $s_I$& $s_C$\\
\hline\\[-0.15cm]
 $1$&
 ${\boldsymbol{\Omega}}=\begin{bmatrix}   1& 0& 0\\ 0& 2.5&0\\ 0& 0& 2.5 \\ \end{bmatrix}$&
 ${\boldsymbol{\delta}}=\begin{bmatrix} 0.10\\ 0.70\\ 0.70 \\ \end{bmatrix}$&
 3.881& 0.314 &
 $\begin{bmatrix} 0.198 \\1.386\\ 1.386 \\ \end{bmatrix}$
&
$\begin{bmatrix} 0.450 \\3.150\\ 3.150 \\ \end{bmatrix}$
&
0.155
&
$\begin{bmatrix} 0.040\\ 0.277\\ 0.277\\ \end{bmatrix}$
&
 0.666
&
$\begin{bmatrix}  0.082\\ 0.574\\ 0.574  \\ \end{bmatrix}$
\\[0.2cm]
 \\[-0.15cm]
   $2$&
 ${\boldsymbol{\Omega}}=\begin{bmatrix}   1& 1& 1\\ 1& 2.5&1\\ 1&1& 10 \\ \end{bmatrix}$&
 ${\boldsymbol{\delta}}=\begin{bmatrix} 0.75\\0.75\\0.65\\ \end{bmatrix}$&
2.712&  0.235  &
 $\begin{bmatrix} 0.752\\   1.044\\   1.028 \\ \end{bmatrix}$
&
$\begin{bmatrix} 2.211\\   3.070 \\  3.023\\ \end{bmatrix}$
&
0.108
&
$\begin{bmatrix} 0.150 \\  0.209 \\  0.206\\ \end{bmatrix}$
&
0.283
&
$\begin{bmatrix}   0.426\\   0.426 \\  0.369  \\ \end{bmatrix}$
  \\[0.2cm]
 \\[-0.15cm]
    $3$&
 ${\boldsymbol{\Omega}}=\begin{bmatrix}   1& 0& 0\\ 0& 2.5&0\\ 0& 0& 5 \\ \end{bmatrix}$&
 ${\boldsymbol{\delta}}=\begin{bmatrix}0.995\\ 0.00\\0.00\\ \end{bmatrix}$&
3.881& 1.294  &
 $\begin{bmatrix} 1.970\\ 0.000\\ 0.000 \\ \end{bmatrix}$
&
$\begin{bmatrix} 1.970\\ 0.000\\ 0.000 \\ \end{bmatrix}$
&
0.155
&
$\begin{bmatrix} 0.394 \\0.000\\ 0.000\\ \end{bmatrix}$
&
 0.666
&
$\begin{bmatrix}  0.816\\ 0.000\\ 0.000  \\ \end{bmatrix}$
 \\[0.2cm]
 \\[-0.15cm]
     $4$&
 ${\boldsymbol{\Omega}}=\begin{bmatrix}   1& 0& 0\\ 0& 1&0\\ 0& 0& 2.5 \\ \end{bmatrix}$&
 ${\boldsymbol{\delta}}=\begin{bmatrix} 0.40\\-0.60\\-0.60\\ \end{bmatrix}$&
2.726&  0.130   &
 $\begin{bmatrix} 0.704\\ -1.056 \\ -1.056 \\ \end{bmatrix}$
&
$\begin{bmatrix}  0.512\\ -0.768\\ -0.768\\ \end{bmatrix}$
&
0.109
&
$\begin{bmatrix} 0.141\\ -0.211\\ -0.211 \\ \end{bmatrix}$
&
0.286
&
$\begin{bmatrix}   0.228\\ -0.342\\ -0.342  \\ \end{bmatrix}$
  \\[0.2cm]
 \\[-0.15cm]
     $5$&
 ${\boldsymbol{\Omega}}=\begin{bmatrix}  1& 1& 1\\ 1& 2.5&1\\ 1&1& 10 \\ \end{bmatrix}$&
  ${\boldsymbol{\delta}}=\begin{bmatrix} 0.55\\0.05\\-0.30\\ \end{bmatrix}$&
1.372  &0.223  &
 $\begin{bmatrix} 1.055\\ -0.140\\ -0.490\\ \end{bmatrix}$
&
$\begin{bmatrix} 0.139\\ -0.018\\ -0.064 \\ \end{bmatrix}$
&
 0.055
&
$\begin{bmatrix}  0.211\\ -0.028\\ -0.098\\ \end{bmatrix}$
&
0.122
&
$\begin{bmatrix}   0.230\\  0.021\\ -0.125  \\ \end{bmatrix}$
 \\[0.2cm]
 \\[-0.15cm]
    $6$&
 ${\boldsymbol{\Omega}}=\begin{bmatrix}  1& 0& 0\\ 0& 2.5&0\\ 0&0& 1 \\ \end{bmatrix}$&
 ${\boldsymbol{\delta}}=\begin{bmatrix} 0.75\\0.35\\0.35 \\ \end{bmatrix}$&
2.106& 0.242  &
 $\begin{bmatrix} 1.211 \\0.565 \\0.565\\ \end{bmatrix}$
&
$\begin{bmatrix}  3.154\\ 1.472\\ 1.472  \\ \end{bmatrix}$
&
0.084
&
$\begin{bmatrix}  0.242 \\0.113\\ 0.113\\ \end{bmatrix}$
&
 0.200
&
$\begin{bmatrix}   0.373 \\0.174 \\0.174  \\ \end{bmatrix}$
 \\[0.2cm]
  \hline
\end{tabular}
\end{table}
\begin{table}[htb!]
\scriptsize\center
\caption{Upper and lower $2.5\%$ critical  values based on the test statistics $\beta_{1,p}$, $\beta_{1p}^2$, $s_{{\rm sum}}$, $s_{{\rm max}}$, $b_{\rm sum}$, $b_{\rm max}$, $Q^*$,  $\boldsymbol{T}_{{\rm sum}}$, $\boldsymbol{T}_{{\rm max}}$, $s_I$,  ${s_C}_{{\rm sum}}$, and ${s_C}_{{\rm max}}$, in Subsection \ref{Power}, obtained from $10000$ simulated samples of standard multivariate normal distribution,  for sample size $n=100$ and dimensions $p\in\{2, \ldots, 8\}$. }
\label{25Percentiles}
\begin{tabular}{c|crrrrrrrrrrrrrr}
\hline
Test Statistics& Percentile &$p=2~$&$p=3~$&$p=4~$&$p=5~$ & $p=6~$&$p=7~$&$p=8~$\\
\hline
$\beta_{1,p}$&    0.025  &0.0000 &0.0000&0.0685&0.1232&0.2304& 0.3170  &0.4388                           \\
                    &	   0.975	 &1.1816&1.4552&2.1717 &2.8814&3.5455& 3.8266 &3.9085 			\\
$\beta_{1p}^2$& 0.025  &0.0000&0.0000&0.0025  &0.0038&0.0041&0.0035 & 0.0042               \\
                    &	   0.975  &0.4474&0.2918&0.2827 &0.2906&0.3167& 0.2311	 &0.2654   				\\
$s_{max}$ &        0.025  &-0.3562& -0.1626&-0.0975&-0.0244&0.0160&0.1061 &0.1452    \\
                    &		0.975 &	0.8894&0.9881&1.1032&1.2195&1.2737& 1.3150 &	1.4936 			\\
$s_{sum}$&          0.025 &-0.9945&-1.3716&-1.6814& -1.9739&-2.2822&-2.5103 &-2.7708                 \\
                    &		0.975 &1.0052&	1.3275&1.4690&1.9442&2.1993&2.5009	 &3.1478			\\
$b_{max}$&         0.025  & -0.6804&-0.3989&-0.2855&-0.0707&0.0000&0.0002 &0.0018                    \\
                    &		0.975  & 1.0957&1.6014&1.7865	&2.4971&2.6128&3.1300  &4.0002                           			\\
$b_{sum}$&         0.025   &-1.7299&-2.8449&-3.9202&-5.5741&-6.4435&-7.3387   &-8.0038                      \\
                    &		0.975  &1.7621&2.8721&3.3774 & 5.6679&6.3422&7.8015 &10.2697				\\
$Q^*$&                0.025  & 0.0000&0.0000&0.0011&0.0009& 0.0009&0.0007   &  0.0006            \\
                    &		0.975    &0.1662&0.0582&0.0339&0.0212&	0.0138&0.0087	 &0.0055			\\
$\boldsymbol{T}_{max}$ &  0.025 &-0.1336&-0.0325&-0.0122&-0.0021&0.0010&0.0051 &0.0054 \\
                    &	                  0.975	&0.3335&0.1976&0.1379&0.1045& 0.0796&0.0626 &0.0560				\\
$\boldsymbol{T}_{sum}$ &  0.025 &-0.3729&-0.2743&-0.2102&-0.1692& -0.1426&-0.1195 &-0.1039      \\
                    &		               0.975 &0.3769&0.2655&0.1836&0.1666&0.1375&0.1191  &0.1180 					\\
${s_I}$&                              0.025 &0.0000& 0.0000&0.0080&0.0133&0.0229&0.0304  &0.0407 \\
                    &		               0.975 &0.1045&0.1301&0.2077&0.3122&0.4770&0.6192 &0.6956					\\
${s_C}_{max}$&                  0.025 &-0.1130&-0.0536&-0.0347&-0.0108&0.0109&0.0311 &0.0482   \\
                    &		               0.975 &0.2710&0.2955&0.3410&0.4033& 0.4830&0.5097 &0.6410  					\\
${s_C}_{sum}$&                  0.025 & -0.2920&-0.4116& -0.5417&-0.6276&-0.7303&-0.8984 &-1.0657    \\
                    &		               0.975 &0.3036&0.3932&0.4644&0.6154&0.7311&0.8275 &1.1957					\\
  \hline
\end{tabular}
\end{table}
\begin{table}[htb!]
\scriptsize\center
\caption{Upper $5\%$ critical values based on the test statistics $\beta_{1,p}$, $\beta_{1p}^2$, $s_{{\rm sum}}$, $s_{{\rm max}}$, $b_{\rm sum}$, $b_{\rm max}$, $Q^*$,  $\boldsymbol{T}_{{\rm sum}}$, $\boldsymbol{T}_{{\rm max}}$, $s_I$,  ${s_C}_{{\rm sum}}$, and ${s_C}_{{\rm max}}$,  in Subsection \ref{Power}, obtained from $10000$ simulated samples of standard multivariate normal distribution, for  sample size $n=100$ and dimensions $p\in\{2, \ldots, 8\}$. }
\label{5Percentiles}
\begin{tabular}{c|rrrrrrrrrrrrrr}
\hline
Test Statistics& $p=2~$&$p=3~$&$p=4~$&$p=5~$ & $p=6~$&$p=7~$&$p=8~$\\
\hline
$\beta_{1,p}$ &0.9325&1.2317&1.7024&2.4021&3.0324& 3.4899 &	 3.8939			\\
$\beta_{1p}^2$ &0.3111&0.2135&0.1997&0.1812&0.2044& 0.1612 &	0.2039				\\
$s_{max}$ &        0.7679 &0.8177&0.9712&1.0931&1.2011& 1.1774 &1.3446					\\
$s_{sum}$&         0.8546	&1.0509&1.2424&1.6682& 1.8934& 1.9996 & 2.4041 				\\
$b_{max}$&         0.9619&1.2206&1.3468& 1.9457&2.0467& 2.2711 & 2.8953 				\\
$b_{sum}$&         1.3001&2.1093& 2.5046&3.9367&4.4979&5.1949 &7.1686   				\\
$Q^*$&              0.1311&	0.0493&0.0266&0.0176&0.0118&0.0079 &0.0055   				\\
$\boldsymbol{T}_{max}$ &0.2880 &0.1635&0.1214&0.0937&0.0751&0.0561 &0.0504 				\\
$\boldsymbol{T}_{sum}$ &  0.3205	&0.2102&0.1553&0.1430&0.1183& 0.0952 &0.0902  				\\
${s_I}$&                              0.0824	&0.1091&0.1549& 0.2375&0.3409&0.4576	 & 0.6792 			\\
${s_C}_{max}$&                  0.2324	  &0.2468&0.2959& 0.3382&0.4069&	 0.4282 &0.5509			\\
${s_C}_{sum}$&                  0.2619	  &0.3224&0.3757&0.5072&0.6228 &0.6621 &	0.8116		\\
  \hline
\end{tabular}
\end{table}
\begin{table}[htb!]
\scriptsize
\center
\caption{ Simulated values of power for all tests based on the test statistics $\beta_{1,p}$, $\beta_{1p}^2$, $s_{{\rm sum}}$, $s_{{\rm max}}$, $b_{\rm sum}$, $b_{\rm max}$, $Q^*$,  $\boldsymbol{T}_{{\rm sum}}$, $\boldsymbol{T}_{{\rm max}}$, $s_I$,  ${s_C}_{{\rm sum}}$, and ${s_C}_{{\rm max}}$, in Subsection \ref{Power}, for bivariate normal distribution against $\mathcal{MMNE}$ distribution.}
\label{power-2var}
\begin{tabular}{cclcccccccccccc}
\hline
$\#$&\multicolumn{2}{c}{Parameters}&$\beta_{1,p}$&$\beta_{1p}^2$& $s_{{\rm max}}$  & $s_{\rm sum}$  &  $b_{{\rm max}}$  & $b_{\rm sum}$  & $Q^*$& $\boldsymbol{T}_{{\rm max}}$  & $\boldsymbol{T}_{\rm sum}$ & $s_I$   &  ${s_C}_{{\rm max}}$  & ${s_C}_{\rm sum}$\\
\hline\\[-0.2cm]
 $1$&
  ${\boldsymbol{\Omega}}=\begin{bmatrix}   1& 0 \\ 0&2.5 \\ \end{bmatrix}$&
   ${\boldsymbol{\delta}}=\begin{bmatrix}0.1\\ 0.1 \\ \end{bmatrix}$&
  0.030& 0.041& 0.034& 0.040& 0.038& 0.035& 0.030& 0.034& 0.040& 0.030& 0.030& 0.044\\[0.3cm]
 $2$&&  ${\boldsymbol{\delta}}=\begin{bmatrix}0.5\\ 0.5 \\ \end{bmatrix}$&
 0.283& 0.121& 0.172& 0.467& 0.492& 0.497& 0.283& 0.172& 0.467& 0.283& 0.167& 0.459\\[0.3cm]
   $3$&&  ${\boldsymbol{\delta}}=\begin{bmatrix}0.1\\ 0.8 \\ \end{bmatrix}$&
0.659& 0.748& 0.711& 0.668& 0.631& 0.349& 0.659 &0.711& 0.668& 0.659& 0.690& 0.634\\[0.3cm]
  $4$&&  ${\boldsymbol{\delta}}=\begin{bmatrix}0.8\\ 0.1 \\ \end{bmatrix}$&
0.682& 0.724 &0.729& 0.700& 0.661& 0.383& 0.682& 0.729& 0.700& 0.682& 0.711& 0.682\\[0.3cm]
  $5$&&  ${\boldsymbol{\delta}}=\begin{bmatrix}0.8\\ 0.8 \\ \end{bmatrix}$&
0.994& 0.994& 0.994& 0.994& 0.994& 0.994& 0.994& 0.994& 0.994& 0.994& 0.994& 0.994\\[0.3cm]
  $6$&&  ${\boldsymbol{\delta}}=\begin{bmatrix}-0.7\\ -0.7 \\ \end{bmatrix}$&
  0.982& 0.982& 0.982& 0.981& 0.981& 0.981& 0.982& 0.982& 0.981& 0.982& 0.982& 0.981\\[0.3cm]
  \hline\\[-0.2cm]
$7$&
  ${\boldsymbol{\Omega}}=\begin{bmatrix}   1&1 \\1&2.5 \\ \end{bmatrix}$& ${\boldsymbol{\delta}}=\begin{bmatrix}0.1\\ 0.1 \\ \end{bmatrix}$&
  0.032& 0.035& 0.043& 0.057& 0.060& 0.059& 0.032& 0.043& 0.057& 0.032& 0.087& 0.140  \\[0.3cm]
$8$&& ${\boldsymbol{\delta}}=\begin{bmatrix}0.5\\ 0.5 \\ \end{bmatrix}$&
 0.074& 0.079& 0.062& 0.159& 0.189& 0.186& 0.074& 0.062& 0.159& 0.074& 0.088& 0.315\\[0.3cm]
$9$&&  ${\boldsymbol{\delta}}=\begin{bmatrix}0.1\\ 0.8 \\ \end{bmatrix}$&
0.979& 0.929& 0.980& 0.668& 0.105& 0.004& 0.979& 0.980& 0.668& 0.979& 0.974& 0.956 \\[0.3cm]
$10$&&  ${\boldsymbol{\delta}}=\begin{bmatrix}0.8\\ 0.1\\ \end{bmatrix}$&
  0.982& 0.978& 0.984& 0.941& 0.639& 0.084& 0.982& 0.984& 0.941& 0.982& 0.977& 0.963\\[0.3cm]
 $11 $&& ${\boldsymbol{\delta}}=\begin{bmatrix}0.8\\ -0.1\\ \end{bmatrix}$&
 0.956& 0.958& 0.959& 0.816& 0.076& 0.002& 0.956& 0.959& 0.816& 0.956& 0.960& 0.938\\[0.3cm]
 $12 $&& ${\boldsymbol{\delta}}=\begin{bmatrix}-0.8\\ 0.1\\ \end{bmatrix}$&
 0.953& 0.955& 0.023& 0.832& 0.003& 0.005& 0.953& 0.023& 0.832& 0.953& 0.024& 0.936\\[0.3cm]
 $13 $&& ${\boldsymbol{\delta}}=\begin{bmatrix}-0.8\\ -0.1\\ \end{bmatrix}$&
 0.968& 0.963& 0.004& 0.921& 0.004& 0.088& 0.968& 0.004& 0.921& 0.968& 0.248& 0.944\\[0.3cm]
 $14 $&& ${\boldsymbol{\delta}}=\begin{bmatrix}0.8\\ 0.8\\ \end{bmatrix}$&
 0.929& 0.929& 0.815& 0.977& 0.979& 0.980& 0.929& 0.815& 0.977& 0.929& 0.918& 0.987\\[0.3cm]
  \hline
\end{tabular}
\end{table}
\begin{table}[htb!]
\scriptsize
\center
\caption{ Simulated values of power for  all tests based on the test statistics $\beta_{1,p}$, $\beta_{1p}^2$, $s_{{\rm sum}}$, $s_{{\rm max}}$, $b_{\rm sum}$, $b_{\rm max}$, $Q^*$,  $\boldsymbol{T}_{{\rm sum}}$, $\boldsymbol{T}_{{\rm max}}$, $s_I$,  ${s_C}_{{\rm sum}}$, and ${s_C}_{{\rm max}}$, in Subsection \ref{Power}, for trivariate normal distribution against $\mathcal{MMNE}$ distribution.}
\label{power-3var}
\begin{tabular}{ccccccccccccccc}
\hline
$\#$&\multicolumn{2}{c}{Parameters}&$\beta_{1,p}$&$\beta_{1p}^2$& $s_{\rm max}$  & $s_{\rm sum}$  &  $b_{\rm max}$  & $b_{\rm sum}$  & $Q^*$& $\boldsymbol{T}_{\rm max}$  & $\boldsymbol{T}_{\rm sum}$ & $s_I$   &  ${s_C}_{\rm max}$  & ${s_C}_{\rm sum}$\\
\hline\\[-0.2cm]
 $1$&
 ${\boldsymbol{\Omega}}=\begin{bmatrix}   1&0&0\\0&2.5&0\\0&0&2.5 \\ \end{bmatrix}$& ${\boldsymbol{\delta}}=\begin{bmatrix}0.1\\ 0.1 \\ 0.1\\ \end{bmatrix}$&
0.056& 0.058& 0.047 &0.036& 0.042& 0.035& 0.056& 0.047& 0.036& 0.056& 0.053& 0.038\\[0.3cm]
 $2$&& ${\boldsymbol{\delta}}=\begin{bmatrix} -0.1\\ -0.1 \\ -0.1\\ \end{bmatrix}$&
 0.057& 0.061& 0.057& 0.035& 0.048& 0.052& 0.057& 0.057& 0.035& 0.057& 0.056& 0.039\\[0.3cm]
 \\[-0.2cm]
  $3$&
 ${\boldsymbol{\Omega}}=\begin{bmatrix}   1&1&1\\1&2.5&1\\1&1&10 \\ \end{bmatrix}$& ${\boldsymbol{\delta}}=\begin{bmatrix}0.1\\ 0.7 \\ 0.1\\ \end{bmatrix}$&
0.601& 0.192& 0.637& 0.166& 0.082& 0.031& 0.601& 0.637& 0.166& 0.601& 0.557& 0.513\\[0.3cm]
 \\[-0.2cm]
  $4$&
 ${\boldsymbol{\Omega}}=\begin{bmatrix}    1&0&0\\0&2.5&0\\0&0&5 \\ \end{bmatrix}$&
 ${\boldsymbol{\delta}}=\begin{bmatrix}0.7\\ 0.1 \\ 0.1\\ \end{bmatrix}$&
0.272& 0.324& 0.309& 0.277& 0.262& 0.168& 0.272& 0.309& 0.277& 0.272& 0.308& 0.287\\[0.3cm]
 \\[-0.2cm]
 $5$&
 ${\boldsymbol{\Omega}}=\begin{bmatrix}   1&0&0\\0&1&0\\0&0&2.5 \\ \end{bmatrix}$& ${\boldsymbol{\delta}}=\begin{bmatrix}0.4\\ 0.6 \\ 0.6\\ \end{bmatrix}$&
0.899& 0.513& 0.816& 0.902& 0.901& 0.902& 0.899& 0.816& 0.902& 0.899& 0.815& 0.901\\
[0.3cm] \\[-0.2cm]
  $6$&
 ${\boldsymbol{\Omega}}=\begin{bmatrix}   1&1&1\\1&2.5&1\\1&1&10 \\ \end{bmatrix}$& ${\boldsymbol{\delta}}=\begin{bmatrix}0.7\\ 0.05 \\ 0.3\\ \end{bmatrix}$&
0.745& 0.674 &0.750& 0.501& 0.277& 0.114& 0.745 &0.750& 0.501& 0.745& 0.638& 0.709\\
[0.3cm] \\[-0.2cm]
   $7$&
 ${\boldsymbol{\Omega}}=\begin{bmatrix}    1&0&0\\0&2.5&0\\0&0&1 \\ \end{bmatrix}$& ${\boldsymbol{\delta}}=\begin{bmatrix}0.7\\ 0.3 \\ 0.3\\ \end{bmatrix}$&
0.572& 0.349& 0.461& 0.726& 0.729& 0.703& 0.572& 0.461& 0.726& 0.572& 0.476& 0.711\\[0.3cm]
 $8$&
  & ${\boldsymbol{\delta}}=\begin{bmatrix} -0.7\\ 0.3 \\ 0.1\\ \end{bmatrix}$&
0.381& 0.228& 0.025& 0.025& 0.015& 0.015& 0.380 &0.025 &0.025& 0.380& 0.029& 0.029\\[0.3cm]
$9$&
 & ${\boldsymbol{\delta}}=\begin{bmatrix} 0.7\\ -0.3 \\ 0.1\\ \end{bmatrix}$&
0.415& 0.248& 0.400& 0.061& 0.034& 0.013& 0.415& 0.401 &0.061 &0.415& 0.414 &0.081\\[0.3cm]
$10$&
 & ${\boldsymbol{\delta}}=\begin{bmatrix}0.7\\ 0.3 \\  -0.1\\ \end{bmatrix}$&
0.394 &0.240& 0.378& 0.344 &0.259 &0.152& 0.392 &0.378& 0.344 &0.394& 0.385 &0.362\\[0.3cm]
$11$&
 & ${\boldsymbol{\delta}}=\begin{bmatrix}-0.7\\ -0.3 \\  0.1\\ \end{bmatrix}$&
0.404& 0.206& 0.073 &0.304& 0.059& 0.146 &0.404 &0.073& 0.304 &0.404 &0.072& 0.317\\[0.3cm]
$12$&
  & ${\boldsymbol{\delta}}=\begin{bmatrix}-0.7\\ 0.3 \\  -0.1\\ \end{bmatrix}$&
0.415& 0.222& 0.034 &0.050 &0.015 &0.022 &0.415 &0.034 &0.050& 0.415& 0.044& 0.069\\[0.3cm]
$13$&
  & ${\boldsymbol{\delta}}=\begin{bmatrix} 0.7\\ -0.3 \\  -0.1\\ \end{bmatrix}$&
0.393& 0.257& 0.356& 0.031& 0.020& 0.015& 0.393 &0.356& 0.031& 0.393& 0.370& 0.037\\[0.3cm]
$14$&
 & ${\boldsymbol{\delta}}=\begin{bmatrix}-0.8\\ -0.8 \\  -0.8\\ \end{bmatrix}$&
0.996& 0.965& 0.997& 0.993 &0.993& 0.993& 0.996& 0.997 &0.993 &0.996& 0.997& 0.993\\[0.3cm]
  \hline
\end{tabular}
\end{table}
\subsection{Comparison and Performance of Different Skewness Measures}\label{Power}
The measures studied in Section \ref{measures-skewness} and in the preceding subsection are not directly comparable with each other. So, for comparing them, we should have measures obtained on the same scale. To get such a set of comparable indices, we study the sample version for each of the skewness measures considered as test statistics for the hypothesis of normal distribution against $\mathcal{MMNE}$ distribution,  and then use the power of test based on different test statistics.
Let $\boldsymbol{Y}_1, \boldsymbol{Y}_2, \ldots, \boldsymbol{Y}_n$ denote a sample of $p\times1$ observations from any  $p$-dimensional distribution. Then, a sample version of all the skewness measures described can be obtained by replacing $\boldsymbol{\xi}$, $\boldsymbol{\Omega}$, and $\boldsymbol{\delta}$ with the maximum likelihood estimates of these quantities \cite{Balakrishnan and Scarpa (2012)}.
 As seen in the previous sections, the Mardia and Malkovich-Afifi measure $\beta_{1,p}$, Srivastava measure $\beta_{1p}^2$, Isogai measure $s_I$, Balakrishnan-Brito-Quiroz measure $Q^*$ are scalar indices and the M\'{o}ri-Rohatgi-Sz\'{e}kely measure $s$, Kollo measure $b$, Balakrishnan-Brito-Quiroz measure $\boldsymbol{T}$, and Isogai measure $s_C$ are vectorial indices.
Here, we study different statistics for testing the null hypothesis and powers for each of these tests to quantify the capacity of each skewness measure to identify the specific asymmetry present in the $\mathcal{MMNE}$ distribution.
The power of the test is a probability, and its use  enables us to compare different statistics, no matter what the original scales of them were.
To obtain a single test statistic for the vectorial measures, we propose two different metrics, namely, the sum and the maximum  (see \cite{Balakrishnan and Scarpa (2012)}, pages 82-83).
For the M\'{o}ri-Rohatgi-Sz\'{e}kely measure, we compute $s_{{\rm sum}}=\sum_{r=1}^p s_r$ and $s_{{\rm max}}=\max_{r\in \{1, \ldots ,p\}} s_r$,
for the Kollo measure $b_{{\rm sum}}=\sum_{r=1}^p b_r$ and $b_{{\rm max}}=\max_{r\in \{1, \ldots ,p\}} b_r$,
for the Balakrishnan-Brito-Quiroz measure $\boldsymbol{T}_{{\rm sum}}=\sum_{r=1}^p T_r$ and $\boldsymbol{T}_{{\rm max}}=\max_{r\in \{1, \ldots ,p\}} T_r$,
for Isogai's measure ${s_C}_{{\rm sum}}= \sum_{r=1}^p{s_C}_r$ and ${s_C}_{{\rm max}}= \max_{r\in \{1, \ldots ,p\}} {s_C}_r $.
The distributions of sample versions  of measures, $\beta_{1,p}$, $\beta_{1p}^2$, $s_{{\rm sum}}$, $s_{{\rm max}}$, $b_{\rm sum}$, $b_{\rm max}$, $Q^*$,  $\boldsymbol{T}_{{\rm sum}}$, $\boldsymbol{T}_{{\rm max}}$, $s_I$,  ${s_C}_{{\rm sum}}$, and ${s_C}_{{\rm max}}$ are not analytically computable easily,
and so we may determine the critical values of  these  tests through Monte Carlo simulations. Two sets of  critical values obtained by Monte Carlo simulation, based on $10000$ samples from the standard multivariate normal distribution, are  presented in Tables \ref{25Percentiles} and \ref{5Percentiles}, for dimensions $p\in\{2, \ldots, 8\}$.
To get the values of critical values,  we first simulated  $10000$ samples of size $n=100$ from the standard multivariate normal distribution with dimensions $p\in\{2, \ldots, 8\}$. We  estimated  the parameters and then  found the values of test statistics.   Then,  we  arranged  the obtained values in increasing order and then  selected  the $2.5$ and  $5$ lower and upper  percentage points as critical values.

For computing the powers of the different tests, based on the above test statistics, we simulated  1000 samples from $\mathcal{MMNE}$ distribution of size $n=100$ for different choices of the parameters of the $\mathcal{MMNE}$ distribution and  estimated the test statistics by  using the ML estimates of parameters evaluated by EM algorithm. Then, we  computed the proportion  of samples falling in the same rejection region.
For test statistics $\beta_{1,p}$, $\beta_{1p}^2$,  $Q^*$  and $s_I$,
we  considered the sample versions exceeding the critical values as critical regions, in the form $CR=\{Q_0>q_\alpha\}$, and
for all other test statistics, the rejection regions  were the  two-sided areas  of  the form  $CR=\{Q_0<q_{1-\alpha/2} ~\mbox{or}~ Q_0>q_{\alpha/2}\}$, where  $Q_0$ is test  statistic  under  null hypothesis and $q_\alpha$ is upper $\alpha$ percentile of distribution of test  statistic.
\begin{table}[htb!]
\scriptsize
\center
\caption{ Simulated values of power for  all tests  based on the test statistics $\beta_{1,p}$, $\beta_{1p}^2$, $s_{{\rm sum}}$, $s_{{\rm max}}$, $b_{\rm sum}$, $b_{\rm max}$, $Q^*$,  $\boldsymbol{T}_{{\rm sum}}$, $\boldsymbol{T}_{{\rm max}}$, $s_I$,  ${s_C}_{{\rm sum}}$, and ${s_C}_{{\rm max}}$, in Subsection \ref{Power}, for seven dimensional  normal distribution against $\mathcal{MMNE}$ distribution when  ${\boldsymbol{\Omega}}=\mathbf{I}_7$.}
\label{power-7var}
\begin{tabular}{cccccccccccccc}
\hline
$\#$&Parameter&$\beta_{1,p}$&$\beta_{1p}^2$& $s_{\rm max}$  & $s_{\rm sum}$  &  $b_{\rm max}$  & $b_{\rm sum}$  & $Q^*$& $\boldsymbol{T}_{\rm max}$  & $\boldsymbol{T}_{\rm sum}$ & $s_I$   &  ${s_C}_{\rm max}$  & ${s_C}_{\rm sum}$\\
\hline
 $1$&
 ${\boldsymbol{\delta}}=(0.1,0.1,0.1,0.1,0.1,0.1,0.1)^\top$&
 0.061& 0.060& 0.061& 0.038& 0.042& 0.046& 0.062& 0.061& 0.044& 0.061& 0.059& 0.045\\
  $2$&
 ${\boldsymbol{\delta}}=(0.7,0.7,0.7,0.7,0.7,0.7,0.7)^\top$&
 0.978& 0.985& 0.000 &0.985 &0.985& 0.985& 0.978& 0.000& 0.985& 0.978& 0.704& 0.985\\
   $3$&
 ${\boldsymbol{\delta}}=(0.1,0.7,0.1,0.7,0.1,0.7,0.1)^\top$&
 0.952& 0.543& 0.035& 0.952 &0.953 &0.951& 0.952 &0.035& 0.952& 0.952 &0.580& 0.954\\
   $4$&
 ${\boldsymbol{\delta}}=(0.4,0.2,0.5,0.1,0.7,0.6,0.3)^\top$&
 0.918 &0.357 &0.033& 0.923 &0.927 &0.923 &0.919 &0.034 &0.923 &0.918 &0.326& 0.925\\
   $5$&
 ${\boldsymbol{\delta}}=-(0.1,0.1,0.1,0.1,0.1,0.1,0.1)^\top$&
 0.052& 0.055& 0.093& 0.058& 0.047 &0.070 &0.052& 0.095& 0.060& 0.052& 0.075 &0.054\\
   $6$&
 ${\boldsymbol{\delta}}=-(0.7,0.7,0.7,0.7,0.7,0.7,0.7)^\top$&
 0.980& 0.982 &0.986 &0.982 &0.983 &0.982 &0.980 &0.986 &0.982 &0.980& 0.983 &0.982\\
   $7$&
 ${\boldsymbol{\delta}}=(0.1,-0.7,0.1,-0.7,0.1,-0.7,0.1)^\top$&
 0.959 &0.516 &0.010& 0.684& 0.000 &0.080& 0.959 &0.010 &0.771& 0.959& 0.007& 0.841\\
   $8$&
 ${\boldsymbol{\delta}}=(-0.4,0.2,-0.5,0.1,-0.7,0.6,-0.3)^\top$&
 0.889& 0.413& 0.015 &0.009& 0.006 &0.005 &0.891& 0.015& 0.014& 0.889& 0.102& 0.156\\
  \hline
\end{tabular}
\end{table}

In the simulation study, we   took  $\boldsymbol{\xi}= \boldsymbol{0}$,  and  the parameters $\boldsymbol{\Omega}$ and $\boldsymbol{\delta}$  as  given in   Tables \ref{power-2var}-\ref{power-7var}.
Tables \ref{power-2var}-\ref{power-7var}  present the power of the proposed tests for bivariate,  trivariate and seven dimensional normal distribution against MMNE distribution, respectively.
The comparison of  different  measures may be done directly from the results in  Tables  \ref{power-2var}-\ref{power-7var}.
These results show clearly which are the poorer indices of skewness among those considered.
Based on our empirical study, by considering different cases of the $\mathcal{MMNE}$ distributions in two, three and seven dimensions, we
 make  the following points:
for all cases with small skewness, as expected, the power of the tests are lower for distributions more similar to the normal, and test statistics $\beta_{1,p}$, $\beta_{1p}^2$,  $Q^*$ and $s_I$ have better performance.
From Tables  \ref{power-2var}-\ref{power-7var}, as expected, for increasing values ​​of the elements of the skewness parameters,  the power values of all  tests  increase.
 For large  elements close to 1 or -1​, for skewness parameters, the power of the tests are  higher and have almost   the same values for different test statistics.

The behaviour of test statistics $\beta_{1,p}$, $Q^*$ and $s_I$, are very close to each other and have the same power. For small values of the skewness parameter,  these test statistics have poorer performance.
The power of the tests for  $s_{{\rm max}}$,  and $\boldsymbol{T}_{{\rm max}}$ statistics are the same, and  the test statistics $s_{\rm sum}$ and $\boldsymbol{T}_{\rm sum}$  often have similar behaviour.
For large and  moderate values ​​of skewness parameters,  $b_{{\rm sum}}$ and $b_{{\rm max}}$    statistics have the lowest test power and have the worst performance compared with other test statistics.
For the bivariate case in Table \ref{power-2var}, when one element of the skewness parameter is large and one is small, the statistics $\beta_{1p}^2$, $\boldsymbol{T}_{{\rm max}}$ and $s_{{\rm max}}$ perform well, but    $b_{{\rm sum}}$ and $b_{{\rm max}}$    statistics have the lowest test power.

For the trivariate case in Table \ref{power-3var}, when one element of the skewness parameter is large and two elements and small, the statistics $\beta_{1p}^2$, $\boldsymbol{T}_{{\rm max}}$ and $s_{{\rm max}}$ have better performance.

From Table \ref{power-7var},  for case 3, the statistic ${s_C}_{{\rm sum}}$   has  the best performance  and  $\boldsymbol{T}_{{\rm max}}$ and $s_{{\rm max}}$ have  lower power close to  $0.05$.
From Table \ref{power-7var},  for case 4, the statistic $b_{{\rm max}}$ has the best performance, but $\boldsymbol{T}_{{\rm max}}$ and $s_{{\rm max}}$ have lower power close to $0.05$.
A result  that we  find from Tables  \ref{power-2var}-\ref{power-7var} is that  the test statistic ${s_C}_{{\rm max}}$ performs better than others in many cases.

\section{Illustrative Examples}\label{Example}
In this section, we fit the $\mathcal{MMNE}$ model for two real data sets to
illustrate the flexibility of the model. It is also compared  with $\mathcal{SN}$ and $\mathcal{ST}$ distributions in terms of some measures of fit.

\subsection{AIS data}\label{AIS-Example}
The first example considers the Australian
Institute of Sport (AIS) data  \cite{Cook and Weisberg (1994)},
containing 11 biomedical measurements on 202 Australian
athletes (100 female and 102 male).
Here, we focus solely on the first 100, and the trivariate case corresponding to BMI, SSF and Bfat variables, where the three acronyms denote Body Mass Index,  Sum of Skin Folds, and Body Fat percentage, respectively. These data are available in the \textsf{R} software, \texttt{sn} package.
Fig. \ref{Hist-AIS} presents the histograms  for the three variables.
Upon using the  EM algorithm, we  obtained the  maximum likelihood estimates of parameters of the model. Table \ref{tab_est_AIS} presents the estimates of parameters $({\boldsymbol{\xi}},{\mathbf{\Omega}},{\boldsymbol{\delta}})$.
Table \ref{tab_skew_AIS} presents values of all skewness  measures   by using the estimates of parameters, presented in Table \ref{tab_est_AIS}.
\begin{figure}[htp]
 \centerline{ \includegraphics[scale=.45]{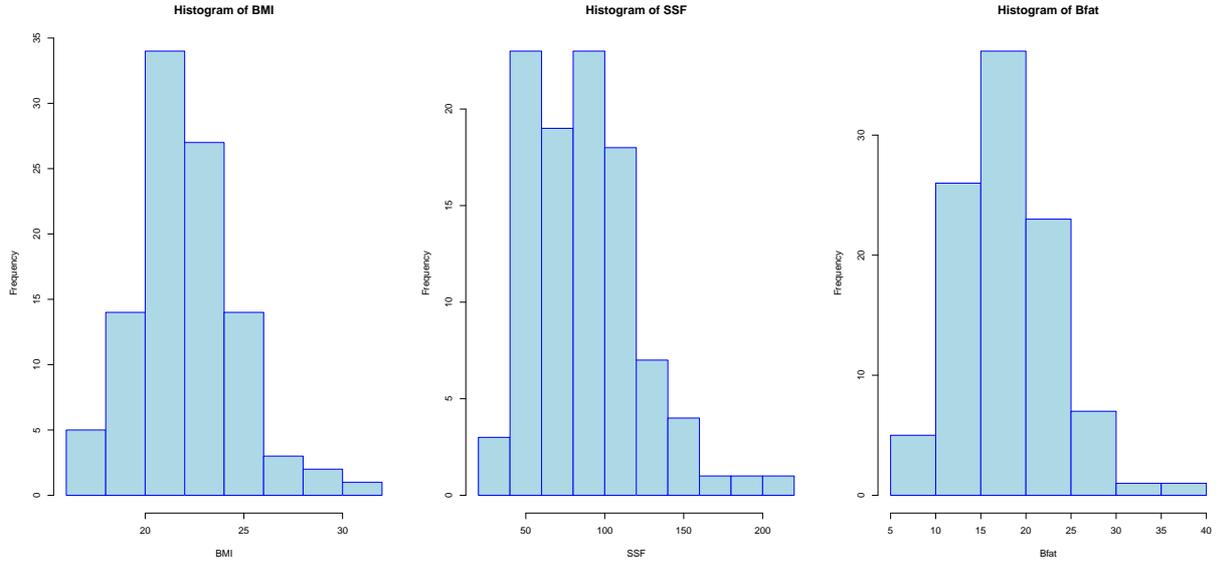} }
\caption{ The  histograms  for the three selected variables BMI, SSF and Bfat of  the AIS data set in Subsection \ref{AIS-Example}.   \label{Hist-AIS}}
\end{figure}
The relative difference in the fit of a number of
candidate models can be compared by using the  maximized log-likelihood values  $ \ell({\boldsymbol{\hat{\theta}}}|{\boldsymbol{y}}) $,  the Akaike information criterion (AIC) and the Bayesian information criterion (BIC). The AIC and BIC indices are defined as
$
AIC = 2k-\ell\left({\boldsymbol{\hat{\theta}}}|{\boldsymbol{y}}\right)$
and
$BIC = k \ln n-2\ell\left({\boldsymbol{\hat{\theta}}}|{\boldsymbol{y}}\right),
$
where $k$ is the number of model parameters and $\ell\left({\boldsymbol{\hat{\theta}}}|{\boldsymbol{y}}\right)$ is the   maximized log-likelihood value of a fitted model.
The larger value of $\ell\left({\boldsymbol{\hat{\theta}}}|{\boldsymbol{y}}\right)$ and the smaller value of AIC or BIC indicates a better fit of the model to the data.

Table \ref{tab_AIC_fit} summarizes the fitting performance of $\mathcal{MMNE}$ model, as compared to the $\mathcal{SN}$ and $\mathcal{ST}$ distributions. From Table \ref{tab_AIC_fit}, it is seen that the $\mathcal{MMNE}$ model provides the best fit overall  as it provides the   largest $\ell({\boldsymbol{\hat{\theta}}}|{\boldsymbol{y}})$ value and
the lowest AIC and BIC  values. Fig. \ref{scatter-AIS} shows the scatter plots of pairs of the three  variables BMI, SSF, Bfat,  along with   the contour plots for the  fitted $\mathcal{MMNE}$, $\mathcal{SN}$ and $\mathcal{ST}$ distributions.\\

\begin{table}[htb!]
\small
\center
\caption{Parameter estimates of the $\mathcal{MMNE}$ distribution  by using EM Algorithm presented in Section \ref{Est-sec}, based on the three selected variables (BMI, SSF, Bfat) of the AIS  data set in Subsection \ref{AIS-Example}.}
\label{tab_est_AIS}
\begin{tabular}{ccc}
    \hline
$\widehat{{\boldsymbol{\xi}}}$&$\widehat{{\boldsymbol{\Omega}}}$& $\widehat{{\boldsymbol{\delta}}}$\\
\hline\\[-0.2cm]
$
\begin{bmatrix}
  20.1099 \\
  56.1969  \\
  13.6666 \\
\end{bmatrix}
$
&
$
\begin{bmatrix}
7.2870&   66.3650 & 10.2745\\
66.3650& 1238.1858& 191.2748\\
10.2745&  191.2748&  31.3535\\
 \end{bmatrix}
$
&
$
\begin{bmatrix}
0.6963\\
0.8747\\
0.7471\\
\end{bmatrix}
$\\[0.2cm]
 \hline
\end{tabular}
\end{table}

\begin{table}[htb!]
\small
\center
\caption{Values of skewness measures in Section \ref{measures-skewness}, based on the three selected variables (BMI, SSF, Bfat) of the AIS  data set in Subsection \ref{AIS-Example}.}
\label{tab_skew_AIS}
\begin{tabular}{cccccccc}
    \hline
$\beta_{1,p}$&$\beta_{1p}^2$& $s$  &  $b$& $Q^*$& $T$ &  $s_I$& $s_C$\\
\hline\\[-0.2cm]
3.5539 & 0.5973  &
$
\begin{bmatrix}
  0.3182\\
  1.7703\\
  -0.5644\\
\end{bmatrix}
$
&
$
\begin{bmatrix}
 0.2080\\
 1.1571\\
 -0.3689\\
\end{bmatrix}
$
&
0.1422
&
$
\begin{bmatrix}
  0.0636\\
  0.3541  \\
  -0.1129 \\
\end{bmatrix}
$
&
0.4800
&
$
\begin{bmatrix}
  0.4920\\
  0.6181  \\
  0.5279 \\
\end{bmatrix}
$
 \\[0.2cm]
 \hline
 \end{tabular}
\end{table}

\begin{table}[htb!]
\small
\center
\caption{Comparison of fitting measures, maximized log-likelihood value $\ell({\boldsymbol{\hat{\theta}}}|{\boldsymbol{y}})$, Akaike information criterion AIC and Bayesian information criterion BIC, for skew-normal ($\mathcal{SN}$), skew-t ($\mathcal{ST}$) and $\mathcal{MMNE}$ distributions for the three selected variables (BMI, SSF, Bfat) of the  AIS data set in Subsection \ref{AIS-Example}.}
\label{tab_AIC_fit}
\begin{tabular}{l|ccc }
      \hline
Distribution &  ~~~~~~~ $\ell({\boldsymbol{\hat{\theta}}}|{\boldsymbol{y}}) $  &    ~~~~~~~  AIC&    ~~~~~~~  BIC\\
\hline
$\mathcal{SN}$   &  ~~~~~~~ -866.2725 & ~~~~~~~   1756.545&  ~~~~~~~ 1787.807     \\
$\mathcal{ST}$   &  ~~~~~~~ -852.1354 &   ~~~~~~~  1730.271&   ~~~~~~~ 1764.138     \\
$\mathcal{MMNE}$   & ~~~~~~~  -850.7388 &  ~~~~~~~  1725.478& ~~~~~~~ 1756.740     \\
\hline
\end{tabular}
\end{table}

\begin{figure}[htp]
 \centerline{ \includegraphics[scale=1]{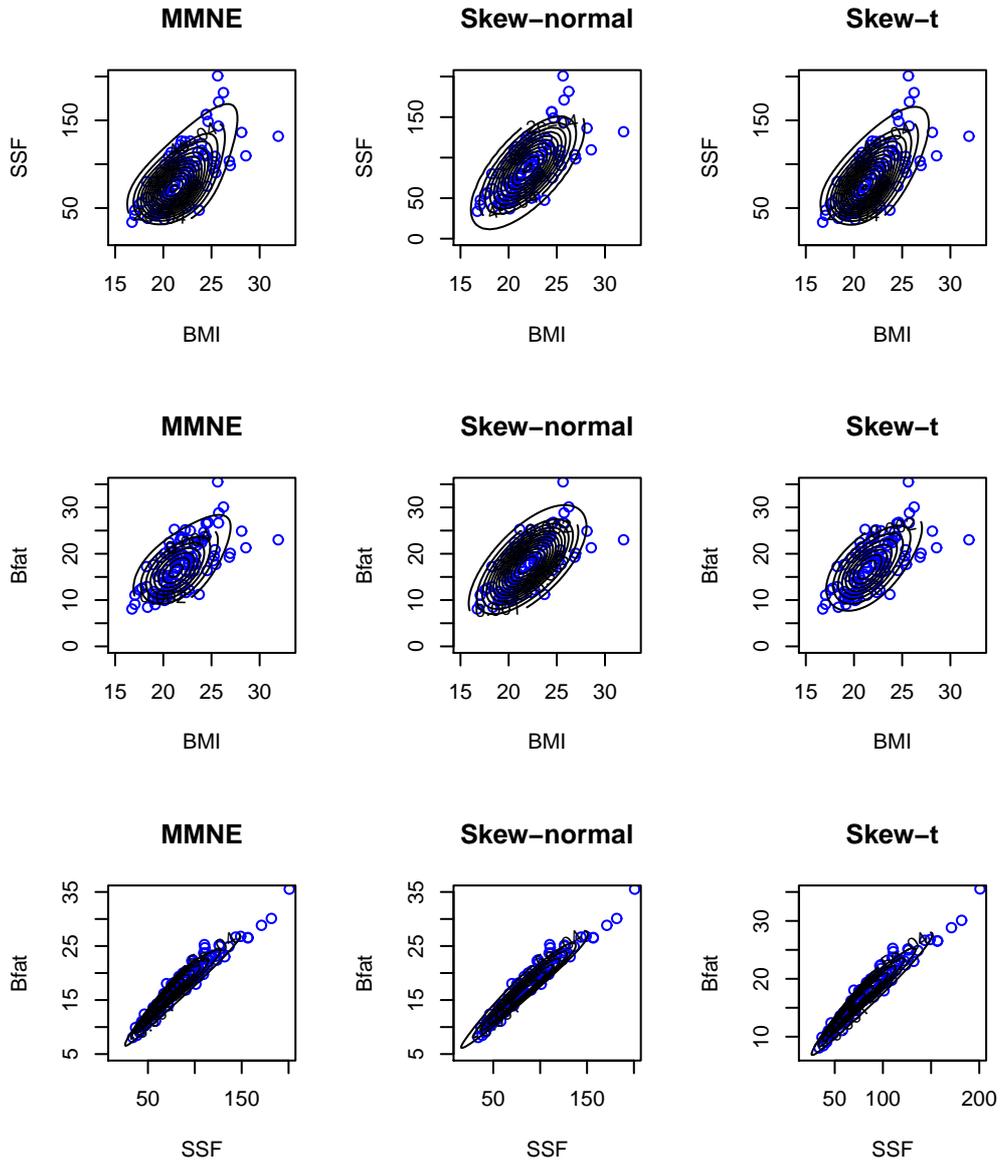} }
\caption{ Scatter plots of pairs of the  three selected variables for the  AIS data set, along with  the contour plots for the  fitted $\mathcal{MMNE}$, skew-normal ($\mathcal{SN}$) and skew-t ($\mathcal{ST}$) distributions  presented in Section \ref{sec1} and \ref{sec2}.   \label{scatter-AIS}}
\end{figure}

\subsection{ Italian olive oil data}\label{olive-Example}
As a second example, we consider the well-known data on the percentage composition of eight fatty acids found by lipid fraction of 572 Italian olive oils.  These data come from three areas; within each area, there are a number of constituent regions,  9 in total. The data set includes a data frame with 572 observations and 10 columns. The first column gives the area (one of Southern Italy, Sardinia, and Northern Italy), the second gives the region, and the remaining 8 columns give the variables. Southern Italy  consists of  North Apulia, Calabria, South Apulia, and Sicily regions, Sardinia is divided into Inland Sardinia, and Coastal Sardinia, and Northern Italy  consists of  Umbria, East Liguria, and West Liguria regions. These data are available in the \textsf{R} software, \texttt{pgmm} package.
\begin{figure}[htp]
 \centerline{ \includegraphics[scale=.45]{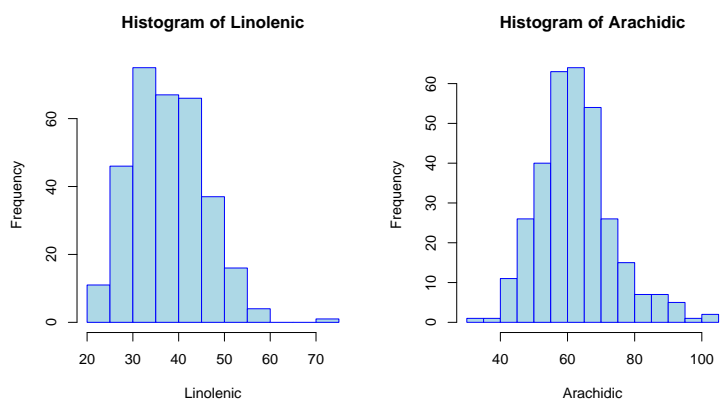} }
\caption{  The  histograms of the two selected variables Linolenic and Arachidic fatty acids of olive oil data set in Subsection \ref{olive-Example}.   \label{Hist-Olive}}
\end{figure}
\begin{table}[htb!]
\small
\center
\caption{Parameter estimates of the $\mathcal{MMNE}$ distribution by using EM Algorithm presented in Section \ref{Est-sec},  based on the two selected variables (Linolenic, Arachidic) of the  olive oil data set  in Subsection \ref{olive-Example}.}
\label{tab_est-Olive}
\begin{tabular}{ccc}
    \hline
$\widehat{{\boldsymbol{\xi}}}$&$\widehat{{\boldsymbol{\Omega}}}$& $\widehat{{\boldsymbol{\delta}}}$\\
\hline\\[-0.2cm]
$
\begin{bmatrix}
  36.8344\\
  55.3462 \\
\end{bmatrix}
$
&
$
\begin{bmatrix}
63.3623&  40.9481\\
40.9481& 124.0575\\
 \end{bmatrix}
$
&
$
\begin{bmatrix}
0.1546\\
0.6977\\
\end{bmatrix}
$
\\[0.2cm]
 \hline
\end{tabular}
\end{table}
\begin{table}[htb!]
\small
\center
\caption{Values of skewness measures in Section \ref{measures-skewness},  based on the two selected variables (Linolenic, Arachidic)  of the olive oil data set  in Subsection \ref{olive-Example}.}
\label{tab_skew-Olive}
\begin{tabular}{cccccccc}
    \hline
$\beta_{1,p}$&$\beta_{1p}^2$& $s$  &  $b$& $Q^*$& $T$ &  $s_I$& $s_C$\\
\hline\\[-0.2cm]
0.5707 &  0.1218 &
$
\begin{bmatrix}
  -0.0492\\
   0.7538\\
\end{bmatrix}
$
&
$
\begin{bmatrix}
   -0.0428\\
   0.6557 \\
\end{bmatrix}
$
&
0.0802&
$
\begin{bmatrix}
  -0.0185\\
  0.2827 \\
\end{bmatrix}
$
&
0.0517
&
$
\begin{bmatrix}
   0.0486\\
    0.2195 \\
\end{bmatrix}
$
 \\[0.2cm]
 \hline
 \end{tabular}
\end{table}

For the purpose of illustration, we consider 323 cases from Southern Italy, and columns (8, 9), Linolenic and Arachidic fatty acids, respectively, so as to consider the  bivariate case.  Fig.~\ref{Hist-Olive} shows the  histograms  of the two selected variables, while Table \ref{tab_est-Olive} presents the estimates of parameters and Table \ref{tab_skew-Olive} presents the values of skewness  measures.
Table \ref{tab_olive_fit} provides the fit of $\mathcal{MMNE}$ model, as compared to those of  $\mathcal{SN}$ and $\mathcal{ST}$ distributions, for the considered data. From  Table \ref{tab_olive_fit}, it is clear that the $\mathcal{MMNE}$ model provides the best overall fit as it  possesses  the largest $\ell({\boldsymbol{\hat{\theta}}}|{\boldsymbol{y}})$ value and
the lowest AIC and BIC  values.
Fig. \ref{fig-olive} shows the scatter plot of the data and the  contour plots of the  fitted $\mathcal{MMNE}$, $\mathcal{SN}$ and $\mathcal{ST}$ distributions.

\begin{table}[htb]
\small
\center
\caption{Comparison of fitting measures, maximized log-likelihood value $\ell({\boldsymbol{\hat{\theta}}}|{\boldsymbol{y}})$, Akaike information criterion AIC and Bayesian information criterion BIC, for  skew-normal ($\mathcal{SN}$), skew-t ($\mathcal{ST}$) and $\mathcal{MMNE}$ distributions for the two selected variables (Linolenic, Arachidic)  of  the olive oil data set  in Subsection \ref{olive-Example}.}
\label{tab_olive_fit}
\begin{tabular}{l|ccc}
    \hline
Distribution & ~~~~~~~$\ell({\boldsymbol{\hat{\theta}}}|{\boldsymbol{y}}) $  &   ~~~~~~~  AIC&  ~~~~~~~   BIC\\
\hline
$\mathcal{SN}$     & ~~~~~~~ -2320.039 &  ~~~~~~~ 4654.079& ~~~~~~~ 4680.522  \\
$\mathcal{ST}$    &  ~~~~~~~-2316.320 &  ~~~~~~~  4648.640&  ~~~~~~~ 4678.861  \\
$\mathcal{MMNE}$     & ~~~~~~~-2314.604 &  ~~~~~~~ 4643.207&  ~~~~~~~4669.651  \\
\hline
\end{tabular}
\end{table}

\begin{figure}[htp]
 \centerline{ \includegraphics[scale=.60]{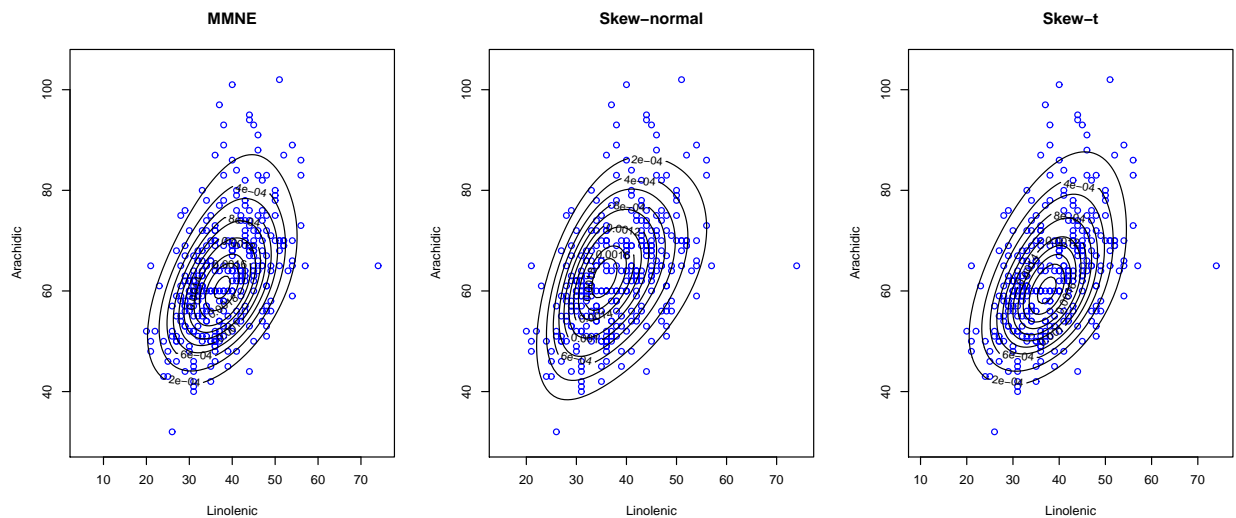} }
\caption{ Scatter plots  of the  olive oil data, and  the contour plots of the  fitted $\mathcal{MMNE}$, skew-normal ($\mathcal{SN}$) and skew-t ($\mathcal{ST}$) distributions presented in Section \ref{sec1} and \ref{sec2}.   \label{fig-olive}}
\end{figure}

\section{Concluding Remarks}
In this paper, we have discussed  the mean mixture of multivariate normal distribution ($\mathcal{MMN}$), which includes the normal, $\mathcal{SN}$, and extended  $\mathcal{SN}$ distributions as  particular cases. We have studied several features of this family of distributions, including the first four moments, the  distributions of affine transformations and canonical forms, estimation of parameters by using an EM-type algorithm with closed-form expressions,  and different measures of multivariate skewness. Two special cases of the $\mathcal{MMN}$ family, with standard gamma and standard exponential distributions as mixing distributions, denoted by $\mathcal{MMNG}$ and $\mathcal{MMNE}$ distributions, have been studied in detail.
A simulation study has been performed to evaluate the performance of the MLEs of parameters of the $\mathcal{MMNE}$ distribution.
From the results in Section \ref{Example}, for the AIS and  olive oil data sets, the $\mathcal{MMNE}$ distribution is shown to provide a better fit than the  $\mathcal{SN}$ and $\mathcal{ST}$ distributions.
Different multivariate measures of skewness have been derived for the $\mathcal{MMNE}$ distribution, and the evaluation of tests based on these measures  is carried out in terms of powers of tests.

There are  several possible directions for future research. For example, the study of finite mixtures and scale mixtures of $\mathcal{MMN}$ family will be of great interest.
In the stochastic representation  in (\ref{Repre}), if the skewness parameter is a matrix, with representation
${\mathbf{Y}} \stackrel{d}{=}  {\boldsymbol{\xi}}+{\boldsymbol{\omega}}\left( {\boldsymbol{\Delta}} \mathbf{U}+ {\mathbf{Z}}\right)$, then ${\mathbf{Y}}$ has the unified skew
normal  ($\mathcal{SUN}$) distribution (see \cite{Arellano Valle and  Azzalini (2006)}), wherein elements  of $\mathbf{U}$ have the  standard half-normal distribution. In this connection, consideration of a general distribution for $\mathbf{U}$ would be of interest.
All the computations presented
in this paper were performed by using the statistical software \textsf{R}, version 4.0.0. The computer program for the implementation of the proposed EM-type algorithm and   comparison of the skewness measures are available as supplementary material associated with this article.\\

\noindent\textbf{CRediT authorship contribution statement}\\

\textbf{Me$'$raj Abdi:} Conceptualization, Methodology, Software, Writing - original draft, Writing – review \& editing. Investigation, Validation.
\textbf{Mohsen Madadi:} Methodology, Supervision, Investigation.
\textbf{Narayanaswamy Balakrishnan:} Supervision, Writing - review \& editing, Methodology.
\textbf{Ahad Jamalizadeh:} Conceptualization, Methodology, Supervision, Visualization.

\section*{Acknowledgments}  The authors are grateful to the Editors and two anonymous reviewers who provided
very helpful feedback, comments, and suggestions, based on which the paper has improved significantly.

\section*{Appendix A. Proofs}

\begin{proof}[\bf Proof of  Lemma \ref{Lem1}.]
By using (\ref{MtX}), we can calculate the partial derivatives of $M_{\mathbf{X}}({\mathbf{t}})$, the MGF of normalized $\mathcal{MMN}$ distribution, that are directly related to the moments of the $\mathcal{MMN}$ random  vector. Suppose ${\mathbf{X}}\sim \mathcal{MMN}_p({\boldsymbol{0}}, {\overline{\mathbf{\Omega}}},{\boldsymbol{\delta}}; H)$.  Then, some derivatives
 of $M_{\mathbf{X}}({\mathbf{t}})$ in (\ref{MtX}) are as follows:
\begin{eqnarray*}
\frac{\partial M_{\mathbf{X}}({\mathbf{t}})}{\partial {\mathbf{t}}}&=&e^{\frac{1}{2}{\mathbf{t}}^\top{\mathbf{ \Sigma}}_{\mathbf{X}} {\mathbf{t}}}
\left[{\mathbf{ \Sigma}}_{\mathbf{X}}{\mathbf{t}} M_U({\mathbf{t}}^\top{\boldsymbol{\delta}})+{\boldsymbol{\delta}} M_U^{(1)}({\mathbf{t}}^\top{\boldsymbol{\delta}})\right],\label{dif1}\\
\frac{\partial^2 M_{\mathbf{X}}({\mathbf{t}})}{\partial {\mathbf{t}}\partial {\mathbf{t}}^\top}&=&e^{\frac{1}{2}{\mathbf{t}}^\top{\mathbf{ \Sigma}}_{\mathbf{X}} {\mathbf{t}}}
 \left\{  M_U({\mathbf{t}}^\top{\boldsymbol{\delta}})\left[{\mathbf{ \Sigma}}_{\mathbf{X}}+({\mathbf{ \Sigma}}_{\mathbf{X}}{\mathbf{t}})\otimes({\mathbf{ \Sigma}}_{\mathbf{X}}{\mathbf{t}})^\top    \right]
 +M_U^{(1)}({\mathbf{t}}^\top{\boldsymbol{\delta}} )\left[ ({\mathbf{ \Sigma}}_{\mathbf{X}}{\mathbf{t}})\otimes {\boldsymbol{\delta}}^\top + {\boldsymbol{\delta}}\otimes ({\mathbf{ \Sigma}}_{\mathbf{X}}{\mathbf{t}})^\top\right]+ M_U^{(2)}({\mathbf{t}}^\top{\boldsymbol{\delta}})~{\boldsymbol{\delta}}\otimes {\boldsymbol{\delta}}^\top \right\},\label{dif2}\\
 \frac{\partial^3 M_{\mathbf{X}}({\mathbf{t}})}{\partial {\mathbf{t}}\partial {\mathbf{t}}^\top\partial {\mathbf{t}}}&=&e^{\frac{1}{2}{\mathbf{t}}^\top{\mathbf{ \Sigma}}_{\mathbf{X}} {\mathbf{t}}}
 \left \{ M_U({\mathbf{t}}^\top{\boldsymbol{\delta}})
 \left[ ({\mathbf{ \Sigma}}_{\mathbf{X}}{\mathbf{t}})\otimes {\mathbf{ \Sigma}}_{\mathbf{X}}
 +\mathrm{vec}({\mathbf{ \Sigma}}_{\mathbf{X}})({\mathbf{ \Sigma}_{\mathbf{X}}}{\mathbf{t}})^\top
 +(\mathbf{I}_p\otimes ({\mathbf{ \Sigma}}_{\mathbf{X}}{\mathbf{t}}))({\mathbf{ \Sigma}}_{\mathbf{X}} + ({\mathbf{ \Sigma}}_{\mathbf{X}}{\mathbf{t}})\otimes ({\mathbf{ \Sigma}}_{\mathbf{X}}{\mathbf{t}})^\top)   \right]\right. \nonumber\\
 &&+ M_U^{(1)}({\mathbf{t}}^\top{\boldsymbol{\delta}})\left[{\boldsymbol{\delta}} \otimes{\mathbf{ \Sigma}}_{\mathbf{X}}+\mathrm{vec}({\mathbf{ \Sigma}}_{\mathbf{X}}){\boldsymbol{\delta}}^\top
+ (\mathbf{I}_p\otimes ({\mathbf{ \Sigma}}_{\mathbf{X}}{\mathbf{t}}))[{\boldsymbol{\delta}}\otimes({\mathbf{ \Sigma}}_{\mathbf{X}}{\mathbf{t}})^\top+({\mathbf{ \Sigma}}_{\mathbf{X}}{\mathbf{t}})\otimes{\boldsymbol{\delta}}^\top]\right.\nonumber\\
&&+  \left.   (\mathbf{I}_p\otimes {\boldsymbol{\delta}} )\left({\mathbf{ \Sigma}}_{\mathbf{X}}+({\mathbf{ \Sigma}}_{\mathbf{X}}{\mathbf{t}})\otimes({\mathbf{ \Sigma}}_{\mathbf{X}}{\mathbf{t}})^\top     \right)     \right]\nonumber\\
 &&+  M_U^{(2)}({\mathbf{t}}^\top{\boldsymbol{\delta}})   \left[ (\mathbf{I}_p\otimes ({\mathbf{ \Sigma}}_{\mathbf{X}}{\mathbf{t}}))({\boldsymbol{\delta}} \otimes {\boldsymbol{\delta}}^\top ) +(\mathbf{I}_p\otimes {\boldsymbol{\delta}} ) \left( {\boldsymbol{\delta}} \otimes ({\mathbf{ \Sigma}}_{\mathbf{X}}{\mathbf{t}})^\top  + ({\mathbf{ \Sigma}}_{\mathbf{X}}{\mathbf{t}}) \otimes {\boldsymbol{\delta}}^\top \right)  \right]+  \left.M_U^{(3)}({\mathbf{t}}^\top{\boldsymbol{\delta}}) (\mathbf{I}_p\otimes {\boldsymbol{\delta}} ) ({\boldsymbol{\delta}} \otimes {\boldsymbol{\delta}}^\top)   \right\},\label{dif3}
\end{eqnarray*}
where $M_U^{(1)}({\mathbf{t}}^\top{\boldsymbol{\delta}})=\frac{\partial M_U({\mathbf{t}}^\top{\boldsymbol{\delta}}) }{\partial \mathbf{t}}$,
$M_U^{(2)}({\mathbf{t}}^\top{\boldsymbol{\delta}})=\frac{\partial^2 M_U({\mathbf{t}}^\top{\boldsymbol{\delta}})}{\partial {\mathbf{t}}\partial {\mathbf{t}}^\top}$
and  $M_U^{(3)}({\mathbf{t}}^\top{\boldsymbol{\delta}})=\frac{\partial^3 M_U({\mathbf{t}}^\top{\boldsymbol{\delta}})}{\partial {\mathbf{t}}\partial {\mathbf{t}}^\top\partial {\mathbf{t}}}$.
Setting $\mathbf{t}=\boldsymbol{0}$,  as in \cite{Genton et al. (2001)}, we obtain the first three moments of the $\mathcal{MMN}$ family. To find the fourth moment, since we only need the value of fourth partial derivative of $M_{\mathbf{X}}({\mathbf{t}}) $ at $\mathbf{t}=\boldsymbol{0}$, say $M_4({\mathbf{X}})=\frac{\partial^4 M_{\mathbf{X}}({\mathbf{t}})}{\partial {\mathbf{t}}\partial {\mathbf{t}}^\top\partial {\mathbf{t}}\partial {\mathbf{t}}^\top}|_{\mathbf{t}=\boldsymbol{0}}$,
we do not need to compute the whole expression.  Instead, we can simply single out all the terms in $\frac{\partial^4 M_{\mathbf{X}}({\mathbf{t}})}{\partial {\mathbf{t}}\partial {\mathbf{t}}^\top\partial {\mathbf{t}}\partial {\mathbf{t}}^\top}$ that do not contain the factor $\mathbf{t}$ or $\mathbf{t}^\top$.
\end{proof}

\noindent{\bf Note 1:}
The stochastic representation
${\mathbf{Y}} \stackrel{d}{=}  {\boldsymbol{\xi}}+{\boldsymbol{\omega}}\left( {\boldsymbol{\delta}} U + {\mathbf{Z}}\right)$
can be used directly as  a way to obtain the first  four moments of ${\mathbf{Y}}$ in the following formulas:
\begin{eqnarray*}
M_1({\mathbf{Y}})&=&{\rm E}({\mathbf{Y}}),~~~~~~~~~
M_2({\mathbf{Y}})={\rm E}\left({\mathbf{Y}}\otimes {\mathbf{Y}}^\top\right)={\rm E}\left({\mathbf{Y}} {\mathbf{Y}}^\top\right),~~~~~~~~
M_3({\mathbf{Y}})={\rm E}\left({\mathbf{Y}}\otimes {\mathbf{Y}}^\top \otimes {\mathbf{Y}}\right)={\rm E}\left[({\mathbf{Y}}\otimes {\mathbf{Y}})  {\mathbf{Y}}^\top\right], \nonumber\\
M_4({\mathbf{Y}})&=&{\rm E}\left({\mathbf{Y}}\otimes {\mathbf{Y}}^\top \otimes {\mathbf{Y}}\otimes {\mathbf{Y}}^\top\right)={\rm E}\left[\left({\mathbf{Y}} {\mathbf{Y}}^\top\right) \otimes \left({\mathbf{Y}} {\mathbf{Y}}^\top\right)\right].
\end{eqnarray*}
The corresponding central moments of ${\mathbf{Y}}$ are then
\begin{eqnarray*}
\overline{M}_1({\mathbf{Y}})&=&{\mathbf{0}},~~~~~~~~~
\overline{M}_2({\mathbf{Y}})={\rm E}\left\{[{\mathbf{Y}}-{\rm E}({\mathbf{Y}})]\otimes [{\mathbf{Y}}-{\rm E}({\mathbf{Y}})]^\top \right\}=\textrm{var}({\mathbf{Y}}),\nonumber\\
\overline{M}_3({\mathbf{Y}})&=&{\rm E}\left\{[{\mathbf{Y}}-{\rm E}({\mathbf{Y}})]\otimes [{\mathbf{Y}}-{\rm E}({\mathbf{Y}})]^\top \otimes [{\mathbf{Y}}-{\rm E}({\mathbf{Y}})]\right\},\nonumber\\
\overline{M}_4({\mathbf{Y}})&=&{\rm E}\left\{\left([{\mathbf{Y}}-{\rm E}({\mathbf{Y}})] [{\mathbf{Y}}-{\rm E}({\mathbf{Y}})]^\top\right) \otimes \left([{\mathbf{Y}}-{\rm E}({\mathbf{Y}})] [{\mathbf{Y}}-{\rm E}({\mathbf{Y}})]^\top\right)\right\}.
\end{eqnarray*}

\noindent{\bf Note 2:} We know that for any multivariate random vector ${\mathbf{Y}}$, the central moments of third and fourth orders are related to the non-central moments by the following relationships (see, for example, \cite{Kollo and Srivastava (2004)} and \cite{Kollo and von Rosen (2005)}):
\begin{eqnarray}
\overline{M}_3({\mathbf{Y}}) &=& M_3({\mathbf{Y}}) - M_2({\mathbf{Y}}) \otimes {\rm E}({\mathbf{Y}}) - {\rm E}({\mathbf{Y}}) \otimes M_2({\mathbf{Y}}) - \mathrm{vec}(M_2({\mathbf{Y}})){\rm E}({\mathbf{Y}})^\top+ 2{\rm E}({\mathbf{Y}}){\rm E}({\mathbf{Y}})^\top \otimes {\rm E}({\mathbf{Y}}),\label{M3-M3bar}\\
\overline{M}_4({\mathbf{Y}}) &=& M_4({\mathbf{Y}}) - (M_3({\mathbf{Y}}) )^\top\otimes {\rm E}({\mathbf{Y}})
 - M_3({\mathbf{Y}}) \otimes {\rm E}({\mathbf{Y}}) ^\top -{\rm E}({\mathbf{Y}})\otimes (M_3({\mathbf{Y}}) )^\top -{\rm E}({\mathbf{Y}}) ^\top \otimes M_3({\mathbf{Y}})
\nonumber\\
&&+M_2({\mathbf{Y}})\otimes  {\rm E}({\mathbf{Y}})E({\mathbf{Y}})^\top
+ ({\rm E}({\mathbf{Y}}) \otimes {\rm E}({\mathbf{Y}}))(\mathrm{vec}(M_2({\mathbf{Y}})))^\top+{\rm E}({\mathbf{Y}}) \otimes M_2({\mathbf{Y}})\otimes {\rm E}({\mathbf{Y}})^\top  \nonumber\\
&&+
{\rm E}({\mathbf{Y}})^\top \otimes M_2({\mathbf{Y}})\otimes {\rm E}({\mathbf{Y}})
+
{\rm E}({\mathbf{Y}})^\top \otimes \mathrm{vec}(M_2({\mathbf{Y}}))\otimes {\rm E}({\mathbf{Y}})^\top+ {\rm E}({\mathbf{Y}}){\rm E}({\mathbf{Y}})^\top \otimes M_2({\mathbf{Y}}) \nonumber\\
&&- 3{\rm E}({\mathbf{Y}}){\rm E}({\mathbf{Y}})^\top \otimes {\rm E}({\mathbf{Y}}){\rm E}({\mathbf{Y}})^\top.\label{M4-M4bar}
\end{eqnarray}

Upon using the relations for affine transformations of moments, we then obtain
\begin{eqnarray}
M_1({\mathbf{AY}})&=&{\rm E}({\mathbf{AY}})={\mathbf{A}}{\rm E}({\mathbf{Y}}), \label{M1-AY}~~~~~~~~~M_2({\mathbf{AY}})={\rm E}\left({\mathbf{AY}}\otimes ({\mathbf{AY}})^\top\right)={\mathbf{A}}{\rm E}\left({\mathbf{Y}}\otimes {\mathbf{Y}}^\top\right) {\mathbf{A}}^\top,\\
M_3({\mathbf{AY}})&=&{\rm E}[({\mathbf{AY}}\otimes {\mathbf{AY}}) ({\mathbf{AY}})^\top]={\rm E}\left\{\mathrm{vec}\left({\mathbf{AY}}({\mathbf{AY}})^\top\right)({\mathbf{AY}})^\top\right\} =({\mathbf{A}}\otimes {\mathbf{A}})M_3({\mathbf{Y}}) {\mathbf{A}}^\top, \label{M3-AY}\\
M_4({\mathbf{AY}})&=&{\rm E}\left({\mathbf{AY}} ({\mathbf{AY}})^\top \otimes {\mathbf{AY}} ({\mathbf{AY}})^\top\right)
=({\mathbf{A}}\otimes {\mathbf{A}})M_4({\mathbf{Y}})\left({\mathbf{A}}\otimes {\mathbf{A}}\right)^\top.\label{M4-AY}
\end{eqnarray}

\begin{proof}[\bf Proof of  Theorem \ref{Thm1}.]
The moment generating function of ${\mathbf{A}}^\top{\mathbf{X}} $ can be written as
\begin{eqnarray*}
M_{{\mathbf{A}}^\top{\mathbf{X}} }({\mathbf{t}})=M_{\mathbf{X}}({\mathbf{A}}{\mathbf{t}})=
e^{\frac{1}{2}{\mathbf{t}}^\top\left({\mathbf{A}}^\top{\overline{\mathbf{\Omega}}}{\mathbf{A}} -{\mathbf{A}}^\top{\boldsymbol{\delta}}{\boldsymbol{\delta}}^\top{\mathbf{A}} \right) {\mathbf{t}}}
M_U\left({\mathbf{t}}^\top {\mathbf{A}}^\top{\boldsymbol{\delta}};{\boldsymbol{\nu}}\right).
\end{eqnarray*}
Upon using the uniqueness property of the moment generating function, the required result  is obtained.
\end{proof}

\begin{proof}[\bf Proof of  Theorem \ref{Thm2}.]  The moment generating function of ${\mathbf{X}}=\textbf{c}+{\mathbf{A}}^\top {\mathbf{Y}}$ can be written as
\begin{eqnarray*}
M_{{\mathbf{X}} }({\mathbf{t}})
&=&e^{{\mathbf{t}}^\top  \textbf{c}}M_{\mathbf{Y}}({\mathbf{A}}{\mathbf{t}})=
e^{ {\mathbf{t}}^\top {\boldsymbol{\xi}}_{{\mathbf{X}} }+
\frac{1}{2}{\mathbf{t}}^\top  \left({\mathbf{\Omega}}_{{\mathbf{X}} }- {\boldsymbol{\omega}}_{{\mathbf{X}} }{\boldsymbol{\delta}}_{{\mathbf{X}} } {\boldsymbol{\delta}}_{{\mathbf{X}} }^\top{\boldsymbol{\omega}}_{{\mathbf{X}} }\right){\mathbf{t}}} M_U\left({\mathbf{t}}^\top {\boldsymbol{\omega}}_{{\mathbf{X}} }{\boldsymbol{\delta}}_{{\mathbf{X}} } ;{\boldsymbol{\nu}}\right),
\end{eqnarray*}
which completes the proof.
\end{proof}

\begin{proof}[\bf Proof of  Theorem \ref{Thm3}.]
 We have introduced the $\mathcal{MMN}$ distribution by assuming ${\mathbf{\Omega}}>0$ through the factorization ${\mathbf{\Omega}}={\boldsymbol{\omega}} \overline{{\mathbf{\Omega}}}{\boldsymbol{\omega}}$. The matrix $\overline{{\mathbf{\Omega}}}$ is a positive definite non-singular matrix if and only if there exists some invertible (non-singular)  matrix ${\mathbf{C}}$ such that  $\overline{{\mathbf{\Omega}}}=\mathbf{C}^\top\mathbf{C}$. If ${\boldsymbol{\delta}}\neq \textbf{0}$, there exists an orthogonal
matrix $\mathbf{P}$ with the first column  being proportional to $\mathbf{C}\overline{{\mathbf{\Omega}}}^{-1}{\boldsymbol{\delta}}$, while for ${\boldsymbol{\delta}}=\textbf{0}$
we set $\mathbf{P}=\mathbf{I}_p$. Finally, define ${\mathbf{A}}_*=\left({\mathbf{C}}^{-1} \mathbf{P}\right)^\top{\boldsymbol{\omega}}^{-1}$. By using Theorem \ref{Thm2},
 we see that $\textbf{Z}^*={\mathbf{A}}_*({\mathbf{Y}}-{\boldsymbol{\xi}})$ has the stated distribution with ${\boldsymbol{\delta}}_{{\mathbf{Z}}^*}=({\mathbf{\delta}}_*, 0, \ldots, 0)^\top$.
 \end{proof}

\begin{proof}[\bf Proof of  Theorem \ref{Thm5-mode}.]
First, consider the mode of the
corresponding canonical variable $Z^*\sim \mathcal{MMN}_p({\boldsymbol{0}},{\mathbf{I}}_p, {\boldsymbol{\delta}}_{\textbf{Z}^*}; H)$. We find this mode  by solving the following equations with respect to $z_1, z_2, \ldots, z_p$:
\begin{eqnarray*}
\frac{\partial f_{Z_1^*}(z_1)}{\partial z_1}=0,~~~~~z_i f_{Z_1^*}(z_1)=0,~~~ i\in\{2,3, \ldots, p\}.
\end{eqnarray*}
The last $p - 1$ equations are fulfilled when $z_i= 0$, while the  root of the first equation corresponds to the mode, $m_0^*$ say, of the $\mathcal{MMN}_1(0,1,{\mathbf{\delta}}_*; H)$ distribution. Therefore, the mode of $\textbf{Z}^*$ is $\textbf{M}_0^*=(m_0^*, 0, \ldots, 0)^\top = \frac{m_0^*}{{\mathbf{\delta}}_*}{\boldsymbol{\delta}}_{\textbf{Z}^*}$. From Theorem  \ref{Thm3}, we can write ${\mathbf{Y}}={\boldsymbol{\xi}}+{\boldsymbol{\omega}}\mathbf{C}^\top \mathbf{P}\textbf{Z}^* $
and  ${\boldsymbol{\delta}}_{\textbf{Z}^*}=\mathbf{P}^\top \mathbf{C}\overline{{\mathbf{\Omega}}}^{-1}{\boldsymbol{\delta}}$. As  the
mode is equivariant with respect to affine transformations, the mode of $\textbf{Y}$ is
\begin{eqnarray*}
\textbf{M}_0={\boldsymbol{\xi}}+\frac{m_0^*}{{\mathbf{\delta}}_*}{\boldsymbol{\omega}}\mathbf{C}^\top \mathbf{P}{\boldsymbol{\delta}}_{\textbf{Z}^*}
={\boldsymbol{\xi}}+\frac{m_0^*}{{\mathbf{\delta}}_*}{\boldsymbol{\omega}}\mathbf{C}^\top \mathbf{P}\mathbf{P}^\top \mathbf{C}\overline{{\mathbf{\Omega}}}^{-1}{\boldsymbol{\delta}}
={\boldsymbol{\xi}}+\frac{m_0^*}{{\mathbf{\delta}}_*}{\boldsymbol{\omega}}{\boldsymbol{\delta}}.
\end{eqnarray*}
Hence, the  result.
\end{proof}

\section*{Appendix B. Computation of Different Measures of Skewness  }

\subsection*{B1. Mardia  Measure of Skewness}\label{Sec-Mardia}

Mardia \cite{Mardia (1970), Mardia (1974)} presented a multivariate measure of  skewness of an arbitrary $p$-dimensional distribution $F$ with mean vector ${\boldsymbol{\mu}}$ and covariance matrix $\mathbf{\Delta}$. Let
${\mathbf{X}} $  and ${\mathbf{Y}}$   be two independent and identically distributed random vectors
from  distribution  $F$. Then, the  measure of skewness is
\begin{eqnarray}
\beta_{1,p}&=&{\rm E}\left[\left\{({\mathbf{X}}-{\boldsymbol{\mu}})^\top {\mathbf{\Delta}}^{-1}({\mathbf{Y}}-{\boldsymbol{\mu}})\right\}^3\right],\label{Mardia-SK}
\end{eqnarray}
where ${\boldsymbol{\mu}}={\rm E}({\mathbf{X}})$ and  ${\mathbf{\Delta}}=\textrm{var}({\mathbf{X}})$.
Mardia  measure of  skewness is location and scale invariant (see \cite{Mardia (1970)}).  From  Theorems \ref{Thm2} and \ref{Thm3}, the $\mathcal{MMN}$ family is closed under affine transformations and have a canonical form. If ${\mathbf{X}}\sim \mathcal{MMN}_p({\boldsymbol{\xi}},{\mathbf{\Omega}}, {\boldsymbol{\delta}}; H)$,
there exists a linear transformation
$\textbf{Z}^*={\mathbf{A}}_*({\mathbf{Y}}-{\boldsymbol{\xi}})$ such that $\textbf{Z}^*\sim \mathcal{MMN}_p({\boldsymbol{0}},{\mathbf{I}}_p, {\boldsymbol{\delta}}_{\textbf{Z}^*}; H)$, where at most one component of ${\boldsymbol{\delta}}_{\textbf{Z}^*}$ is not zero.
Without loss of any generality, we take the first component of $\textbf{Z}^*$ to be skewed and denote it by $Z_1^*$,  and so for computing the measure in (\ref{Mardia-SK}), we can use the canonical form of the $\mathcal{MMN}$ family. Let ${\mathbf{X}}^*$ and ${\mathbf{Y}}^*$ be
two independent and identically distributed random vectors
from $\mathcal{MMN}_p({\boldsymbol{0}},{\mathbf{I}}_p, {\boldsymbol{\delta}}_{\textbf{Z}^*}; H)$. Now,
by using
${\boldsymbol{\mu}}^*={\rm E}({\mathbf{X}}^*)={\rm E}({\mathbf{Y}}^*)=  {\rm E}(U){\boldsymbol{\delta}}_{\textbf{Z}^*}$ and
${\boldsymbol{\Delta}}^*=\textrm{var}({\mathbf{X}}^*)=\textrm{var}({\mathbf{Y}}^*)={\mathbf{I}}_p+ \left(\textrm{var}(U)-1\right){\boldsymbol{\delta}}_{\textbf{Z}^*}{\boldsymbol{\delta}}_{\textbf{Z}^*}^\top$ in  (\ref{Mardia-SK}), the Mardia  measure of  skewness can be expressed as
\begin{eqnarray}\label{Mardia-skew}
\beta_{1,p}&=&{\rm E}\left[\left\{({\textbf{X}^*}-{\boldsymbol{\mu}}^*)^\top [{\boldsymbol{\Delta}}^*]^{-1}(\textbf{Y}^*-{\boldsymbol{\mu}}^*])\right\}^3\right] =\left({\rm E}\left[\frac{Z_1^*-\delta_*}{\sqrt{\textrm{var}(Z_1^*)}}\right]^3\right)^2 =(\gamma_1^*)^2,
\end{eqnarray}
where $\gamma_1^*$ is  the univariate skewness of  $Z_1^*\sim \mathcal{MMN}_1(0,1,{\mathbf{\delta}}_*; H)$ of the canonical form (see Theorem \ref{Thm3}).
An explicit formula of $\gamma_1^*$  can be found  in \cite{Negarestani et al. (2019)} for the univariate case.

\subsection*{B2. Malkovich-Afifi  Measure of Skewness}\label{Sec-Malkovich-Afifi-skew}

Malkovich and Afifi \cite{Malkovich and Afifi (1973)} proposed a measure of multivariate skewness as a different type of generalization of the univariate measure. By denoting the unit $p$-dimensional sphere by ${\phi}_p=\left\{{\mathbf{u}}\in \mathbb{R}^p; ||{\mathbf{u}}||=1  \right\}$, for ${\mathbf{u}}\in {\phi}_p$, the usual univariate measure of skewness in the ${\mathbf{u}}$-direction is
\begin{eqnarray}\label{malkovich1skew}
\beta_{1}({\mathbf{u}})&=&\frac{\left[ {\rm E}\left\{{\mathbf{u}}^\top({\mathbf{Y}}-{\rm E}({\mathbf{Y}}))\right\}^3  \right]^2}{\left[\textrm{var}({\mathbf{u}}^\top{\mathbf{Y}})\right]^3},
\end{eqnarray}
and so the Malkovich-Afifi multivariate extension of it is  defined as
\begin{eqnarray}\label{malkovich2skew}
\beta_1^{*}= \sup_{{\mathbf{u}}\in {\phi}_p}  \beta_{1}({\mathbf{u}}).
\end{eqnarray}

Malkovich-Afifi  measure of multivariate skewness is also location and scale invariant.
\cite{Malkovich and Afifi (1973)} then defined the measures in (\ref{malkovich1skew}) and  (\ref{malkovich2skew})  and showed that if ${\mathbf{Z}}$ is the standardized variable ${\mathbf{Z}}={\mathbf{\Delta}}^{-1/2}({\mathbf{Y}}-{\boldsymbol{\mu}})$, an equivalent
version is
$
\beta_1^{*}= \sup_{{\mathbf{u}}\in {\phi}_p} \left({\rm E}\left[({\mathbf{u}}^\top{\mathbf{Z}})^3   \right]  \right)^2.
$ 
For obtaining $\beta_1^{*}$ for the $\mathcal{MMN}$ family, it is convenient to use the canonical form.
If ${\mathbf{Y}} \sim  \mathcal{MMN}_p({\boldsymbol{\xi}}, {\mathbf{\Omega}}, {\boldsymbol{\delta}}; H)$,  there exists a linear transformation
$\textbf{Z}^*={\mathbf{A}}_*({\mathbf{Y}}-{\boldsymbol{\xi}})$ such that $\textbf{Z}^*\sim \mathcal{MMN}_p({\boldsymbol{0}},{\mathbf{I}}_p, {\boldsymbol{\delta}}_{\textbf{Z}^*}; H)$, where at most one component of ${\boldsymbol{\delta}}_{\textbf{Z}^*}$ is not zero.
This means that the Malkovich-Afifi index, which is the maximum of the univariate skewness  measures among all
the directions of the unit sphere, will be, for $\textbf{Z}^*$, the index of asymmetry in the only (if there is) skew direction (without loss
of any generality, we take the first component of $\textbf{Z}^*$ to be skewed and denote it by $Z_1^*$):
\begin{eqnarray}\label{beta1star}
\beta_{1}^*=\beta_{1}^*({\mathbf{u}})&=&\sup_{{\mathbf{u}}\in {\phi}_p} \frac{\left[ {\rm E}\left\{{\mathbf{u}}^\top({\mathbf{Y}}-{\rm E}({\mathbf{Y}}))\right\}^3  \right]^2}{\left[\textrm{var}({\mathbf{u}}^\top{\mathbf{Y}})\right]^3}=\frac{\left[ {\rm E}\left\{Z_1^*-{\rm E}(Z_1^*)\right\}^3  \right]^2}{\left[\textrm{var}(Z_1^*)\right]^3}=(\gamma_1^*)^2.
\end{eqnarray}

As in the case of Mardia index, we have used $\gamma_1^*$ to denote the univariate skewness measure of the unique (if any) skewed component of the canonical
form ${\mathbf{Z}}^*$.  As this measure is location and scale invariant, it is invariant for linear transforms and consequently (\ref{beta1star}) is also
the Malkovich-Afifi measure for ${\mathbf{Y}}$, and thus it is the same as the Mardia index in (\ref{Mardia-skew}).

\subsection*{B3. Srivastava  Measure of Skewness}\label{Sec-Srivastava-skew}

Using principal components ${\mathbf{F}}={\mathbf{\Gamma}} {\mathbf{Y}}$,  Srivastava \cite{Srivastava (1984)} developed  a measure of skewness for the  multivariate vector ${\mathbf{Y}}$, where ${\mathbf{\Gamma}}=({\boldsymbol{\gamma}}_1, \ldots, {\boldsymbol{\gamma}}_p)$ is the matrix of eigenvectors of the covariance matrix ${\mathbf{\Delta}}$, that is, an orthogonal matrix such that ${\mathbf{\Gamma}}^\top{\mathbf{\Delta}}{\mathbf{\Gamma}}=\mathbf{\Lambda}$, and $\lambda_1, \ldots, \lambda_p$ are the corresponding eigenvalues. Srivastava's measure of skewness  for ${\mathbf{Y}}$ may then be presented as
\begin{eqnarray}
\beta_{1p}^2&=&\frac{1}{p}\sum_{i=1}^{p} \left\{\frac{{\rm E}(F_i-\theta_i)^3}{\lambda_i^{3/2}}\right\}^2=\frac{1}{p}\sum_{i=1}^{p} \left\{\frac{{\rm E}[{\boldsymbol{\gamma}}_i^\top({\mathbf{Y}}-{\boldsymbol{\mu}})]^3}{\lambda_i^{3/2}}\right\}^2,\label{Sriv-skew}
\end{eqnarray}
where $F_i={\boldsymbol{\gamma}}_i^\top {\mathbf{Y}}$ and $\theta_i={\boldsymbol{\gamma}}_i^\top {\boldsymbol{\mu}}$.
The measure in (\ref{Sriv-skew})  is  based on central moments of third order ${\rm E}[{\boldsymbol{\gamma}}_i^\top({\mathbf{Y}}-{\boldsymbol{\mu}})]^3$.
For obtaining this measure for the $\mathcal{MMN}$ distribution, we only need to obtain the non-central moments up to third
order. Upon using the relations in (\ref{M3-M3bar})-(\ref{M3-AY}), we can obtain the third central moment by replacing ${\mathbf{A}}$ by ${\boldsymbol{\gamma}}_i$ in (\ref{noncentralzzEE}).

\subsection*{B4. M\'{o}ri-Rohatgi-Sz\'{e}kely    Measure of Skewness}\label{Sec-Mori-skew}

M\'{o}ri {\rm et al.} \cite{Mori et al. (1993)}  suggested a vectorial measure of skewness as a $p$-dimensional  vector.
If ${\mathbf{Z}}={\mathbf{\Delta}}^{-1/2}({\mathbf{Y}}-{\boldsymbol{\mu}})=(Z_1, \ldots, Z_p)^\top$ is the standardized vector, this measure
can be written in terms of coordinates of ${\mathbf{Z}}$ as
\begin{eqnarray}\label{Moriii-skew}
s({\mathbf{Y}})={\rm E}(\|{\mathbf{Z}} \|^2 {\mathbf{Z}})={\rm E}\left(({\mathbf{Z}}^\top{\mathbf{Z}}) {\mathbf{Z}}\right)=\sum_{i=1}^{p}{\rm E}\left(Z_i^2{\mathbf{Z}}\right) =\left(\sum_{i=1}^{p}{\rm E}\left(Z_i^2 Z_1\right), \ldots,\sum_{i=1}^{p} {\rm E}\left(Z_i^2 Z_p\right)    \right)^\top.
\end{eqnarray}

All the quantities involved in (\ref{Moriii-skew}) are specific non-central moments of third order of ${\mathbf{Z}}$. When ${\mathbf{Y}}$ has a multivariate $\mathcal{MMN}$ distribution, ${\mathbf{Z}}$ is still $\mathcal{MMN}$ distribution, and so we can use once again the expressions in Theorem  \ref{ThmLem2}. Now,   ${\mathbf{Z}}={\mathbf{\Delta}}^{-1/2}({\mathbf{Y}}-{\boldsymbol{\mu}})=(Z_1, \ldots, Z_p)^\top$ has the distribution $\mathcal{MMN}_p({\boldsymbol{\xi}}_{\mathbf{Z}},{\mathbf{\Omega}}_{\mathbf{Z}},{\boldsymbol{\delta}}_{\mathbf{Z}}; H)$.

Furthermore, upon replacing ${\mathbf{A}}$ by  ${\mathbf{\Delta}}^{-1/2}$ in
the third central moment in  (\ref{noncentralzzEE}), using  (\ref{M3-M3bar})-(\ref{M3-AY}), and  the moments in  Theorem \ref{ThmLem2},  we can compute  $s({\mathbf{Y}})$ in (\ref{Moriii-skew}).

\subsection*{B5. Kollo  Measure of Skewness }

Kollo \cite{Kollo (2008)} noticed that M\'{o}ri-Rohatgi-Sz\'{e}kely skewness measure $s({\mathbf{Y}})$  does not include all third-order mixed moments. To include all mixed
moments of the third order, he defined a skewness vector of ${\mathbf{Y}}$ as
\begin{eqnarray}\label{Kollo-skew}
b({\mathbf{Y}})&=&{\rm E}\left(\sum_{i,j}^{p}(Z_i Z_j){\mathbf{Z}}\right)=\left(\sum_{i,j}^{p}{\rm E}\left[(Z_i Z_j) Z_1\right], \ldots,\sum_{i,j}^{p}{\rm E}\left[(Z_i Z_j) Z_p\right]\right)^\top.
\end{eqnarray}
The required moments can be obtained from  Theorem \ref{ThmLem2} and the corresponding measure in (\ref{Kollo-skew})
can then be computed.

\subsection*{B6. Balakrishnan-Brito-Quiroz  Measure of Skewness}\label{Sec-Bala-skew}

When reporting the skewness of a univariate distribution, it is customary to indicate
skewness direction by referring to skewness 'to the left' (negative) or 'to the right' (positive).
It seems natural that, in the multivariate setting, one would also like to indicate
a direction for the skewness of a distribution.

Both Mardia and  Malkovich-Afifi measures give an overall view of skewness  measures without any specific reference to the direction of skewness. For
this reason, \cite{Balakrishnan et al. (2007)} modified the Malkovich-Afifi measure to produce an overall vectorial measure of
skewness as
\begin{eqnarray}\label{T-index}
{\mathbf{T}}=\int_{{\phi}_p} {\mathbf{u}}c_1({\mathbf{u}})d \lambda({\mathbf{u}}),
\end{eqnarray}
where $c_1({\mathbf{u}})={\rm E}\left[\left({\mathbf{u}}^\top{\mathbf{Z}} \right)^3 \right] $ is a signed measure of skewness of the standardized variable
${\mathbf{Z}}={\mathbf{\Delta}}^{-1/2}({\mathbf{Y}}-{\boldsymbol{\mu}})$ in the direction
of  ${\mathbf{u}}$, and $\lambda$ denotes the rotationally invariant probability measure on the unit  $p$-dimensional sphere
${\phi}_p=\left\{{\mathbf{u}}\in \mathbb{R}^p; ||{\mathbf{u}}||=1  \right\}$.

From  \cite{Balakrishnan et al. (2007)} and \cite{Balakrishnan and Scarpa (2012)}, it turns out that the computation of ${\mathbf{T}}$ is straightforward and, when the distribution of ${\mathbf{Y}}$ is absolutely continuous with respect to Lebesgue measure and symmetric
(in the broad sense specified below), it has, under some moment assumptions, a
Gaussian asymptotic distribution with a limiting covariance matrix, ${\mathbf{\Sigma}}_T$, that can be
consistently estimated from the $Z_i$ sample.

If $c_1({\mathbf{u}})$ is negative, it indicates skewness in the direction
of $-{\mathbf{u}}$, while ${\mathbf{u}}c_1({\mathbf{u}})$ provides a vectorial index of skewness in the ${\mathbf{u}}$ (or $-{\mathbf{u}}$) direction. Summation of these vectors over ${\mathbf{u}}$ (in the form of an integral) will then yield an overall vectorial measure of skewness presented earlier in (\ref{T-index}).

For obtaining a single measure, \cite{Balakrishnan et al. (2007)} proposed
the quantity ${\mathbf{Q}}={\mathbf{T}}^\top {\mathbf{\Sigma}}_{\mathbf{T}}^{-1} {\mathbf{T}}$, where ${\mathbf{T}}$ is as in (\ref{T-index}) and ${\mathbf{\Sigma}}_{\mathbf{T}}$ is the covariance matrix of ${\mathbf{T}}$.
However, the covariance matrix  ${\mathbf{\Sigma}}_{\mathbf{T}}$  depends on the moments of sixth order. Sixth order moments in this family are not in explicit form, and so as done in \cite{Balakrishnan and Scarpa (2012)}, by replacing ${\mathbf{\Sigma}}_{\mathbf{T}}$ by ${\mathbf{\Sigma}}_{\mathbf{Z}}$, we obtain ${\mathbf{Q}}^*={\mathbf{T}}^\top {\mathbf{\Sigma}}_{\mathbf{Z}}^{-1} {\mathbf{T}}$, to provide a
reasonable measure of overall  skewness.

In the following, evaluation of  ${\mathbf{T}}$  using the integrals of some monomials over the unit sphere ${\phi}_p$ are required. From \cite{Balakrishnan et al. (2007)}, let  $u_j$ be the $j$-th coordinate of a point ${\mathbf{u}}\in{\phi}_p$. Then, the  values of the
integrals
\begin{eqnarray*}
J_4=\int_{{\phi}_p} u_j^4 d\lambda({\mathbf{u}})=\frac{3}{p(p+2)},~~~~~~~J_{2,2}=\int_{{\phi}_p} u_j^2 u_i^2 d\lambda({\mathbf{u}})=\frac{1}{p(p+2)},
\end{eqnarray*}
for $j \neq i, 1 \leq j, i \leq p$, are obtained using Theorem 3.3 of \cite{Fang et al. (1990)}. We see that the  above
integrals do not depend on the particular choices of $j$ and $ i$. Therefore, the $r$-th
coordinate of ${\mathbf{T}}$  is simply
$
{\mathbf{T}}_r=J_4 {\rm E}\left(Z_r^3\right)+ 3 \sum_{i\neq r} J_{2,2} {\rm E}\left(Z_i^2 Z_r\right).
$ 
So, we must obtain the moments $ {\rm E}\left(Z_r^3\right)$ and ${\rm E}\left(Z_i^2 Z_r\right)$. The required moments can be obtained as
$
{\rm E}\left(Z_i^3\right)={\mathbf{M}}_3^{\mathbf{Z}}[(i - 1)p + i, i] \mbox{~ and  ~~}
{\rm E}\left(Z_i^2 Z_j\right)={\mathbf{M}}_3({\mathbf{Z}})[(i - 1)p + i, j],
$ 
where  ${\mathbf{M}}_3({\mathbf{Z}})[., .]$ denotes the elements of  ${\mathbf{M}}_3({\mathbf{Z}})$, third moment of the     $\mathcal{MMN}_p({\boldsymbol{\xi}}_{\mathbf{Z}},{\mathbf{\Omega}}_{\mathbf{Z}},{\boldsymbol{\delta}}_{\mathbf{Z}}; H)$   distribution.
Upon using the above moments, we can obtain the elements of ${\mathbf{T}}$ as follows:
\begin{eqnarray}\label{Trrr}
{\mathbf{T}}_r=\frac{3}{p(p+2)} {\rm E}\left(Z_r^3\right)+ 3 \sum_{i\neq r} \frac{1}{p(p+2)} {\rm E}\left(Z_i^2 Z_r\right).
\end{eqnarray}

\subsection*{B7. Isogai Measure of Skewness}\label{Sec-Isogai-skew}

Isogai \cite{Isogai (1982)} considered an overall extension of Pearson measure of skewness to a
multivariate case in the form
\begin{eqnarray*}
S_I=\left({\boldsymbol{\mu}}-{\mathbf{M}}_0\right)^\top g^{-1}\left({\mathbf{\Delta}}\right)\left({\boldsymbol{\mu}}-{\mathbf{M}}_0\right),
\end{eqnarray*}
where ${\boldsymbol{\mu}}$, ${\mathbf{M}}_0$ and  ${\mathbf{\Delta}}$ are the mean, mode and  the covariance matrix of $\textbf{Y}$, respectively. The function  $g\left({\mathbf{\Delta}}\right)$ is an  ``appropriate''  function of the covariance matrix. To derive this measure of skewness, we need to obtain the mode of the $\mathcal{MMN}$ distribution, but the uniqueness
of the mode for the family of mean mixture of normal distributions is an open problem. For obtaining this measure, we choose $g(.)$ to be the identity function. This measure is location and scale invariant, and  so by using   the canonical form of  the  $\mathcal{MMN}$ distribution, we get
\begin{eqnarray*}
S_I=\frac{\left[{\mathbf{\delta}}_* {\rm E}(U)-m_0^* \right]^2}{1+{\mathbf{\delta}}_*^2[\textrm{var}(U)-1])},
\end{eqnarray*}
where ${\mathbf{\delta}}_*=\left({\boldsymbol{\delta}}^\top\overline{{\mathbf{\Omega}}}^{-1} {\boldsymbol{\delta}}\right)^{1/2}$ and $m_0^*$
is the mode of the single scalar $\mathcal{MMN}$ distribution in the canonical form. This index is essentially the Mahalanobis distance between the null vector and the vector ${\rm E}(\textbf{Y})-{\mathbf{M}}_0$, and it is   indeed  location and scale invariant.
Another vectorial measure has been given by \cite{Balakrishnan and Scarpa (2012)} as $S_C={\boldsymbol{\omega}}^{-1}\left({\boldsymbol{\mu}}-{\mathbf{M}}_0\right)$, which is a natural choice to characterize the direction of the asymmetry of the  multivariate $\mathcal{SN}$ distribution. Using the same reasoning for the $\mathcal{MMN}$ distribution, we can consider
$S_C=\left({\rm E}(U)-\frac{m_0^*}{{\mathbf{\delta}}_*}\right){\boldsymbol{\delta}}$,
and so, the direction of ${\boldsymbol{\delta}}$ can be regarded as a measure of vectorial skewness for the  $\mathcal{MMN}$ distribution.

\end{document}